\documentclass[a4paper,USenglish,cleveref,autoref,thm-restate]{lipics-v2021}

\title{Fair Termination of Asynchronous Binary Sessions}

\author{Luca Padovani}{University of Bologna}{}{https://orcid.org/0000-0001-9097-1297}{}
\author{Gianluigi Zavattaro}{University of Bologna and INRIA OLAS team}{}{https://orcid.org/0000-0003-3313-6409}{}
\authorrunning{L. Padovani and G. Zavattaro}
\Copyright{Luca Padovani and Gianluigi Zavattaro}

\category{} 
\relatedversion{A shorter version of this paper appears in the proceedings of ECOOP'25.}

\acknowledgements{We are grateful to the ECOOP reviewers for their thoughtful
feedback and questions. The authors have been partially supported by the ANR
project SmartCloud ANR-23-CE25-0012 and the PNRR project CN--HPC Spoke 9
Innovation Grant RTMER.}

\nolinenumbers

\EventEditors{John Q. Open and Joan R. Access}
\EventNoEds{2}
\EventLongTitle{42nd Conference on Very Important Topics (CVIT 2016)}
\EventShortTitle{CVIT 2016}
\EventAcronym{CVIT}
\EventYear{2016}
\EventDate{December 24--27, 2016}
\EventLocation{Little Whinging, United Kingdom}
\EventLogo{}
\SeriesVolume{42}
\ArticleNo{23}

\usepackage{xspace}
\usepackage{xcolor}
\usepackage{stmaryrd,cmll}
\usepackage{mathtools}
\usepackage{mathpartir}
\usepackage{prooftree}
\usepackage{tikz}
\usetikzlibrary{positioning}

\newif\ifcomments
\commentstrue
\commentsfalse

\newcommand{\marginnote}[2]{%
  \ifcomments%
    $^{\color{magenta}\mathclap\star}$%
    \marginpar[
        \flushright\tiny\sf\textbf{#1}: #2
    ]{
        \flushleft\tiny\sf\textbf{#1}: #2
    }%
  \fi%
}
\newcommand{\LP}[1]{\marginnote{Luca}{\color{blue}#1}}
\newcommand{\GZ}[1]{\marginnote{Gigio}{\color{purple}#1}}

\definecolor{mc}{rgb}{0.5,0,0}
\definecolor{mgreen}{rgb}{0,0.4,0.4}

\newcommand{\set}[1]{\braces{#1}}
\newcommand{\Nat}{\mathbb{N}}
\newcommand{\TagSet}{\mathsf{Tags}}

\newcommand{\rulename}[1]{\textnormal{\textup{\textsc{\small\bracks{#1}}}}}
\newcommand{\eoe}{\hfill$\lrcorner$}
\newcommand{\seqof}[1]{\overline{#1}}

\newcommand{\mkkeyword}[1]{\mathsf{\color{blue}#1}}

\newcommand{\Calculus}{\textsf{CaP}\xspace}
\newcommand{\CP}{\textsf{CP}\xspace}
\newcommand{\muCP}{$\mu$\CP}
\newcommand{\MALL}{\textsf{MALL}\xspace}
\newcommand{\muMALL}{$\mu$\MALL}

\newcommand{\infercorule}[3][]{%
  {\mprset{fraction={===}}\inferrule[#1]{#2}{#3}}%
}

\newcommand{\proofcase}[1]{\textbf{Case #1.}\xspace}

\newcommand{\ie}{\emph{i.e.}\xspace}
\newcommand{\cf}{\emph{cf.}\xspace}
\newcommand{\eg}{\emph{e.g.}\xspace}
\newcommand{\etal}{\emph{et al.}\xspace}

\newcommand{\parens}[1]{(#1)}
\newcommand{\bracks}[1]{[#1]}
\newcommand{\braces}[1]{\{#1\}}
\newcommand{\angles}[1]{\langle#1\rangle}

\newcommand{\ceils}[1]{\lceil#1\rceil}

\newcommand{\pol}{\dagger}
\newcommand{\inp}{{?}}
\newcommand{\out}{{!}}

\newcommand{\ann}[2][0]{\mafter[#1]{#2}}

\newcommand{\MessageType}{\MessageTypeS}
\newcommand{\MessageTypeS}{\sigma}
\newcommand{\MessageTypeT}{\tau}

\newcommand{\SessionTypeT}

\newcommand{\One}{\mathbf{1}}
\newcommand{\Zero}{\mathbf{0}}
\renewcommand{\Bot}{\bot}
\newcommand{\Top}{\top}
\newcommand{\End}{\mkkeyword{end}}

\newcommand{\choice}{\mathbin\oplus}

\newcommand{\Plus}[1]{{\oplus}\set{#1}}
\newcommand{\With}[1]{{\&}\set{#1}}
\newcommand{\DecWith}[2]{{\&}^{\color{red}#1}\set{#2}}
\newcommand{\Times}[2]{#1 \otimes #2}
\newcommand{\Par}[2]{#1 \parr #2}

\newcommand{\Tag}[1][a]{\mathsf{\color{teal}#1}}
\newcommand{\TagA}{\Tag[a]}
\newcommand{\TagB}{\Tag[b]}
\newcommand{\TagC}{\Tag[c]}
\newcommand{\TagD}{\Tag[d]}

\newcommand{\Action}{\ActionA}
\newcommand{\ActionA}{\alpha}
\newcommand{\ActionB}{\beta}

\newcommand{\Actions}{\ActionsA}
\newcommand{\ActionsA}{\varphi}
\newcommand{\ActionsB}{\psi}

\newcommand{\session}[2]{#1\mathbin\|#2}

\newcommand{\x}{x}
\newcommand{\y}{y}
\newcommand{\z}{z}

\newcommand{\Message}{\MessageM}
\newcommand{\MessageM}{M}

\newcommand{\Hole}{\bracks\,}

\newcommand{\Buffer}{\BufferB}
\newcommand{\BufferB}{\mathcal{B}}

\newcommand{\parop}{\mathbin|}

\newcommand{\Done}{\mkkeyword{done}}
\newcommand{\Link}[2]{#1 \leftrightarrow #2}
\newcommand{\Call}[2]{#1\ifblank{#2}{\angles{}}{\angles{#2}}}
\newcommand{\Close}[1]{\mkkeyword{close}~#1}
\newcommand{\Wait}[1]{\mkkeyword{wait}~#1}
\newcommand{\Select}[2]{#1\mathbin\vartriangleleft#2}
\newcommand{\OneCase}[2]{#1\mathbin\vartriangleright#2}
\newcommand{\Case}[2]{\OneCase{#1}{\set{#2}}}
\newcommand{\Fork}[4]{#1\parens{#2}\bracks{#3}\ifblank{#4}{}{.#4}}
\let\OldJoin\Join
\renewcommand{\Join}[2]{#1\parens{#2}}

\newcommand{\Cut}[3]{\parens{#1}(#2 \parop #3)}
\newcommand{\Cast}[1]{\ceils{#1}}

\newcommand{\Let}[2]{#1\ifblank{#2}{\parens{}}{\parens{#2}}}

\newcommand{\Measure}{\MeasureN}

\newcommand{\MeasureN}{n}

\newcommand{\mafter}[2][0]{#2\ifblank{#1}{}{^{{\color{mc}#1}}}}

\newcommand{\wtp}[3][n]{#2 \mafter[#1]\vdash #3}

\newcommand{\wtb}[5][n]{#2 \mafter[#1]\vdash #3; #4 : #5}

\newcommand{\EmptyContext}{\emptyset}
\newcommand{\Context}{\ContextC}
\newcommand{\ContextC}{\Gamma}
\newcommand{\ContextD}{\Delta}
\newcommand{\ContextE}{\Theta}

\newcommand{\DoneRule}{\rulename{done}\xspace}
\newcommand{\CallRule}{\rulename{call}\xspace}
\newcommand{\LinkRule}{\rulename{link}\xspace}
\newcommand{\CloseRule}{\rulename{$\One$}\xspace}
\newcommand{\WaitRule}{\rulename{$\Bot$}\xspace}
\newcommand{\ForkRule}{\rulename{$\Times{}{}$}\xspace}
\newcommand{\JoinRule}{\rulename{$\Par{}{}$}\xspace}

\newcommand{\SelectRule}{\rulename{$\oplus$}\xspace}

\newcommand{\CaseRule}{\rulename{$\&$}\xspace}
\newcommand{\ChoiceRule}{\rulename{choice}\xspace}
\newcommand{\CutRule}{\rulename{cut}\xspace}

\newcommand{\EmptyBufferRule}{\rulename{b-empty}\xspace}
\newcommand{\SelectBufferRule}{\rulename{b-select}\xspace}
\newcommand{\ForkBufferRule}{\rulename{b-fork}\xspace}

\newcommand{\tr}{\mathsf{tr}}

\newcommand{\inputs}{\mathsf{inp}}
\newcommand{\outputs}{\mathsf{out}}
\newcommand{\co}[1]{\overline{#1}}
\newcommand{\fn}[1]{\mathsf{fn}(#1)}

\newcommand{\dual}[1]{#1^\bot}
\newcommand{\subst}[2]{\set{#1/#2}}
\newcommand{\dom}[1]{\mathsf{dom}(#1)}
\newcommand{\depth}[1]{\mathsf{depth}(#1)}

\newcommand{\thread}[1]{\mathsf{thread}(#1)}

\newcommand{\rrel}{\mathcal{R}}
\newcommand{\srel}{\mathcal{S}}

\newcommand{\eqdef}{\stackrel{\smash{\textsf{\upshape\tiny def}}}=}

\newcommand{\correct}[2]{#1 \OldJoin #2}
\newcommand{\cc}[2]{\correct{#1}{#2}}

\newcommand{\subt}{\leqslant}

\newcommand{\ghsubt}{\subt_{\mathsf{s}}}
\newcommand{\mysubt}{\subt_{\mathsf{a}}}
\newcommand{\blzsubt}{\subt_{\mathsf{fa}}}

\newcommand{\auxsubt}{\subt_{\mathsf{aux}}}

\newcommand{\red}{\rightarrow}
\newcommand{\nred}{\arrownot\red}
\newcommand{\wred}{\red^*}

\newcommand{\lred}[1]{\stackrel{#1}{\longrightarrow}}
\newcommand{\lmust}[1]{\stackrel{#1}{\longmapsto}}
\newcommand{\xlmust}[1]{\xmapsto{#1}}
\newcommand{\lind}[1]{\stackrel{#1}{\Longmapsto}}
\newcommand{\nlmust}[1]{\longarrownot\lmust{#1}}
\newcommand{\nlind}[1]{\Longarrownot\lind{#1}}
\newcommand{\ired}[1]{\lred{#1}_{\textsf{ind}}}
\newcommand{\cored}[1]{\lred{#1}_{\textsf{coind}}}
\newcommand{\xlred}[1]{\xrightarrow{#1}}
\newcommand{\nlred}[1]{\longarrownot\lred{#1}}

\newcommand{\pcong}{\sqsupseteq}
\newcommand{\positive}{\mathsf{pos}}
\newcommand{\negative}{\mathsf{neg}}

\newcommand{\TagPlay}{\Tag[play]}
\newcommand{\TagQuit}{\Tag[quit]}
\newcommand{\TagWin}{\Tag[win]}
\newcommand{\TagLose}{\Tag[lose]}

\newcommand{\TagReq}{\Tag[req]}
\newcommand{\TagResp}{\Tag[resp]}
\newcommand{\TagStop}{\Tag[stop]}

\newcommand{\StreamServer}{\mathit{StreamS}}
\newcommand{\StreamClient}{\mathit{StreamC}}
\newcommand{\BatchServer}{\mathit{BatchS}}
\newcommand{\BatchServerAux}{\BatchServer'}
\newcommand{\BatchClient}{\mathit{BatchC}}
\newcommand{\BatchClientAux}{\BatchClient'}

\newcommand{\Client}{\mathit{Client}}

\newcommand{\Server}{\mathit{Server}}
\newcommand{\Worker}{\mathit{Worker}}

\newcommand{\TagTask}{\Tag[task]}
\newcommand{\TagRes}{\Tag[res]}
\newcommand{\TagCommand}{\Tag[cmd]}
\newcommand{\TagData}{\Tag[data]}

\newcommand\citep[1]{\cite{#1}}
\newcommand\citet[1]{\cite{#1}}

\begin{document}
\maketitle

\begin{abstract}
  We study a theory of asynchronous session types ensuring that well-typed
  processes terminate under a suitable fairness assumption. Fair termination
  entails \emph{starvation freedom} and \emph{orphan message freedom} namely
  that all messages, including those that are produced \emph{early} taking
  advantage of asynchrony, are eventually consumed. The theory is based on a
  novel \emph{fair asynchronous subtyping} relation for session types that is
  coarser than the existing ones. The type system is also the first of its kind
  that is firmly rooted in linear logic: fair asynchronous subtyping is
  incorporated as a natural generalization of the cut and axiom rules of linear
  logic and asynchronous communication is modeled through a suitable set of
  \emph{commuting conversions} and of \emph{deep cut reductions} in linear logic
  proofs.
\end{abstract}

\section{Introduction}
\label{sec:introduction}

Session type systems~\citep{Honda93,HondaVasconcelosKubo98,HuttelEtAl16} have
become a widespread formalism for ensuring a variety of safety and liveness
properties of communicating processes through static analysis. Session types
specify the \emph{sequences} of messages that can be sent over a channel, and
the type system makes sure that (1) well-typed processes comply with this
protocol specification and that (2) the peer endpoints of the channel are used
according to compatible protocols.

Many session type systems are defined for a synchronous calculus or language
even if the actual underlying communication model is meant to be asynchronous.
The point is that synchronous theories of session types are simpler and easier
to work with and most if not all properties of a synchronous session type system
hold even if the actual communication model is asynchronous. However, awareness
of the communication model can help relaxing the type system and thus enlarging
the family of well-typed processes. This observation has led to the study of
\emph{asynchronous subtyping relations} for session
types~\cite{MostrousYoshidaHonda09,MostrousYoshida15,ChenDezaniScalasYoshida17}
allowing processes to \emph{anticipate} output messages with respect to the
protocol specification they are expected to comply with, provided that
anticipated outputs \emph{do not depend} on incoming messages. This apparent
violation of the protocol specification enabled by asynchronous subtyping is
harmless precisely because output actions are non-blocking in an asynchronous
setting.

\newcommand\Split{\mathit{Split}}
\newcommand\Gather{\mathit{Gather}}

The available session type systems based on asynchronous subtyping focus on the
enforcement of safety properties but struggle at ensuring liveness properties of
many simple communication patterns. As an illustration of such patterns consider
a server that, in order to fulfill a request received from a session $x$, splits
the request into an arbitrary number of tasks handled by a separate worker and
then gathers the partial results to answer the client. We might be interested in
establishing whether the client eventually receives a response.

To be more concrete, let us model the server as the term
\[
    \Case\x{\TagReq : \Cut\y{\Call\Split{x,y}}{\Call\Worker\y}}
\]
which indicates the input of the request followed by the spawning of a $\Split$
process and of a $\Worker$ process connected by a new session $y$.

The $\Split$ process is modeled by the following definitions:
\begin{align}
    \label{eq:split}
    \Let\Split{x,y} & = \Select\y\TagTask.\Call\Split{x,y} \choice \Select\y\TagStop.\Call\Gather{x,y} \\
    \label{eq:gather}
    \Let\Gather{x,y} & = \Case\y{\TagRes : \Call\Gather{x,y}, \TagStop : \Wait\y.\Select\x\TagResp.\Close\x }
\end{align}
according to which the server sends to the worker a non-deterministically chosen
number of $\TagTask$s followed by a $\TagStop$ on session $y$, it then gathers
from the worker an arbitrary number of $\TagRes$ults followed by a $\TagStop$,
and finally sends back to the client a $\TagResp$onse on session $x$.

One possible modeling of the worker process is according to the definition
\begin{equation}
    \label{eq:worker}
    \Let\Worker\y = \Case\y{
        \TagTask : \Select\y\TagRes.\Call\Worker\y,
         \TagStop : \Select\y\TagStop.\Close\y
    }
\end{equation}
so that the worker sends a $\TagRes$ult for each $\TagTask$ it receives.
The eye-catching aspect of $\Worker$ is that it does not interact according to
the ``complementary'' protocol implemented on the server side, at least not in
the sense that is usually intended in (synchronous) session type theories.
Indeed, while the server first sends all the $\TagTask$s \emph{and then} gathers
all the $\TagRes$ults, the worker eagerly sends one $\TagRes$ult after receiving
each $\TagTask$.

In an asynchronous setting, the fact that $\Worker$ sends some messages earlier
than expected is not an issue since output actions are non-blocking. However, we
would like to be sure that these early $\TagRes$ults do not keep accumulating
and are eventually consumed by $\Gather$. There is nothing in the modeling of
$\Split$, $\Gather$ and $\Worker$ that prevents this from happening, but
\emph{we can prove the eventual consumption of every $\TagRes$ult only under the
assumption that sooner or later $\Split$ will send $\TagStop$ to $\Worker$}.
If we broaden our viewpoint, we see that the same assumption is necessary to
prove that the client interacting with the server does not starve. While the
server is running, the client is awaiting for a $\TagResp$onse from session $x$.
In order to \emph{prove} that the client will eventually receive a
$\TagResp$onse, we have to assume that $\Worker$ will terminate the session $y$,
which in turn requires the assumption that $\Split$ will eventually send
$\TagStop$ to $\Worker$.
In general, the proof of any non-trivial liveness property that concerns the
eventual production or consumption of a message on a certain session may require
the assumption that every other session eventually terminates. For this reason,
ensuring the eventual termination of sessions should be a primary goal of any
session type system aimed at enforcing liveness properties.
As we have seen in the discussion above, proving the eventual termination of a
session may require some \emph{fairness assumptions} like the fact that $\Split$
will eventually stop sending $\TagTask$s. For this reason, in the literature
such eventual termination property is referred to as \emph{fair
termination}~\citep{GrumbergFrancezKatz84,Francez86,AptFrancezKatz87}.

In this work we study a theory of asynchronous session types and an associated
session type system ensuring that well-typed processes are fairly terminating
under a suitable fairness assumption. The fair termination of processes entails
the fair termination of the sessions they operate on, hence that all messages
produced -- including those sent earlier than expected -- are eventually
consumed and that all processes waiting for a message eventually receive one. In
other words, the very same type system prevents \emph{process starvation} and
ensures the absence of \emph{orphan messages}.
Related works achieve these goals only partially:
\begin{itemize}
\item The theories of asynchronous session types developed by
    Mostrous~\etal~\citet{MostrousYoshidaHonda09,MostrousYoshida15}, Chen
    \etal~\citet{ChenDezaniScalasYoshida17}
	and Ghilezan \etal~\citet{GhilezanEtAl23} require that every output can be only
    anticipated with respect to finitely many inputs.
    Since there is no finite upper bound to the number of
    $\TagTask$s that $\Split$ can produce, each early $\TagRes$ produced by
    $\Worker$ anticipates an unbounded number of inputs leaving $\Worker$ out of
    reach for the aforementioned theories.
\item Bravetti, Lange and Zavattaro~\citep{BravettiLangeZavattaro24} study an
    asynchronous subtyping relation for session types that allows early outputs
    to anticipate an unbounded number of inputs. However, their subtyping
    relation disallows any \emph{covariance} of outputs, which is needed to
    account for the fact that the behavior of $\Worker$ is \emph{more
    deterministic} than the behavior expected by $\Gather$. Indeed, while
    $\Gather$ expects to receive an arbitrary number of $\TagRes$ults, $\Worker$
    produces exactly as many $\TagRes$ults as the number of $\TagTask$s it
    receives from $\Split$.
    %
\item None of the aforementioned works provides guarantees on the eventual
    fulfillment of the client's request, either because they do not make sure
    that every session eventually
    terminates~\cite{ChenDezaniScalasYoshida17,GhilezanEtAl23} or because they
    do not take multiple sessions into
    account~\cite{GhilezanEtAl23,BravettiLangeZavattaro24}.
\item
    Ciccone, Dagnino and
    Padovani~\citep{CicconePadovani22A,CicconePadovani22C,CicconeDagninoPadovani24}
    study session type systems ensuring the fair termination of sessions, but
    their works are based on synchronous communication models that do not
    support any form of output anticipation like the one exemplified by
    $\Worker$.
\end{itemize}

The theory of asynchronous session types that we propose in this paper addresses
all these limitations. In addition, we also make the following technical
contributions:
\begin{itemize}
\item We define a \emph{fair asynchronous subtyping relation} for session types
    that is coarser than those in the literature for the family of eventually
    terminating session types. Unlike many existing fair/asynchronous subtyping
    relations~\citep{MostrousYoshida15,CicconePadovani22A,CicconeDagninoPadovani24,BravettiLangeZavattaro24}
    our subtyping relation is \emph{closed under duality}. This property is key
    for proving the type system sound.
\item We give the first \emph{fair asynchronous semantics} of session types
    using a labelled transition system (LTS) defined by \emph{bounded
    coinduction}~\citep{AnconaDagninoZucca17,Dagnino19}. The adoption of this
    semantics allows us to characterize fair asynchronous subtyping in a way
    that is structurally the same as the one for the well-known synchronous
    subtyping defined by Gay and Hole~\citep{GayHole05}.
\item Our theory of asynchronous session types with fair asynchronous subtyping
    is the first one where the process model and the type system are rooted in
    linear logic~\cite{Wadler14,CairesPfenningToninho16,LindleyMorris16}. We
    incorporate fair asynchronous subtyping in the type system as generalized
    forms of the cut and linear logic axiom thanks to the aforementioned closure
    under duality. Also, instead of introducing explicit message buffers, we
    model asynchronous communications by means of suitable \emph{commuting
    conversions} and \emph{deep cut reductions} in linear logic proofs.
\end{itemize}

\subparagraph*{Structure of the paper.}
\Cref{sec:language} describes our calculus of asynchronous processes and the
properties we enforce. \Cref{sec:session-types} introduces the semantics of
asynchronous session types while \Cref{sec:subtyping} studies fair asynchronous
subtyping and its properties. \Cref{sec:type-system} describes the type system
and \Cref{sec:subtyping-usage-patterns} illustrates some common usage patters of
fair asynchronous subtyping.
\Cref{sec:subtyping-inclusions} shows how our LTS can be easily tailored to
characterize various asynchronous subtyping relations that appear in the
literature and compares them with our subtyping relation.
\Cref{sec:related} discusses related work in more detail and
\Cref{sec:conclusion} recaps and outlines future work.
Proofs and additional technical material are provided in the Appendix.
\section{A Calculus of Asynchronous Processes}
\label{sec:language}

\begin{table}
    \caption{Syntax of \Calculus.}
    \label{tab:syntax}
    \centering
    \begin{math}
        \begin{array}[t]{@{}rcll@{}}
            P, Q & ::= & & \textbf{Process} \\
            &   & \Link\x\y & \text{link} \\
            & | & \Close\x & \text{signal output} \\
            & | & \Select\x\Tag.P & \text{tag output} \\
            & | & \Fork\x\y{P}{Q} & \text{channel output} \\
            & | & P \choice Q & \text{choice} \\
        \end{array}
        ~
        \begin{array}[t]{@{}rcll@{}}
            & | & \Done & \text{termination} \\
            & | & \Call{A}{\seqof\x} & \text{invocation} \\
            & | & \Wait\x.P & \text{signal input} \\
            & | & \Case\x{\Tag_i:P_i}_{i\in I} & \text{tag input} \\
            & | & \Join\x\y.P & \text{channel input} \\
            & | & \Cut\x{P}{Q} & \text{composition} \\
        \end{array}
    \end{math}
\end{table}

In this section we present syntax and semantics of a calculus of asynchronous
processes called \Calculus and we formulate the safety and liveness properties
that our type system guarantees.
At the surface level, \Calculus closely resembles other calculi of binary
sessions based on linear logic such as \CP~\cite{Wadler14} or
\muCP~\cite{LindleyMorris16}. The main differences between \Calculus and these
calculi are that \Calculus supports general recursion, it includes a
non-deterministic choice operator and above all it models asynchronous
communication by giving a non-blocking semantics to output actions, which
essentially act like message buffers.

The syntax of \emph{processes} is shown in \Cref{tab:syntax} and makes use of an
infinite set of \emph{channels} ranged over by $x$, $y$ and $z$, a set $\TagSet$
of \emph{tags} ranged over by $\TagA$, $\TagB$, $\dots$ and a set of
\emph{process names} ranged over by $A$, $B$, etc.
The terminated process, which performs no actions, is denoted by $\Done$.
The term $\Call{A}{\seqof\x}$, where $\seqof\x$ is a possibly empty sequence of
channels, represents the \emph{invocation} of the process named $A$. We assume
that for each such invocation there exists a unique global definition of the
form $\Let{A}{\seqof\x} = P$ that gives meaning to the name $A$. We also assume
that all invocations occur guarded by a prefix or by a non-deterministic choice.
The term $\Link\x\y$ models a \emph{link}, that is a process that forwards every
message received from $x$ to $y$ and vice versa, effectively unifying the two
channels. Links are typical of calculi based on linear logic since their typing
rule correspond to the axiom of the logic.
The terms $\Close\x$ and $\Wait\x.P$ model processes that respectively
send and receive a termination signal on $x$. Note that $\Close\x$ has no
continuation, as is the case in most calculi based on linear logic, whereas
$\Wait\x.P$ continues as $P$.
The term $\Fork\x\y{P}{Q}$ 
models a bifurcating session: 
it creates a new channel $y$, sends $y$ on $x$, forks a
new process $P$ that uses $y$, and then continues as $Q$. The process
$\Join\x\y.P$ waits for a channel $y$ from $x$ and then continues as $P$.
The terms $\Select\x\Tag.P$ and $\Case\x{\Tag_i:P_i}_{i\in I}$ model processes
that respectively send and receive a tag. The sender selects one particular tag $\Tag$
to send. The receiver continues as $P_i$ depending on the tag $\Tag_i$ that it
receives. In the examples we sometimes use tags as an abstract representation of
more complex messages, such as $\TagReq$uests or $\TagTask$s.
The term $\Cut\x{P}{Q}$ models the parallel composition of two processes $P$ and
$Q$ connected by the channel $x$. We often refer to this term as a \emph{cut},
since its typing rule coincides with the cut rule of linear logic.
Finally, the term $P \choice Q$ models the non-deterministic choice between $P$
and $Q$.

It is known that links in conjunction with bifurcating sessions can be used to
encode \emph{session delegation}, whereby processes exchange an \emph{existing}
(rather than new) channel $z$ on $x$. This behavior can be modeled by a term of
the form $\Fork\x\y{\Link\y\z}{P}$.
As we will see in \Cref{sec:subtyping-usage-patterns}, \Calculus links also act
as \emph{explicit casts} enabling useful forms of subsumption.

The notions of \emph{free} and \emph{bound} channels for processes are defined
as expected with the proviso that a process of the form $\Fork\x\y{P}{Q}$ binds
$y$ in $P$ but not in $Q$.
We identify processes up to renaming of bound channels and we write $\fn{P}$ for
the set of channels occurring free in $P$.
Also, for every global definition $\Let{A}{\seqof\x}={P}$ we assume $\fn{P} =
\set{\seqof\x}$.

\begin{table*}
    \caption{Structural pre-congruence and reduction semantics of \Calculus.}
    \label{tab:semantics}
    \centering
    \begin{math}
        \displaystyle
        \begin{array}{@{}lr@{~}c@{~}ll@{}}
            \rulename{s-call}
            & \Call{A}{\seqof\x}
            & \pcong & P
            & \Let{A}{\seqof\x}={P}
            \\
            \rulename{s-link}
            & \Link\x\y
            & \pcong & \Link\y\x
            \\
            \rulename{s-comm}
            & \Cut\x{P}{Q}
            & \pcong & \Cut\x{Q}{P}
            \\
            \rulename{s-wait}
            & \Cut\x{\Buffer_x[\Wait\y.P]}{Q}
            & \pcong & \Wait\y.\Cut\x{\Buffer_x[P]}{Q}
            & x \ne y
            \\
            \rulename{s-case}
            & \Cut\x{\Buffer_x[\Case\y{\Tag_i:P_i}_{i\in I}]}{Q}
            & \pcong & \Case\y{\Tag_i : \Cut\x{\Buffer_x[P_i]}{Q}}_{i\in I}
            & x \ne y
            \\
            \rulename{s-join}
            & \Cut\x{\Buffer_x[\Join\y\z.P]}{Q}
            & \pcong & \Join\y\z.\Cut\x{\Buffer_x[P]}{Q}
            & x \ne y, z \not\in \fn{\Buffer_x}
            \\
            \rulename{s-pull-0}
            & \Cut\x{\Buffer_x[P]}{\Buffer_x'[\Link\x\y]}
            & \pcong & \Buffer_y[\Cut\x{P}{\Buffer_x'[\Link\x\y]}]
            &
            \\
            \rulename{s-pull-1}
            & \Cut\x{\Fork\y\z{P}{Q}}{R}
            & \pcong & \Fork\y\z{\Cut\x{P}{R}}{Q}
            & x \ne y, x\in\fn{P}
            \\
            \rulename{s-pull-2}
            & \Cut\x{\Buffer_y[P]}{Q}
            & \pcong & \Buffer_y[\Cut\x{P}{Q}]
            & x \ne y, x \not\in \fn{\Buffer_y}
            \\
            \rulename{s-pull-3}
            & \Buffer_x[\Fork\y\z{P}{Q}]
            & \pcong & \Fork\y\z{\Buffer_x[P]}{Q}
            & x \ne y, x\in\fn{P}
            \\
            \rulename{s-pull-4}
            & \Buffer_x[\Buffer_y'[P]]
            & \pcong & \Buffer_y'[\Buffer_x[P]]
            & x \ne y, x \in \fn{P}
            \\
            \\
            \rulename{r-choice}
            & P_1 \choice P_2
            & \red & P_k
            & k \in \{1,2\}
            \\
            \rulename{r-link}
            & \Cut\x{\Link\x\y}{P}
            & \red & P\subst\y\x
            \\
            \rulename{r-close}
            & \Cut\x{\Close\x}{\Wait\x.P}
            & \red & P
            \\
            \rulename{r-select}
            & \Cut\x{\Select\x{\Tag_k}.P}{\Buffer_x[\Case\x{\Tag_i:Q_i}_{i\in I}]}
            & \red & \Cut\x{P}{\Buffer_x[Q_k]}
            & k \in I
            \\
            \rulename{r-fork}
            & \Cut\x{\Fork\x\y{P}{Q}}{\Buffer_x[\Join\x\y.R]}
            & \red & \Cut\y{P}{\Cut\x{Q}{\Buffer_x[R]}}
            &
            \\
            \rulename{r-cut}
            & \Cut\x{P}{R}
            & \red & \Cut\x{Q}{R}
            & P \red Q
            \\
            \rulename{r-buffer}
            & \Buffer_x[P]
            & \red & \Buffer_x[Q]
            & P \red Q
            \\
            \rulename{r-str}
            & P & \red & Q
            & P \pcong R \red Q
            \\
        \end{array}
    \end{math}
\end{table*}

We have anticipated that \Calculus adopts an asynchronous communication model.
In practice this would be implemented by FIFO buffers storing messages that
have been produced but not consumed. In \Calculus, we model asynchrony giving a
non-blocking semantics to the output actions $\Fork\x\y{Q}{P}$ and
$\Select\x\Tag.P$. That is, we allow the continuation $P$ to reduce and/or
interact with other sub-processes even if the prefix has not been consumed. In a
sense, we consider the prefixes $\Fork\x\y{Q}{}$ and $\Select\x\Tag$ as
\emph{floating messages} or parts of a \emph{buffer} associated with channel
$x$.
In general, we call \emph{buffer} any term generated by the following grammar:
\[
    \textbf{Buffer}
    \qquad
    \Buffer_x ::= \Hole \mid \Select\x\Tag.\Buffer_x \mid \Fork\x\y{P}{\Buffer_x}
\]

A buffer is either empty, represented by a \emph{hole} $\Hole$, or a tag $\Tag$
sent on $x$ followed by a buffer for $x$, or a fresh channel $y$ (with
associated process $P$) sent on $x$ and followed by a buffer for $x$. Note that
the annotation $x$ in the metavariable $\Buffer_x$ is meant to \emph{bind} the
channel on which the messages in the buffer have been sent. Therefore, having at
our disposal a buffer $\Buffer_x$, we can write $\Buffer_y$ for the buffer that
has the same structure as $\Buffer_x$ but where $x$ has been replaced by $y$.
Buffers vaguely resemble \emph{reduction contexts}, except that they allow us to
place a hole \emph{behind} output prefixes since these are meant to be
non-blocking. Hereafter we write $\Buffer_x[P]$ for the process obtained by
replacing the hole in $\Buffer_x$ with $P$.

The operational semantics of \Calculus is given by a structural pre-congruence
relation $\pcong$ and a reduction relation $\red$, both defined in
\Cref{tab:semantics}.
In simple words, structural pre-congruence relates processes that are
essentially equivalent except for the order of independent actions, while
reduction describes communications and the resolution of non-deterministic
choices.

We can roughly classify the rules for structural pre-congruence in four groups.
The first group contains \rulename{s-call}, \rulename{s-link} and
\rulename{s-comm}, which capture expected properties of process invocations,
links and parallel compositions: a process invocation $\Call{A}{\seqof\x}$ is
indistinguishable from $P$ if $\Let{A}{\seqof\x}={P}$; a link $\Link\x\y$ is
indistinguishable from $\Link\y\x$; a cut $\Cut\x{P}{Q}$ is indistinguishable
from $\Cut\x{Q}{P}$, that is parallel composition is commutative.

The second group of rules contains \rulename{s-wait}, \rulename{s-case} and
\rulename{s-join}. These rules allow an input action on some channel $y$ to be
extruded from a cut on $x$ when $x \ne y$. The purpose of these transformations
is to move inputs on $y$ close to outputs on $y$, so as to enable interactions
in the session $y$. These transformations correspond to well-known
rearrangements of linear logic proofs (for example, they are sometimes referred
to as \emph{external reductions}~\citep{BaeldeDoumaneSaurin16,Doumane17}),
except that in \Calculus we allow the input prefix to be found \emph{within an
arbitrary buffer} $\Buffer_x$, coherently with the intuition that (output)
actions in a buffer are non-blocking. The side condition
$z\not\in\fn{\Buffer_x}$ in \rulename{s-join} makes sure that free occurrences
of $z$ in $\Buffer_x$ are not accidentally captured by the binding prefix. Since
we consider processes (and buffers) modulo renaming of bound names and there is
an infinite supply of names, the bound $z$ can always be renamed so as to enable
this process transformation.

The third group of rules contains \rulename{s-pull-0}, \rulename{s-pull-1} and
\rulename{s-pull-2}. These rules are similar to those of the second group,
except that they allow whole \emph{buffers} of messages on $y$ to float through
cuts on $x$ when $x \ne y$. \rulename{s-pull-2} is simple and speaks for itself
while \rulename{s-pull-1} deals with the case in which the name $x$ bound by the
cut occurs in the process $P$ associated with a fresh channel $z$ sent on $y \ne
x$. In this case the cut as a whole becomes associated with $z$. In principle we
should also specify the side condition $x\not\in\fn{Q}$, but this condition
holds for well-typed processes when $x\in\fn{P}$.
The rule \rulename{s-pull-0} covers a somewhat
peculiar case of \rulename{s-pull-2}, whereby a buffer $\Buffer_x$ can be
extruded from the cut on $x$ because there is a link $\Link\x\y$ that acts as a
forwarder from $x$ to $y$. Note that the extruded buffer $\Buffer_x$ turns into
$\Buffer_y$, so as to reflect the forwarding effect of the link.

The fourth and final group of rules contains \rulename{s-pull-3} and
\rulename{s-pull-4}. These rules allow buffers for different channels to be
permuted. Again \rulename{s-pull-3} deals with the special case in which a
channel $x$ occurs in the process $P$ associated with a fresh channel $z$.

In the definition of structural pre-congruence there are some glaring omissions
(\eg associativity of parallel composition) and very few rules are invertible.
This is not because the missing rules would be unsound, but because they turn
out to be unnecessary for proving that well-typed processes are deadlock free
and fairly terminating.

Base reductions consist of \rulename{r-choice}, \rulename{r-close},
\rulename{r-select} and \rulename{r-fork} which respectively model the reduction
of a non-deterministic process, the termination of the session $x$ and the
consumption of tags and channels sent on channel $x$.
The rules are almost standard, except that \rulename{r-select} and
\rulename{r-fork} allow input actions on $x$ to operate from within arbitrary
buffers for $x$. The buffers represent asynchronously sent messages that do not
block subsequent actions. At the logical level, these interactions correspond to
\emph{deep cut reductions} in a sense that resembles \emph{deep
inference}~\cite{GuglielmiStrassburger01}, whereby logical rules can be applied
deep within a context.
Unlike \rulename{r-select} and \rulename{r-fork}, there is no buffer around
$\Wait\x.P$ in \rulename{r-close}. This rule implies that a session cannot be
closed unless all the messages (asynchronously) produced therein have also been
consumed.
Note that this property is enforced
by the type system and is not meant to be checked at runtime.
%
Rules \rulename{r-cut}, \rulename{r-buffer} and \rulename{r-str} propagate
reductions across cuts and buffers and close them by structural pre-congruence.

Hereafter we write $\wred$ for the reflexive, transitive closure of $\red$. We
write $P \red$ if $P \red Q$ for some $Q$ and $P \nred$ if not $P \red$.

We can now formally define the safety and liveness properties we are interested
in.

\begin{definition}[deadlock freedom]
    \label{def:df}
    We say that $P$ is \emph{deadlock free} if for every $Q$ such that $P \wred Q
    \nred$ we have $Q \pcong \Done$.
\end{definition}

\newcommand{\Stuck}{\mathit{Stuck}}

Deadlock freedom is an instance of safety property. A deadlock-free process
either reduces or it is (structurally pre-congruent to) $\Done$. For example,
the process $\Let\Omega{}={\Call\Omega{} \choice \Call\Omega{}}$ is deadlock
free (we have $\Call\Omega{} \red \Call\Omega{}$) whereas
$\Cut\x{\Close\y}{\Wait\x.\Done}$ is deadlocked.
When a deadlock-free process stops reducing, it contains no pending actions and
all of its sessions have been closed. In particular, all the messages in buffers
have been consumed.

The liveness properties we are interested in are related to \emph{termination},
of which there exist several variants.
A \emph{reduction sequence} of $P$ is a sequence $(P_0,P_1,\dots)$ such that
$P_0 = P$ and $P_i \red P_{i+1}$ whenever $i+1$ is not greater than the length
of the sequence. A \emph{run} is a maximal reduction sequence, in the sense that
either it is infinite or the last process in the sequence (say $P_n$) does not
reduce (that is, $P_n \nred$).
We say that $P$ is \emph{weakly terminating} if it has a finite run, that $P$ is
\emph{terminating} if every run of $P$ is finite, and that $P$ is
\emph{diverging} if every run of $P$ is infinite.
For example, $\Call\Omega{} \choice \Done$ is weakly terminating but not
terminating, $\Done \choice \Done$ is terminating, and $\Call\Omega{}$ is
diverging.
Note that here we call ``termination'' the mere inability to reduce further and
not the fact that a process has become $\Done$. For example, $\Done$, $\Close\x$
and $\Cut\x{\Close\y}{\Wait\x.\Done}$ are all terminated (they do not reduce),
but only $\Done$ is also deadlock free. So it really is the combination of
deadlock freedom (\Cref{def:df}) and some termination property that we wish to
enforce with our type system.

The termination property we target in this work is called \emph{fair
termination}~\cite{GrumbergFrancezKatz84,Francez86,AptFrancezKatz87}. Fair
termination consists of those processes such that all of their infinite runs are
considered to be unrealistic or \emph{unfair} and therefore can be ignored
insofar termination is concerned. We could say that these processes may diverge
in principle, but they terminate in practice.
Clearly, fair termination depends on a \emph{fairness notion} that discriminates
\emph{fair} runs from \emph{unfair} ones. Among all fairness notions, here we
consider a particular instance of \emph{full fairness}~\citep{GlabbeekHofner19}.

\begin{definition}
    \label{def:fair-run}
    A run is \emph{fair} if it contains finitely many weakly terminating
    processes.
\end{definition}

Remember that a run is a \emph{maximal} reduction sequence of a process. So,
along a fair run the process only has finitely many chances to terminate. This
can happen either because the run is finite (the process eventually terminates)
or because the run contains a diverging process (at some point termination is no
longer possible). For example, the infinite run $(\Call\Omega{} \choice \Done,
\Call\Omega{}, \dots)$ is fair because only the first process in it is weakly
terminating.
We find it useful to also look at the \emph{negation} of the notion of fair run:
an \emph{unfair run} is necessarily infinite and contains infinitely many weakly
terminating processes. In other words, an unfair run describes a computation
along which termination is always reachable, but it is never reached as if the
process is avoiding it on purpose.
For example, if $\Let{A}{} = \Call{A}{} \choice \Done$ then the infinite run
$(\Call{A}{}, \Call{A}{}, \dots)$ is unfair because $\Done$ is always reachable
but never reached.

\begin{definition}
    \label{def:ft}
    We say that $P$ is \emph{fairly} \emph{terminating} if every \emph{fair} run
    of $P$ is finite.
\end{definition}

For example, $\Let{A}{} = \Call{A}{} \choice \Done$ is fairly terminating
whereas $\Call\Omega{} \choice \Done$ is not because it has an infinite fair
run.
There are two reasons why full fairness is a suitable fairness assumption in our
setting. First, full fairness has been shown to be the \emph{strongest}
conceivable fairness assumption~\citep{GlabbeekHofner19}, which means that it
allows us to target the \emph{largest} family of fairly terminating processes.
Second, it has been observed \citep{CicconeDagninoPadovani22} that this family
admits the following alternative characterization which does not mention fair
runs at all.

\begin{restatable}{theorem}{thmproofprinciple}
    \label{thm:proof-principle}
    $P$ is fairly terminating iff each $Q$ such that $P \wred Q$ is weakly
    terminating.
\end{restatable}

The relevance of this characterization rests in the fact that it provides the
key proof method for the soundness of our type system. Indeed, suppose that the
type system ensures that well-typed processes \emph{weakly} terminate. We expect
the type system to also enjoy subject redution, namely the property that
well-typed processes always reduce to well-typed processes. But then, using the
right-to-left implication in \Cref{thm:proof-principle}, the very same type
system also ensures that well-typed processes fairly terminate.

\begin{example}
    \label{ex:server-worker}
    Consider the $\Split$ process defined in \eqref{eq:split}. We derive
    \[
        \Call\Split{x,y} \red
        \Select\y\TagTask.\Call\Split{x,y} \red
        \cdots
        \red
        (\Select\y\TagTask)^n.\Call\Split{x,y}
    \]
    where $(\Select\y\TagTask)^n$ denotes $n$ subsequent $\Select\y\TagTask$
    prefixes, using repeated applications of \rulename{r-choice} and
    \rulename{r-buffer}. Notice how the messages pile up as the process reduces
    and also that, at any time, the process may reduce to
    $(\Select\y\TagTask)^n.\Select\y\TagStop.\Call\Gather{x,y}$ which is
    weakly terminating.
    That is, $\Call\Split{x,y}$ is fairly terminating by
    \Cref{thm:proof-principle}.
    We also have
    \[
        \Cut\y{\Call\Split{x,y}}{\Call\Worker\y} \wred
        \Cut\y{(\Select\y\TagTask)^n.\Call\Split{x,y}}{(\Select\y\TagRes)^m.\Call\Worker\y}
    \]
    for every $n$ and $m$, indicating that $\Split$ has produced $m+n$ tasks and
    $\Worker$ has consumed only $m$ of them.
    Since $\Split$ can always send $\TagStop$ and consume all the $\TagRes$
    messages produced by $\Worker$, we also have
    \[
        \Cut\y{(\Select\y\TagTask)^n.\Call\Split{x,y}}{(\Select\y\TagRes)^m.\Call\Worker\y}
        \wred
        \Select\x\TagResp.\Done
        \nred
    \]
    indicating that $\Cut\y{\Call\Split{x,y}}{\Call\Worker\y}$ too is fairly
    terminating by \Cref{thm:proof-principle}.
    \eoe
\end{example}

\section{Asynchronous Session Types}
\label{sec:session-types}

\subsection{Syntax}
\label{sec:types-syntax}

\emph{Session types} are generated by the productions below
\[
    \textstyle
    \textbf{Session type}
    \qquad
    S, T ::= \One \mid \Bot \mid \Plus{\Tag_i : S_i}_{i\in I} \mid \With{\Tag_i : S_i}_{i\in I} \mid \Times{S}{T} \mid \Par{S}{T}
\]
and adhere to the usual interpretation given to propositions of multiplicative
additive linear logic (\MALL)~\cite{CairesPfenningToninho16,Wadler14}: the
constants $\One$ and $\Bot$ describe processes respectively sending and
receiving a termination signal; the additive connectives $\Plus{\Tag_i :
S_i}_{i\in I}$ and $\With{\Tag_i : S_i}_{i\in I}$ describe processes sending and
receiving a tag $\Tag_i$ and then behaving according to $S_i$; the
multiplicative connectives $\Times{S}{T}$ and $\Par{S}{T}$ describe processes
exchanging a channel of type $S$ and then behaving according to $T$.

Compared to the usual linear logic propositions, we observe the following
differences:
\begin{itemize}
\item The additive connectives are $n$-ary instead of binary and make use of
    explicit tags for improved generality and readability. In each additive
    connective $\Plus{\Tag_i : S_i}_{i\in I}$ and $\With{\Tag_i : S_i}_{i\in I}$
    we assume that the set $I$ is finite and that the tags $\Tag_i$ are pairwise
    disjoint.
\item The additive constants $\Zero \eqdef \Plus{}$ and $\Top \eqdef \With{}$
    are defined as degenerate (empty) versions of the additive connectives
    instead of being built-in.
\item The productions shown above are meant to be interpreted
    \emph{coinductively}. That is, we consider session types the possibly
    infinite \emph{regular} trees built using the above productions. We define
    possibly infinite session types as solutions of equations of the form $S =
    \cdots$ where the metavariable $S$ may occur (guarded) on the right hand
    side of `$=$'. It is a known fact that every such finite system of equations
    admits a unique regular solution~\cite{Courcelle83}.
\end{itemize}

The \emph{dual} of a session type $S$, denoted by $\dual{S}$, describes the
mirrored protocol of $S$ and is corecursively defined by following equations:
\begin{align*}
    \dual\One & = \Bot &
    \dual{\Plus{\Tag_i : S_i}_{i\in I}} & = \With{\Tag_i : \dual{S_i}}_{i\in I} &
    \dual{(\Times{S}{T})} & = \Par{\dual{S}}{\dual{T}}
    \\
    \dual\Bot & = \One &
    \dual{\With{\Tag_i : S_i}_{i\in I}} & = \Plus{\Tag_i : \dual{S_i}}_{i\in I} &
    \dual{(\Par{S}{T})} & = \Times{\dual{S}}{\dual{T}}
\end{align*}

We say that the session types of the form $\One$, $\Plus{\Tag_i : S_i}_{i\in I}$
and $\Times{S}{T}$ are \emph{positive}, whereas the session types of the form
$\Bot$, $\With{\Tag_i : S_i}_{i\in I}$ and $\Par{S}{T}$ are \emph{negative}.
Positive session types describe protocols that begin with an \emph{output
action}, whereas negative session types describe protocols that begin with an
\emph{input action}. We write $\positive(S)$ and $\negative(S)$ to state that
$S$ is positive and negative, respectively.

\subsection{Fair Asynchronous Semantics}
\label{sec:types-lts}

\newcommand{\OneT}{\rulename{must-$\One$}\xspace}
\newcommand{\BotT}{\rulename{must-$\Bot$}\xspace}
\newcommand{\TimesT}{\rulename{must-$\Times{}{}$}\xspace}
\newcommand{\ParT}{\rulename{must-$\Par{}{}$}\xspace}
\newcommand{\PlusT}{\rulename{must-$\oplus$}\xspace}
\newcommand{\WithT}{\rulename{must-$\&$}\xspace}
\newcommand{\APlusT}{\rulename{may-$\oplus$}\xspace}
\newcommand{\AWithT}{\rulename{may-$\with$}\xspace}
\newcommand{\ATimesT}{\rulename{may-$\Times{}{}$}\xspace}
\newcommand{\AParT}{\rulename{may-$\Par{}{}$}\xspace}
\newcommand{\CoPlusT}{\rulename{fair-$\oplus$}\xspace}
\newcommand{\CoWithT}{\rulename{fair-$\with$}\xspace}

\newcommand{\SyncTrans}{cazzzz}
\newcommand{\AsyncTrans}{\rulename{must-async}\xspace}

\newcommand{\unit}{{*}}

\begin{table}
    \caption{\label{tab:lts}Labelled transition system for asynchronous session types.}
    \begin{mathpar}
        \inferrule[\OneT]{~}{
            \One \lred{\out\unit} \One
        }
        \and
        \inferrule[\TimesT]{~}{
            \Times{S}{T} \lred{\out S} T
        }
        \and
        \inferrule[\ATimesT]{
            T \lred{\inp\MessageType} T'
        }{
            \Times{S}{T} \lred{\inp\MessageType} \Times{S}{T'}
        }
        \\
        \inferrule[\BotT]{~}{
            \Bot \lred{\inp\unit} \Bot
        }
        \and
        \inferrule[\ParT]{~}{
            \Par{S}{T} \lred{\inp S} T
        }
        \and
        \inferrule[\AParT]{
            T \lred{\out\MessageType} T'
        }{
            \Par{S}{T} \lred{\out\MessageType} \Par{S}{T'}
        }
        \\
        \inferrule[\PlusT]{~}{
            \Plus{\Tag_i : S_i}_{i\in I} \lred{\out\Tag_k} S_k
        }
        \and
        \inferrule[\APlusT]{
            \forall i\in I: S_i \lred{\inp\MessageType} T_i
        }{
            \Plus{\Tag_i : S_i}_{i\in I} \lred{\inp\MessageType} \Plus{\Tag_i : T_i}_{i\in I}
        }
        \and
        \infercorule[\CoPlusT]{
        	\exists k\in I: S_k \lred{\inp\MessageType} T_k
        }{
            \Plus{\Tag_i : S_i}_{i\in I} \lred{\inp\MessageType} \Plus{\Tag_i : T_i}_{i\in I}
        }
        \\
        \inferrule[\WithT]{~}{
            \With{\Tag_i : S_i}_{i\in I} \lred{\inp\Tag_k} S_k
        }
        \and
        \inferrule[\AWithT]{
            \forall i\in I: S_i \lred{\out\MessageType} T_i
        }{
            \With{\Tag_i : S_i}_{i\in I} \lred{\out\MessageType} \With{\Tag_i : T_i}_{i\in I}
        }
        \and
        \infercorule[\CoWithT]{
        	\exists k\in I: S_k \lred{\out\MessageType} T_k
        }{
            \With{\Tag_i : S_i}_{i\in I} \lred{\out\MessageType} \With{\Tag_i : T_i}_{i\in I}
        }
    \end{mathpar}
\end{table}

We define the labelled transition system (LTS) for session types using the rules
in \Cref{tab:lts}. Labels of the transition system can be of the form
$\inp\MessageType$ (input of a message of type $\MessageType$) or
$\out\MessageType$ (output of a message of type $\MessageType$) where
$\MessageType$ is a \emph{message type} of the form $\unit$ (the type of a
termination signal), $\Tag$ (the singleton type of the tag $\Tag$) or $S$ (the
session type of a channel).
Hereafter we use $\MessageTypeS$ and $\MessageTypeT$ to range over message types
and $\ActionA$ and $\ActionB$ to range over labels.

Before we describe the rules in detail, we must point out two unusual but
important aspects of the LTS.
First of all, the LTS is specified as a \emph{Generalized Inference System} (GIS
for short~\citep{AnconaDagninoZucca17,Dagnino19}). A GIS consists of two sets of
rules, those that are meant to be interpreted coinductively (the singly-lined
rules in \Cref{tab:lts}) and those that are meant to be interpreted inductively
(the singly-lined rules plus the doubly-lined rules in \Cref{tab:lts}).
If we call $\ired\Action$ the relation that is defined by the inductive part of
the GIS, then the actual relation $\lred\Action$ being defined is the largest
one included in $\ired\Action$ that satisfies the singly-lined rules in
\Cref{tab:lts}. The interested reader may refer to the literature for a thorough
presentation of GIS \citet{AnconaDagninoZucca17,Dagnino19}, but the examples we
are about to discuss should suffice to clarify the nature of transitions.

The other unusual aspect of the LTS is that a transition $S \lred\Action T$ is
not an indication of what a process complying with $S$ necessarily does, but
rather of what the process is allowed or able to do. In particular, a transition
$S \lred{\out\MessageType} T$ means that a process complying with $S$ is allowed
to output a message of type $\MessageType$, even though $S$ may be negative. We
say that this is an \emph{early output transition} because it describes an
output that may occur ahead of time. Dually, a transition $S
\lred{\inp\MessageType} T$ means that a process complying with $S$ is eventually
able to input a message of type $\MessageType$, even though $S$ may be positive.
We say that this is a \emph{late input transition} because it describes the
consumption of a message that may occur later on.


The axioms \rulename{must-*} are used to derive what we call \emph{immediate
transitions}. These are the expected transitions of session types, whereby no
input is late and no output is early.
For technical reasons it is convenient to have transitions also for $\One$ and
$\Bot$. In this way we do not have to distinguish $\One$ and $\Bot$ from other
non-terminated protocols when defining our notions of generalized duality
(\Cref{def:cac}) and subtyping (\Cref{def:cas}).
%

The rules \rulename{may-*} deal with late inputs and early outputs.
%
As an example, consider the session type $S = \With{\TagA : \Plus{\TagC : T}}$
for which we may derive the transition sequences
\[
    S \lred{\inp\TagA} \Plus{\TagC : T} \lred{\out\TagC} T
    \text{\qquad and\qquad}
    S \lred{\out\TagC} \With{\TagA : T} \lred{\inp\TagA} T
\]

The sequence on the left, obtained by \WithT followed by \PlusT, is an ordinary
one: actions are performed according to the syntactic structure of the type. The
sequence on the right, obtained using \PlusT, \APlusT and \WithT, is peculiar to
the asynchronous setting: it describes a situation in which the output
$\out\TagC$ may be performed earlier than the input $\inp\TagA$. Since
communication is asynchronous and $\TagC$ is going to be sent anyway, a process
complying with $S$ might choose to send it \emph{early}, before waiting for
$\TagA$. Note that this is just a possibility: a process strictly adhering to
the transition sequence on the left would still comply with $S$.

A dual reasoning applies to late inputs. If we consider the type $S' =
\Plus{\TagA : \With{\TagC : T'}}$ we may derive the transition sequences
\[
    S' \lred{\out\TagA} \With{\TagC : T'} \lred{\inp\TagC} T'
    \text{\qquad and\qquad}
    S' \lred{\inp\TagC} \Plus{\TagA : T'} \lred{\out\TagA} T'
\]

Again, the sequence on the left is standard. The sequence on the right, obtained
using \WithT, \AWithT and \PlusT, says that a process complying with $S'$ is
able to consume a $\TagC$ message, even though this will happen only after the
process has sent $\TagA$. The fact that the input action is late cannot cause
issues, since both the outgoing $\TagA$ and the incoming $\TagC$ are sent
asynchronously with a non-blocking operation.

As it is clear looking at the rules \APlusT and \AWithT, late inputs and early
outputs concerning a branching session type must be derivable for every branch:
a message may be sent early (before an input) only if it is independent of the
input; a message may be received late (after an output) only if it is
independent of the output. 
For example, if we take $S = \With{\TagA : \Plus{\TagC : S_1}, \TagB :
\Plus{\TagC : S_2, \TagD : S_3}}$ then we can derive
\[
    S \lred{\out\TagC} \With{\TagA : S_1, \TagB : S_2}
    \text{\quad and also\quad}
    S \lred{\inp\TagB} \Plus{\TagC : S_2, \TagD : S_3} \lred{\out\TagD} S_3
    \text{\quad but not\quad}
    S \lred{\out\TagD} S'
\]
no matter what $S'$ could be. A process complying with $S$ may send $\TagC$
early, before receiving either $\TagA$ or $\TagB$, since the output of $\TagC$
is allowed regardless of the input. On the contrary, the output $\TagD$ is
allowed only if the input is $\TagB$ and so it cannot be anticipated before the
input.
Symmetrically, if we consider the session type $T = \Plus{\TagA : \With{\TagC : T_1}, \TagB :
\With{\TagC : T_2, \TagD : T_3}}$, we may derive
\[
    T \lred{\inp\TagC} \Plus{\TagA : T_1, \TagB : T_2}
    \text{\quad and also\quad}
    T \lred{\out\TagB} \With{\TagC : T_2, \TagD : T_3} \lred{\inp\TagD} T_3
    \text{\quad but not \quad}
    T \lred{\inp\TagD} T'
\]

The (late) input transition $\TagC$ is enabled because the process is able to
receive $\TagC$ regardless of the tag $\TagA$ or $\TagB$ that it sends. On the
contrary, the input transition on $\TagD$ is enabled only if the process sends
$\TagB$.

A subtler case of late input is illustrated by the session type $S_1 =
\Plus{\TagA : S_1, \TagB : \With{\TagC : S_2}}$, which describes the behavior of
a process that sends an arbitrary number of $\TagA$'s or a $\TagB$ and then
waits for a $\TagC$. Since the singly-lined rules in \Cref{tab:lts} are
interpreted coinductively, we can derive $S_1 \lred{\inp\TagC} S_1'$ where $S_1'
= \Plus{\TagA : S_1', \TagB : S_2}$ by means of the following \emph{infinite}
derivation:
\begin{equation}
    \label{eq:deep-input-derivation}
    \begin{prooftree}
        \[
            \mathstrut\smash\vdots
            \justifies
            S_1 \lred{\inp\TagC} S_1'
            \using\APlusT
        \]
        \[
            \justifies
            \With{\TagC : S_2} \lred{\inp\TagC} S_2
            \using\WithT
        \]
        \justifies
        S_1 \lred{\inp\TagC} S_1'
        \using\APlusT
    \end{prooftree}
\end{equation}

The transition $S_1 \lred{\inp\TagC} S_1'$ says that a process complying with
$S_1$ performs a late input of $\TagC$. Compared to the other examples of late
inputs, this one looks more questionable and for a good reason: it may be the
case that a process sending an \emph{infinite} sequence of $\TagA$'s complies
with $S_1$. So, claiming that such process is eventually able to input a $\TagC$
may lead to a message remaining orphan. In this work, however, we \emph{assume}
that a well-typed process complying with $S_1$ will also be \emph{fair}, in the
sense that it will eventually stop sending $\TagA$'s and will send the $\TagB$
ensuring that the $\TagC$ message is consumed.
%
%
The same assumption is made, in dual form, for early outputs. For example, the
session type $T_1 = \With{\TagA : T_1, \TagB : \Plus{\TagC : T_2}}$ performs the
early output transition $T_1 \lred{\out\TagC} T_1'$ where $T_1' = \With{\TagA :
T_1', \TagB : T_2}$. That is, a process complying with $T_1$ is allowed to send
$\TagC$ early because it is guaranteed to eventually receive a $\TagB$
message.

The infinite derivation that we have just shown in
\eqref{eq:deep-input-derivation} reminds us that we must be careful in the use
of coinduction for defining the LTS. Indeed, an unconstrained use of coinduction
would allow us to derive early/late transitions that do not correspond to any
``real'' action of a session type, solely using \rulename{may-*} rules. Consider
for example the infinite session type $S = \Plus{\Tag : S}$, which allows
sending a neverending sequence of $\Tag$'s. It is easy to derive $S
\lred{\inp\MessageType} S$ for every $\MessageType$ by an infinite derivation
consisting of \APlusT rules only, despite the protocol described by $S$ does
\emph{not} enable any input transition!
Of course we want to make sure that, whenever we derive a late/early transition,
this is justified by the use of at least one \rulename{must-*} rule somewhere in
the derivation. This is the reason why the LTS is defined through a GIS and not
simply by the (coinductively interpreted) singly-lined rules in \Cref{tab:lts}.
On the one hand, we want to make sure that that a late/early transition is
enabled along \emph{every} branch of a session type. This is an \emph{invariant
property} enforced by the \rulename{may-*} rules. On the other hand, we want to
make sure that there exists at least one branch along which the transition
eventually originates for real. This is a \emph{well-founded property} enforced
by the \rulename{must-*} rules. As it has already been observed
elsewhere~\cite{CicconePadovani22B}, GIS are a convenient way of defining
relations like $\lred\Action$ that mix invariant and well-founded properties at
the same time.
The effect of defining the LTS as a GIS is that, whenever we build a (possibly
infinite) derivation for $S \lred\Action T$ using the singly-lined rules, we
must also be able to find, \emph{for each} judgment $S_i \lred{\Action_i} T_i$
in this derivation, a \emph{finite} derivation for $S_i \ired{\Action_i} T_i$.
In \eqref{eq:deep-input-derivation} this is achieved easily, as shown below:
\[
    \begin{prooftree}
        \[
            \justifies
            \With{\TagC : S_2} \ired{\inp\TagC} S_2
            \using\WithT
        \]
        \justifies
        S_1 \ired{\inp\TagC} S_1'
        \using\CoPlusT
    \end{prooftree}
\]

Notice the key role of \CoPlusT in building this \emph{finite} derivation for
$S_1 \ired{\inp\TagC} S_1'$. Since $S_1$ is an \emph{infinite} session type, we
would not be able to find a finite derivation for $S_1 \ired{\inp\TagC} S_1'$ if
we insisted on using \APlusT only, since this rule requires us to derive the
late input transition in \emph{every} branch of the session type. Instead,
according to \CoPlusT it suffices to find \emph{one} branch along which we can
eventually derive the late input transition directly. In the literature on GIS
the doubly-lined rules are called \emph{corules}. Here, we have called them
\CoPlusT and \CoWithT because they somehow capture the fairness assumption of
the LTS: whenever we derive late/early transitions in a looping session type
like $S_1$ and $T_1$, the fairness assumption makes sure that the conversation
eventually follows a branch leading out of the loop.

There is one exception to what we have just said about transitions being
eventually derived by \rulename{must-*} rules. Recalling that $\Zero = \Plus{}$
and $\Top = \With{}$, by \APlusT and \AWithT we can easily derive $\Zero
\lred{\inp\MessageType} \Zero$ and $\Top \lred{\out\MessageType} \Top$ for every
$\MessageType$. These derivations are trivially valid for the GIS since they are
finite. We might be tempted to flag these cases as pathological. After all, the
protocol $\Zero$ describes an unrealistic process that is able to \emph{input
anything} and the protocol $\Top$ describes a uncontrollable process that may
\emph{output anything}. While it is true that such behaviors are practically
useless, we will see in \Cref{sec:subtyping} that the derivability of these
transitions for $\Zero$ and $\Top$ has an important impact in the resulting
subtyping relation, for which $\Zero$ and $\Top$ will play the role of least and
the greatest element.

Hereafter, we let $\ActionsA$ and $\ActionsB$ range over finite sequences of
labels, we write $\xlred{\Action_1\cdots\Action_n}$ for the composition
$\lred{\Action_1}\cdots\lred{\Action_n}$, we write $S \lred\Action$ if $S
\lred\Action T$ for some $T$ and $S \nlred\Action$ if not $S
\lred\Action$.

\subsection{Properties of Asynchronous Session Types}

To substantiate the claim that each transition derivable by the GIS is ``real''
(\ie it originates from the syntax of the session type), we prove that every
input/output transition can be derived by the application of an axiom in
\Cref{tab:lts} after every maximal, strongly fair sequence of immediate
outputs/inputs.
Strong fairness is a weaker assumption implied by the full fairness of
\cref{def:fair-run}~\citep{GlabbeekHofner19}.
Formally, a (possibly infinite) sequence of transitions $S_0 \lred{\Action_1}
S_1 \lred{\Action_2} \cdots$ is \emph{strongly
fair}~\citep{Francez86,GlabbeekHofner19} if, whenever some $S'$ occurs
infinitely often in the sequence $S_0S_1\cdots$ and $S' \lred\Action T'$, then
also $T'$ occurs infinitely often in the same sequence. Intuitively, a strongly
fair sequence of transitions does not discriminate those transitions that are
enabled infinitely often. For example, if $S = \Plus{\TagA : S, \TagB : T}$, the
infinite sequence $S \lred{\out\TagA} S \lred{\out\TagA} \cdots$ of transitions
is strongly unfair, because the transition $S \lred{\out\TagB} T$ is infinitely
often enabled but never performed. On the contrary, if $S = \Plus{\TagA : S}$,
then the infinite sequence $S \lred{\out\TagA} S \lred{\out\TagA} \cdots$ of
transitions is strongly fair.

\begin{restatable}{theorem}{thmfas}
    \label{thm:fas}
    Let $\Action$ be the label of an input/output transition. Then $S
    \lred\Action$ if and only if every maximal, strongly fair sequence of
    immediate output/input transitions $S \lred{\Action_1} S_1 \lred{\Action_2}
    \cdots$ is finite and ends in some $T$ such that $T \lred\Action$ is
    derivable by an axiom in \Cref{tab:lts}.
\end{restatable}

Note that, in the statement of \Cref{thm:fas}, we cannot require the transition
$T \lred\Action$ to be immediate because of the phony early/late transitions
enabled by $\Zero$ and $\Top$.

We are accustomed to think that session types are used to describe race-free
interactions, but we have seen examples of session types that simultaneously
enable both input and output transitions. When this happens, such transitions
are independent and do not interfere with each other. We formalize this fact by
establishing a diamond property for session types.

\begin{restatable}{proposition}{propdiamond}
    \label{prop:diamond}
    If\/ $S \lred{\inp\MessageTypeS} S'$ and $S \lred{\out\MessageTypeT} S''$,
    then  $S' \lred{\out\MessageTypeT} T$ and $S'' \lred{\inp\MessageTypeS} T$
    for some $T$.
\end{restatable}

\section{Fair Asynchronous Subtyping}
\label{sec:subtyping}

We define the subtyping relation for asynchronous session types in three steps.
First of all, we formalize what we mean by \emph{correct asynchronous
composition} between the session types describing the protocols implemented by
two processes using the two endpoints of a session. In many session type
systems, this notion coincides with session type duality. Since we have to take
asynchrony into account, duality alone is too strict so we need a notion of
correct composition using the LTS defined in \Cref{sec:session-types}. Once this
notion is in place, asynchronous subtyping can be defined as the relation that
preserves correct asynchronous composition. This relation is ``sound'' by
definition, but its properties can be intrinsically difficult to grasp. The
final step will be to provide a \emph{precise} (\ie sound and complete)
alternative characterization of asynchronous subtyping that sheds light on its
properties.

The definition of correct asynchronous composition is given below.

\begin{definition}
  \label{def:cac}
  We say that $\rrel$ is a \emph{correct asynchronous composition} if $(S,T) \in
  \rrel$ implies:
  \begin{enumerate}
  \item\label{cac:pp} either $\positive(S)$ or $\positive(T)$;
  \item\label{cac:oi} if $S \lred{\out\MessageType} S'$ and
      $\MessageType\in\set\unit\cup\TagSet$, then $T \lred{\inp\MessageType} T'$
      and $(S',T') \in \rrel$;
  \item\label{cac:io} if $T \lred{\out\MessageType} T'$ and
      $\MessageType\in\set\unit\cup\TagSet$, then $S \lred{\inp\MessageType} S'$
      and $(S',T') \in \rrel$;
  \item\label{cac:ois} if $S \xlred{\out S_1} S_2$, then $T \xlred{\inp T_1}
      T_2$ and $(S_1,T_1) \in \rrel$ and $(S_2,T_2) \in \rrel$;
  \item\label{cac:ios} if $T \xlred{\out T_1} T_2$, then $S \xlred{\inp S_1}
      S_2$ and $(S_1,T_1) \in \rrel$ and $(S_2,T_2) \in \rrel$.
  \end{enumerate}
  We write $\correct{}{}$ for the largest correct asynchronous composition.
\end{definition}

\newcommand{\rcac}[1]{\Cref{def:cac}(\ref{cac:#1})}

In words, \cref{cac:oi,cac:io,cac:ois,cac:ios} state that $(S,T)$ forms a
correct composition if, whenever one of the two types performs a (possibly
early) output transition $\out\MessageType$, the other type is able to respond
with a (possibly late) compatible input transition $\inp\MessageTypeT$ and the
continuations remain correct. In general this is not enough to guarantee
progress because early outputs are only allowed but not mandatory. For example,
the types $\With{\TagA : \Plus{\TagB : \One}}$ and $\With{\TagB : \Plus{\TagA :
\Bot}}$ satisfy \cref{cac:oi,cac:io}, but two processes strictly adhering to
these protocols (\ie without performing early outputs) would starve.
\Cref{cac:pp} requires that at least one of the two types is positive, namely
that at least on one side of the session the outputs are guaranteed to be
immediate. This is enough to ensure progress. For example, we have
$\correct{\With{\TagA : \Plus{\TagB : \One}}}{\Plus{\TagA : \With{\TagB :
\Bot}}}$ as well as $\correct{\Plus{\TagB : \With{\TagA : \One}}}{\Plus{\TagA :
\With{\TagB : \Bot}}}$.

It is easy to see that $\correct{}{}$ is symmetric and that duality implies
correctness:
\begin{restatable}{proposition}{propdualcorrect}
  \label{prop:dual-correct}
  $\correct{\dual{S}}S$ holds for every session type $S$.
\end{restatable}

Other properties of $\correct{}{}$ are more surprising. For example, if $S =
\With{\TagA : S, \TagB : \Plus{\TagC : \One}}$ and $R = \Plus{\TagA : R}$, we
have that $\correct{S}{R}$ does \emph{not} hold because of the early output
transition $S \lred{\out\TagC}$ to which $R$ is unable to respond. The lack of
compatibility between $S$ and $R$ is due to the fact that (the process behaving
as) $S$ makes a \emph{fairness assumption} on the behavior of the process it is
interacting with. More precisely, a process complying with $S$ assumes that
sooner or later a $\TagB$ message will be received, finally enabling the output
of $\TagC$. In anticipation of this, the process may decide to perform an early
output of $\TagC$, but in doing so it would generate an orphan message when
interacting with another process adhering to $R$, which is not honoring this
fairness assumption.
It is also easy to see that $\correct\Zero{S}$ holds for every $S$ and that
$\correct\Top{S}$ implies $S = \Zero$. These properties of $\Zero$ and $\Top$
follow directly from the \APlusT and \AWithT rules that we have commented in
\Cref{sec:session-types}.

\begin{example}
  \label{ex:server-worker-correct}
  Let $\out\Tag.S$ stand for $\Plus{\Tag : S}$ and $(\out\Tag)^n.S$ stand for
  $\out\Tag\ldots\out\Tag.S$ with $n$ $\out\Tag$ prefixes.
  Consider the session types $S = \Plus{ \TagTask : S, \TagStop : T}$ and $T =
  \With{\TagRes : T, \TagStop : \Bot}$ and $U = \With{\TagTask : \Plus{\TagRes :
  U}, \TagStop : \Plus{\TagStop : \One}}$ which respectively describe the
  behaviors of $\Split$, $\Gather$ and $\Worker$ of \Cref{sec:introduction} on
  the channel $y$.
  %
  It is easy to establish that
  \[
    \srel \eqdef \set{
      (S, (\out\TagRes)^n.U) \mid n \in \Nat
    } \cup \set{
      (T, (\out\TagRes)^n.\out\TagStop.\One) \mid n \in \Nat
    } \cup \set{
      (\Bot, \One)
    }
  \]
  is a correct asynchronous composition hence $\correct{S}{U}$.
  \eoe
\end{example}

We can now define fair asynchronous subtyping semantically using \emph{Liskov's
substitution principle}~\cite{LiskovWing94} where the property being preserved
is session correctness (\Cref{def:cac}).

\begin{definition}[fair asynchronous subtyping]
  \label{def:subt}
  We say that $S$ is a \emph{fair asynchronous subtype} (or just \emph{subtype})
  of $T$, notation $S \subt T$, if $\correct{R}{T}$ implies $\correct{R}{S}$ for
  every $R$.
\end{definition}

Paraphrasing, this definition says that a process using a channel $x$ according
to $T$ can be safely replaced by a process using $x$ according to $S$ when $S$
is a subtype of $T$. Indeed, the peer process, which is assumed to use the same
channel $x$ according to some session type $R$ such that $\correct{R}{T}$, will
still interact correctly after the substitution has taken place.

A few subtyping relations are easy to figure out. For example, we have $\Zero
\subt S$ and $S \subt \Top$ for every $S$ because of the properties of $\Zero$
and $\Top$ that we have pointed out above.
It is also easy to see that $\Plus{\TagA : \With{\TagB : S}} \subt \With{\TagB :
\Plus{\TagA : S}}$ holds in an asynchronous setting. After all, the process that
behaves according to $\Plus{\TagA : \With{\TagB : S}}$ is sending a tag $\TagA$
that also the process that behaves according to $\With{\TagB : \Plus{\TagA :
S}}$ \emph{may} anticipate. Note that the inverse relation $\With{\TagB :
\Plus{\TagA : S}} \subt \Plus{\TagA : \With{\TagB : S}}$ does not hold, despite
the fact that $\Plus{\TagA : \With{\TagB : S}}$ and $\With{\TagB : \Plus{\TagA :
S}}$ perform exactly the same transitions, because $\With{\TagA : \Plus{\TagB :
\dual{S}}}$ forms a correct asynchronous composition with $\Plus{\TagA :
\With{\TagB : S}}$ but not with $\With{\TagB : \Plus{\TagA : S}}$ (\Cref{cac:pp}
of \Cref{def:cac} is violated).

The above notion of subtyping is sound ``by definition'', but provides little
information concerning the shape of related session types. To compensate for
this problem, we also give a sound and complete coinductive characterization of
$\subt$, which relates directly to \Cref{def:cac}.

\begin{definition}[coinductive asynchronous subtyping]
  \label{def:cas}
  We say that $\srel$ is a \emph{coinductive asynchronous subtyping} if $(S,T)
  \in \srel$ implies:
  \begin{enumerate}
  \item\label{cas:pol} either $\positive(S)$ or $\negative(T)$;
  \item\label{cas:inp} if $T \lred{\inp\MessageType} T'$ and
      $\MessageType\in\set\unit\cup\TagSet$, then $S \lred{\inp\MessageType} S'$
      and $(S',T') \in \srel$;
  \item\label{cas:out} if $S \lred{\out\MessageType} S'$ and
      $\MessageType\in\set\unit\cup\TagSet$, then $T \lred{\out\MessageType} T'$
      and $(S',T') \in \srel$;
  \item\label{cas:inps} if $T \xlred{\inp T_1} T_2$, then $S \xlred{\inp S_1}
      S_2$ and $(S_1,T_1) \in \srel$ and $(S_2,T_2) \in \srel$;
  \item\label{cas:outs} if $S \xlred{\out S_1} S_2$, then $T \xlred{\out T_1}
      T_2$ and $(S_1,T_1) \in \srel$ and $(S_2,T_2) \in \srel$.
  \end{enumerate}
\end{definition}

\newcommand{\rcas}[1]{\Cref{def:cas}(\ref{cas:#1})}

\cref{cas:inp,cas:out,cas:inps,cas:outs} of \Cref{def:cas} specify the expected
requirements for a session subtyping relation: every input transition of the
supertype $T$ must be matched by an input transition of the subtype $S$ and the
corresponding continuations should still be related by subtyping; dually, every
output transition of the subtype $S$ must be matched by an output transition of
the supertype $T$ and the corresponding continuations should still be related by
subtyping.
Interestingly, these items are essentially the same found in analogous
characterizations of \emph{synchronous subtyping for session
types}~\cite{GayHole05}, modulo the different orientation of $\subt$ due to our
viewpoint based on the substitution of processes rather than on the substitution
of channels.\footnote{The interested reader may refer to Gay~\cite{Gay16} for a
comparison of the two viewpoints.}
However, the clauses of \cref{def:cas} are not mutually exclusive, because the
same session type may perform both input and output transitions. Also, session
types related by subtyping need not start with the same type constructor. In
this respect, \cref{cas:pol} makes sure that the smaller session type can only
anticipate (and not postpone) outputs, as argued above.

\Cref{def:cas} is a sound and complete characterization of $\subt$.

\begin{restatable}{theorem}{thmsubt}
    \label{thm:subt}
    $\subt$ is the largest coinductive asynchronous subtyping.
\end{restatable}

Using \Cref{def:cas,thm:subt} we can prove some significant properties of
$\subt$:
\begin{itemize}
  \item (input contravariance) $\With{\Tag_i : S_i}_{i\in I} \subt \With{\Tag_i :
    S_i}_{i\in J}$ if $J \subseteq I$ by \Cref{cas:inp};
  \item (output covariance) $\Plus{\Tag_i : S_i}_{i\in I} \subt \Plus{\Tag_i :
    S_i}_{i\in J}$ if $I \subseteq J$ by \Cref{cas:out};
  \item (output anticipation) $\Plus{\TagA_j : \With{\TagB_i : S_{ij}}_{i\in
    I}}_{j\in J} \subt \With{\TagB_i : \Plus{\TagA_j : S_{ij}}_{j\in J}}_{i\in
    I}$. Note that the inverse relation does not hold, despite these two session
    types have exactly the same transitions, because of \Cref{cas:pol}.
\end{itemize}

There are exceptions to output covariance and input contravariance when one of
the two types specifies a non-terminating protocol. The next example illustrates
one of such cases.

\begin{example}
  \label{ex:failed-variance}
  Consider $S = \With{\TagA : S, \TagB : \Plus{\TagC : \One}}$ and $T =
  \With{\TagA : T}$. Despite $S$ is a subtype of $T$ for other session
  subtypings
  \cite{GayHole05,MostrousYoshida15,ChenDezaniScalasYoshida17,BravettiLangeZavattaro24},
  we have $S \not\subt T$ because $S \lred{\out\TagC}$ and $T
  \nlred{\out\TagC}$. To see the reason why admitting this relation could cause
  a problem, consider a process that complies with the protocol $R = \Plus{\TagA
  : R}$. Such a process would output infinitely many $\TagA$'s and would not
  expect to input anything. Still, if we consider $S' = \With{\TagA : S', \TagB
  : \One}$ we have that $\Plus{\TagC : S'} \subt S$ holds. That is, a process
  complying with $\Plus{\TagC : S'}$ performs immediately the early output in
  $S$. By transitivity of $\subt$ (whose validity is implied by
  \Cref{thm:subt}), we would also have $\Plus{\TagC : S'} \subt T$. Now the
  problem is clear: $\correct{R}{T}$ holds, but a process complying with $R$
  would not be able to handle the incoming $\TagC$ message if we composed it
  with a process complying with $\Plus{\TagC : S'}$.
  \eoe
\end{example}

The cases in which co/contra variance does not hold have no impact on the
typeability of \Calculus processes: since our type system ensures the fair
termination of well-typed processes, protocols like $T$ in
\cref{ex:failed-variance} are not inhabited. We will see in
\cref{sec:subtyping-inclusions} that $\subt$ includes other subtyping relations
supporting full co/contra variance for the family of fairly terminating session
types, those that always allow the protocol to end.

\begin{example}
  \label{ex:satellite}
  We borrow a scenario from Bravetti, Lange and
  Zavattaro~\cite{BravettiLangeZavattaro24} to showcase an interesting example
  of fair asynchronous subtyping. Imagine a system made of a ground station and
  a satellite such that, at each flyby, the satellite sends data from the
  previous orbit and receives commands to execute in the next one.
  In principle, the ground station should follow the protocol $U =
  \With{\TagData : U, \TagStop : V}$ where $V = \Plus{\TagCommand : V, \TagStop
  : \Bot}$.
  However, since the flyby window may be short, it makes sense to implement the
  ground station so that it communicates with the satellite in full duplex by
  anticipating the output of the commands. In this case, the ground station
  follows the protocol $S = \Plus{\TagCommand : S, \TagStop : T}$ where $T =
  \With{\TagData : T, \TagStop : \Bot}$.
  Is this implementation correct? We can answer in the affirmative by proving $S
  \subt U$ with the diagram below, which represents a coinductive asynchronous
  subtyping containing the pair $(S,U)$. In the diagram we also use the types
  $S' = \Plus{\TagTask : S', \TagStop : \Bot}$ and $U' = \With{\TagRes : U',
  \TagStop : \Bot}$.
  \begin{center}
    \begin{tikzpicture}[xscale=4,thick,auto]
      \node (Q0) at (0,0) {$(S,U)$};
      \node (Q1) at (1,1) {$(T,U')$};
      \node (Q2) at (1,-1) {$(S',V)$};
      \node (Q3) at (2,0) {$(\Bot,\Bot)$};
      \draw[->,loop above] (Q0) to node {$\out\TagCommand$,$\inp\TagData$} (Q0);
      \draw[->] (Q0) to node {$\out\TagStop$} (Q1);
      \draw[->] (Q0) to node[swap] {$\inp\TagStop$} (Q2);
      \draw[->,loop above] (Q1) to node {$\inp\TagData$} (Q1);
      \draw[->,loop above] (Q2) to node {$\out\TagCommand$} (Q2);
      \draw[->] (Q1) to node {$\inp\TagStop$} (Q3);
      \draw[->] (Q2) to node[swap] {$\out\TagStop$} (Q3);
      \draw[->,loop above] (Q3) to node {$\inp\unit$} (Q3);
    \end{tikzpicture}
  \end{center}

  Note that $S$ allows for the anticipation of an unbounded number of outputs
  before an unbounded number of inputs.
  %
  %
  %
  \eoe
\end{example}

We now list a few properties of $\subt$. First of all, we establish that $\subt$
is closed by duality.

\begin{proposition}
  \label{prop:dual-subtype}
  If $S \subt T$, then $\dual{T} \subt \dual{S}$.
\end{proposition}
\begin{proof}
  By \Cref{thm:subt} it suffices to show that if $\srel$ is a coinductive
  asynchronous subtyping, then so is $\dual\srel \eqdef \set{ (\dual{T},
  \dual{S}) \mid (S,T) \in \srel}$. This follows immediately from
  \Cref{def:cas}.
\end{proof}

Then, we show how to characterize $\correct{}{}$ solely in terms of
$\subt$.

\begin{restatable}{theorem}{thmcorrectsubt}
    \label{thm:correct-subt}
    $\correct{S}{T}$ if and only if $S \subt \dual{T}$
\end{restatable}

Finally, we observe that similarly to the other asynchronous session type
theories correct composition is undecidable. The source of undecidability
follows from the possibility to use communicating buffers to model unbounded
memories like tapes of Turing Machines \cite{LangeY17} or queues of Queue
Machines \cite{BravettiCZ17}.  

\begin{restatable}{theorem}{undecidability}
    \label{thm:undecidable}
    Given two types $S$ and $T$, the problem of checking $\correct{S}{T}$ is
    undecidable.
\end{restatable}

As a direct consequence of \cref{thm:undecidable,thm:correct-subt}, we have that
checking $S \subt T$ is also an undecidable problem. Despite these negative
results, the coinductive characterizations of correctness (\Cref{def:cac}) and
subtyping (\Cref{def:cas}) are useful to prove that these relations hold in
specific cases, as done in \Cref{ex:server-worker-correct,ex:satellite}.
\section{Type System}
\label{sec:type-system}

In this section we describe the type system for \Calculus ensuring that
well-typed processes weakly terminate, besides being free from communication
errors and deadlocks. \Cref{thm:proof-principle} then allows us to conclude that
the very same type system also ensures fair termination when we make the
fairness assumption stated in \Cref{def:fair-run}.
The type system is based on the proof rules of \MALL but also includes typing
rules for the non-logical process forms, namely termination, process invocation
and non-deterministic choice.

\subsection{Measuring Processes and Type Annotations}
\label{sec:measure}

It is a known fact that infinitary proof systems for linear logic (\eg
\muMALL~\cite{BaeldeDoumaneSaurin16,Doumane17}) require validity conditions to
enjoy the cut elimination property. As a consequence, the infinitary type system
for \Calculus requires validity conditions to ensure the (weak) termination
property. We implement these validity conditions using the technique of Dagnino
and Padovani~\cite{DagninoPadovani24} and decorate types and typing judgments
with quantitative information that estimates the amount of effort required to
terminate a process. It is natural to measure such effort in terms of
\emph{number of reductions} of that process (\Cref{tab:semantics}). Since most
reductions arise from session interactions and each session interaction involves
a process outputting a message on a channel (with a \emph{positive} type) and a
process inputting a message from the same channel (with a \emph{negative} type)
we count the number of process forms representing outputs to establish the
number of reductions that are necessary to normalize a process. In the presence
of branching processes (those whose behavior depends on a tag received from a
channel), we can tentatively compute an upper bound for terminating \emph{each}
branch.

Let us use the definitions in \Cref{sec:introduction} to walk through this
approach on increasingly complex examples. To start, consider the process
$\Gather$, which we repeat here for convenience:
\begin{align*}
    \Let\Gather{x,y} & = \Case\y{\TagRes : \Call\Gather{x,y}, \TagStop : \Wait\y.\Select\x\TagResp.\Close\x }
\end{align*}

Now suppose that $n$ is the measure associated with $\Gather$ and note that
$\Gather$ is a branching and recursive process. The $\TagRes$ branch does not
contain any output actions, whereas the $\TagStop$ branch performs two outputs
on $x$. Therefore, $n$ should satify the relations $n \geq n$ (for the $\TagRes$
branch) and $n \geq 2$ (for the $\TagStop$ branch) whose least solution is $n =
2$. That is, $\Gather$ weakly terminates by performing at most two outputs. Note
that the actual number of interactions performed by $\Gather$ depends on the
number of received $\TagRes$ messages and may be larger than $2$. However, these
received messages are accounted for in the measure of the sender, while the
measure $2$ we give to $\Gather$ only accounts for the messages sent by
$\Gather$.

If we now consider the process
\begin{align*}
    \Let\Split{x,y} = \Select\y\TagTask.\Call\Split{x,y} \choice \Select\y\TagStop.\Call\Gather{x,y}
\end{align*}
we have to measure a non-deterministic choice.
Since non-deterministic choices are performed autonomously by a process and we
are interested in proving that processes \emph{weakly} terminate, we can measure
a non-deterministic choice by considering the branch with the least measure and
the additional reduction due to \rulename{r-choice}. So, in this case, the
measure of $\Split$ -- say $m$ -- must satisfy the equation $m \geq 1 +
\min\set{1 + m, 3}$ where $1 + m$ is the measure of the left branch and $3$ is
the (least) measure of the right branch, having already established that
$\Gather$ can be given measure $2$. The least least solution of this constraint
is $m = 4$.

A more challenging process to measure is
\[
    \Let\Worker\y = \Case\y{
        \TagTask : \Select\y\TagRes.\Call\Worker\y,
        \TagStop : \Select\y\TagStop.\Close\y
    }
\]
in which each $\TagTask$ input is immediately followed by a $\TagRes$ output.
If we try to measure $\Worker$ following the same reasoning applied to
$\Gather$, we end up looking for some $n$ satisfying the system of inequations
$n \geq 1 + n$ (for the $\TagTask$ branch) and $n \geq 2$ (for the $\TagStop$
branch). Since no $n$ satisfies these inequations, we have to refine our
measuring strategy or else $\Worker$ would end up being ill typed.
This is precisely what happens in some type systems ensuring the fair
termination of sessions~\cite{CicconePadovani22A,CicconeDagninoPadovani22} which
is unfortunate because $\Worker$ is not ill behaved after all: the fairness
assumption we are making ensures that the number of tasks will be finite, albeit
unbounded. Likewise, the number of results will be finite but unbounded as well.
The problem is that we are unable to find a \emph{single} measure for $\Worker$
that accounts for all possibilities.

We refine our measuring technique allowing the measure of $\Worker$ to
\emph{depend} on the tags it receives. We realize this dependency by annotating
tags with measures that are charged to senders and discharged from receivers.
To see this mechanism at work in the case of $\Split$ and $\Worker$, consider
the following annotated session types:
\begin{align}
    \label{eq:split-type}
    S & = \Plus{\ann[1]\TagTask : S, \ann[0]\TagStop : T} \\
    \label{eq:gather-type}
    T & = \With{\ann[0]\TagRes : T, \ann[0]\TagStop : \Bot} \\
    \label{eq:worker-type}
    U & = \With{\ann[1]\TagTask : \Plus{\ann[0]\TagRes : U}, \ann[0]\TagStop : \Plus{\ann[0]\TagStop : \One}}
\end{align}

The annotations in $\ann[1]\TagTask$ and $\ann[0]\TagStop$ mean that a process
like $\Split$, which complies with $S$, is \emph{charged} by $1$ unit of measure
whenever it sends $\TagTask$ and by 0 when it sends $\TagStop$. This measure is
charged \emph{in addition} to the cost of the output it is performing and
accounts for the cost of sending the result in $\Worker$. Since this cost is
already charged on $\Split$, it can be \emph{discharged} from the $\TagTask$
branch of $\Worker$ when we compute its measure, as shown by $U$. This way we
end up solving for $\Worker$ the equations $n \geq n$ (for the $\TagTask$
branch) and $n \geq 2$ (for the $\TagStop$ branch) whose least solution is $n =
2$.
In principle we have to reconsider the measure we gave to $\Split$ to account
for the 1 unit that is now charged for each output of $\TagTask$, but since the
measure of $\Split$ was determined by the $\TagStop$ branch (which has a null
annotation) it remains unchanged.

From now on we consider a refinement of the session types presented in
\Cref{sec:session-types} where tags are annotated with natural numbers.
These annotations do not interfere in any way with the properties and
characterizations of session types we have presented there and
in~\Cref{sec:subtyping}, with the proviso that wherever we have written matching
message types in \cref{sec:session-types,sec:subtyping} we mean that measures
also match.
In particular, \Cref{def:cac} entails that the amount of measure charged on one
side of the session matches the amount of measure discharged from the other side
of the session (like for $S$ and $U$ above) and the clauses of \Cref{def:cas}
entail that matching actions in session types related by subtyping carry the
same measure.
In principle, it would be possible to relax these constraints and allow some
variance of measure annotations. Since we do not have concrete examples that
take advantage of this further refinement, we spare the additional complexity.

\subsection{Typing Rules}
\label{sec:typing-rules}

\begin{table}
    \caption{Typing rules for \Calculus.}
    \label{tab:typing-rules}
    \begin{mathpar}
        \inferrule[\DoneRule]{~}{
            \wtp[n]\Done\EmptyContext
        }
        \and
        \inferrule[\CallRule]{
            \wtp{P}{\seqof{x : S}}
        }{
            \wtp[m+n]{\Call{A}{\seqof\x}}{\seqof{x : S}}
        }
        ~
        \Let{A}{\seqof\x} = P
        \and
        \inferrule[\CloseRule]{~}{
            \wtp[1+n]{\Close\x}{x : \One}
        }
        \and
        \inferrule[\WaitRule]{
            \wtp{P}\Context
        }{
            \wtp{\Wait\x.P}{\Context, x : \Bot}
        }
        \and
        \inferrule[\ChoiceRule]{
            \wtp[n_1]{P}\Context
            \\
            \wtp[n_2]{Q}\Context
        }{
            \wtp[1+n_k]{P \choice Q}\Context
        }
        \and
        \inferrule[\ForkRule]{
            \wtp[m]{P}{\ContextC, y : S}
            \\
            \wtp[n]{Q}{\ContextD, x : T}
        }{
            \wtp[1+m+n]{\Fork\x\y{P}{Q}}{\ContextC, \ContextD, x : \Times{S}{T}}
        }
        \and
        \inferrule[\JoinRule]{
            \wtp{P}{\Context, y : S, x : T}
        }{
            \wtp{\Join\x\y.P}{\Context, x : \Par{S}{T}}
        }
        \and
        \inferrule[\SelectRule]{
            \wtp[n]{P}{\Context, x : S_k}
        }{
            \wtp[1+n+m_k]{\Select\x{\Tag_k}.P}{\Context, x : \Plus{\ann[m_i]{\Tag_i} : S_i}_{i\in I}}
        }
        \and
        \inferrule[\CaseRule]{
            \forall i\in I: \wtp[n+m_i]{P_i}{\Context, x : S_i}
        }{
            \wtp[n]{\Case\x{\Tag_i:P_i}_{i\in I\cup J}}{\Context, x : \With{\ann[m_i]{\Tag_i} : S_i}_{i\in I}}
        }
        \and
        \inferrule[\LinkRule]{~}{
            \wtp[1+n]{\Link\x\y}{x : S, y : T}
        }
        ~
        \dual{S} \subt T
        \and
        \inferrule[\CutRule]{
            \wtp[m]{P}{\ContextC, x : S}
            \\
            \wtp[n]{Q}{\ContextD, x : T}
        }{
            \wtp[m+n]{\Cut\x{P}{Q}}{\ContextC, \ContextD}
        }
        ~
        \cc{S}{T}
    \end{mathpar}
\end{table}

The typing rules for \Calculus are shown in \Cref{tab:typing-rules}.
Judgments have the form $\wtp[n]{P}\Context$ where $P$ is the process being
typed, $n$ is its measure and $\Context$ is a \emph{typing context}, namely a
partial function that maps channels to types.
We let $\ContextC$ and $\ContextD$ range over typing contexts, we write
$\EmptyContext$ for the empty context, $x : S$ for the singleton context that
maps $x$ to $S$, and $\ContextC, \ContextD$ for the union of $\ContextC$ and
$\ContextD$ when they have disjoint domains.
The typing rules shown in \Cref{tab:typing-rules} are meant to be interpreted
coinductively, hence a process is well typed provided there is a (possibly
infinite) derivation built using those rules. The structure of the rules is
essentially the same of other session type systems based on (classical) linear
logic~\citep{Wadler14,LindleyMorris16,DagninoPadovani24}, so we will mainly
focus on those aspects of the rules that are new or different in our setting.

The rule \DoneRule states that the terminated process can have any measure and
is well typed only in the empty context.

The rule \CallRule states that a process invocation is well typed in a context
$\seqof{x : S}$ if so is its definition in the same context. Recall that a
typing derivation can be infinite in the case of recursive processes and note
that the measure of the invocation may be larger than that of its definition,
allowing some measure to be discarded from one invocation to the next. This is
not strictly necessary for the soundness of type system, but it is sometimes
convenient to obtain simpler typing derivations.

The rules \CloseRule and \WaitRule deal with session termination in the expected
way. The measure of $\Close\x$ is strictly positive whereas the measure of
$\Wait\x.P$ coincides with that of $P$ since the measure $n$ in a typing
judgment $\wtp\Context{P}$ only accounts for the outputs performed by $P$.

The rules \ForkRule and \JoinRule deal with the communication of channels.
The measure of a channel output $\Fork\x\y{P}Q$ accounts for both the measure of
$P$ (the process being spawned that uses one end of the fresh session $y$) and
the measure of $Q$, while the additional unit of measure accounts for the output
being performed.
The measure of a channel input $\Join\x\y.P$ simply coincides with that of the
continuation process.

The rules \SelectRule and \CaseRule deal with the communication of tags.
The measure of a tag output $\Select\x\Tag.P$ accounts for the measure of the
continuation $P$ and of the measure annotation $m$ associated with $\Tag$ in the
type of $x$, while the additional unit of measure accounts for the output being
performed.
The measure of a tag input $\Case\x{\Tag_i : P_i}_{i\in I}$ is the residual
measure of each branch after the measure annotation $m_i$ associated with
$\Tag_i$ in the type of $x$ has been discharged.
Note that a tag input may in general provide more branches than those actually
occurring in the type of $x$.
As a special case, \CaseRule also captures the introduction rule for $\Top$ in
\MALL. In particular, we have $\wtp[n]{\Case\x{}}{\Context, x : \Top}$ for every
$n$ and every $\Context$.

The rule \ChoiceRule deals with non-deterministic choices in the expected way.
The rule does not specify which branch of a non-deterministic choice determines
the measure of the whole choice. Since the typeability of a process rests on its
measurability, the obvious strategy is to pick the continuation with the
smallest measure (\cf \Cref{sec:measure}).

The rule \LinkRule is a generalization of the axiom in \MALL. It allows the
unification of $x$ and $y$ provided that $\dual{S} \subt T$ (or equivalently
$\dual{T} \subt S$, by \Cref{prop:dual-subtype}) where $S$ is the type of $x$
and $T$ is the type of $y$. As we will see in
\Cref{sec:subtyping-usage-patterns}, this formulation of the link typing rule
allows us to capture some recurring usage patterns of subtyping. A link has a
strictly positive measure since it accounts for one \rulename{r-link} reduction.

Finally, the rule \CutRule rule ensures that the processes $P$ and $Q$ connected
by the new session $x$ use the channel correctly requiring $\cc{S}{T}$ to hold
where $S$ is the type of $x$ as used by $P$ and $T$ is the type of $x$ as used
by $Q$. Note that this correctness condition can be equivalently expressed in
terms of subtyping with the relation $S \subt \dual{T}$ by
\Cref{thm:correct-subt}.
The measure of the composition is the combination of the measures of $P$ and
$Q$.

\begin{table}
    \caption{\label{tab:server-typing}Typing derivations for $\Split$, $\Gather$ and $\Worker$.}
    \begin{center}
        \begin{prooftree}
            \[
                \[
                    \mathstrut\smash\vdots
                    \justifies
                    \wtp[2]{
                        \Call\Gather{x,y}
                    }{
                        x : V,
                        y : T    
                    }
                \]
                \[
                    \[
                        \[
                            \justifies
                            \wtp[1]{
                                \Close\x
                            }{
                                x : \One
                            }
                            \using\CloseRule
                        \]
                        \justifies
                        \wtp[2]{
                            \Select\x\TagResp.\Close\x
                        }{
                            x : V
                        }
                        \using\SelectRule
                    \]
                    \justifies
                    \wtp[2]{
                        \Wait\y.\Select\x\TagResp.\Close\x
                    }{
                        x : V,
                        y : \Bot
                    }
                    \using\WaitRule
                \]
                \justifies
                \wtp[2]{
                    \Case\y{\TagRes : \Call\Gather{x,y}, \TagStop : \Wait\y.\Select\x\TagResp.\Close\x }
                }{
                    x : V,
                    y : T    
                }
                \using\CaseRule
            \]
            \justifies
            \wtp[2]{
                \Call\Gather{x,y}
            }{
                x : V,
                y : T
            }
            \using\CallRule
        \end{prooftree}
        \\\bigskip
        \begin{prooftree}
            \[
                \[
                    \[
                        \mathstrut\smash\vdots
                        \justifies
                        \wtp[4]{
                            \Call\Split{x,y}
                        }{
                            x : V,
                            y : S
                        }
                    \]
                    \justifies
                    \wtp[6]{
                        \Select\y\TagTask.\Call\Split{x,y}
                    }{
                        x : V,
                        y : S
                    }
                    \using\SelectRule
                \]
                \[
                    \[
                        \mathstrut\smash\vdots
                        \justifies
                        \wtp[2]{
                            \Call\Gather{x,y}
                        }{
                            x : V,
                            y : T
                        }
                        \using\LinkRule
                    \]
                    \justifies
                    \wtp[3]{
                        \Select\y\TagStop.\Call\Gather{x,y}
                    }{
                        x : V,
                        y : S
                    }
                    \using\SelectRule
                \]
                \justifies
                \wtp[4]{
                    \Select\y\TagTask.\Call\Split{x,y}
                    \choice
                    \Select\y\TagStop.\Call\Gather{x,y}
                }{
                    x : V,
                    y : S
                }
                \using\ChoiceRule
            \]
            \justifies
            \wtp[4]{
                \Call\Split{x,y}
            }{
                x : V,
                y : S
            }
            \using\CallRule
        \end{prooftree}
        \\\bigskip
        \begin{prooftree}
            \[
                \[
                    \[
                        \mathstrut\smash\vdots
                        \justifies
                        \wtp[2]{
                            \Call\Worker\y
                        }{
                            y : U
                        }
                    \]
                    \justifies
                    \wtp[3]{
                        \Select\y\TagResp.\Call\Worker\y
                    }{
                        y : \Plus{ \ann[0]\TagResp : U }
                    }
                    \using\SelectRule
                \]
                \[
                    \[
                        \justifies
                        \wtp[1]{
                            \Close\y
                        }{
                            y : \One
                        }
                        \using\CloseRule
                    \]
                    \justifies
                    \wtp[2]{
                        \Select\y\TagStop.\Close\y
                    }{
                        y : \Plus{ \ann[0]\TagStop : \One }
                    }
                    \using\SelectRule
                \]
                \justifies
                \wtp[2]{
                    \Case\y{
                        \TagTask : \Select\y\TagResp.\Call\Worker\y,
                        \TagStop : \Select\y\TagStop.\Close\y
                    }
                }{
                    y : U
                }
                \using\CaseRule
            \]
            \justifies
            \wtp[2]{
                \Call\Worker\y
            }{
                y : U
            }
            \using\CallRule
        \end{prooftree}
    \end{center}
\end{table}

\begin{example}
    \label{ex:server-typing}
    Let us build a typing derivation for the server $\Case\x{\TagReq :
    \Cut\y{\Call\Split{x,y}}{\Call\Worker\y}}$ using the types $S$, $T$ and $U$
    defined in \Cref{eq:split-type,eq:gather-type,eq:worker-type}. It is
    convenient to also use the type $V = \Plus{ \ann[0]\TagResp : \One }$.
    \Cref{tab:server-typing}
    shows the typing derivations for the $\Split$,
    $\Gather$ and $\Worker$ processes. Note that the measures obtained for these
    processes coincide with those previously inferred in \Cref{sec:measure}.
    From $\correct{S}{U}$ (\Cref{ex:server-worker-correct}) we derive
    \[
        \begin{prooftree}
            \[
                \[
                    \mathstrut\smash\vdots
                    \justifies
                    \wtp[4]{
                        \Call\Split{x,y}
                    }{
                        x : V,
                        y : S
                    }
                \]
                \[
                    \mathstrut\smash\vdots
                    \justifies
                    \wtp[2]{
                        \Call\Worker\y
                    }{
                        y : U
                    }
                \]
                \justifies
                \wtp[6]{
                    \Cut\y{\Call\Split{x,y}}{\Call\Worker\y}
                }{
                    x : V
                }
                \using\CutRule
            \]
            \justifies
            \wtp[6]{
                \Case\x{\TagReq : \Cut\y{\Call\Split{x,y}}{\Call\Worker\y}}
            }{
                x : \With{ \ann[0]\TagReq : V }
            }
            \using\CaseRule
        \end{prooftree}
    \]
    confirming that the server is well typed.
    \eoe
\end{example}

\subsection{Properties of Well-Typed Processes}
\label{sec:soundness}

We conclude this section presenting a series of results showing the main
properties of well-typed processes, starting from the preservation of typing
under reductions (\ie subject reduction) which underlies most of the subsequent
results. Subject reduction is formulated in a slightly non-standard way because
the typing context may become ``smaller'' -- \ie ``more precise'' -- according
to the subtyping relation $\subt$ after each reduction. This is a consequence of
the formulation of the rule \LinkRule, where the types of $x$ and $y$ need not
be dual to each other, but can be related by subtyping. In the statement of
\Cref{thm:sr}, we write $\subt$ for the pointwise extension of subtyping to
typing contexts.

\begin{restatable}[subject reduction]{theorem}{thmsr}
    \label{thm:sr}
    If $P \red Q$, then $\wtp[n]{P}\ContextC$ implies $\wtp[m]{Q}\ContextD$ for
    some $\ContextD \subt \ContextC$.
\end{restatable}

In general it is not possible to establish an order relationship between $m$
(the measure of the process after the reduction) and $n$ (the measure of the
process before the reduction) because of the rule \ChoiceRule. However, when $P
\red$ it is always possible to reduce $P$ so that its measure strictly
decreases. This is the key reasoning for the forthcoming termination results.

While deadlock freedom (\Cref{def:df}) is usually taken for granted for session
type systems based on linear logic, where it is a straightforward consequence of
cut elimination, its proof for \Calculus is more elaborated because of the deep
cut reductions that model asynchrony. In particular, all of the reduction rules
of \Cref{tab:semantics} except \rulename{r-close} take into account the presence
of message buffers. The commuting conversions dealing with message buffers play
a key role in rearranging the structure of processes so as to guarantee deadlock
freedom in the presence of asynchronous communications.

\begin{restatable}[deadlock freedom]{theorem}{thmdf}
    \label{thm:df}
    If $\wtp[n]{P}\EmptyContext$, then $P$ is deadlock free.
\end{restatable}

\Cref{thm:df} holds for closed processes (those well typed in the empty
context), but its proof is based on a more general \emph{productivity} result
(\Cref{lem:df} in \Cref{sec:proofs-deadlock-freedom}) that applies to open
processes as well. Intuitively, a productive process is always able to reduce or
it is structurally precongruent to another process that exposes some actions on
a free name.

Next we have two termination results. The first one ensures that well-typed
processes are weakly terminating, namely that they have (at least) one finite
run.

\begin{restatable}[weak termination]{theorem}{thmwt}
    \label{thm:wt}
    If $\wtp[n]{P}\Context$, then $P$ is weakly terminating.
\end{restatable}

Weak termination is a liveness property and, as such, its proof is based on a
well-founded argument. The quantity $n$ is one of the components of a
lexicographically ordered pair that serves as basis for the induction. The other
component of the pair measures the depth of the ``unguarded'' top-level
structure of the process which is guaranteed to be finite because of the
guardedness restriction we have assumed on processes in \Cref{sec:language}.
This component is needed to cope with the typing rules \CallRule and \CutRule
for which the measure of the premise(s) is not necessarily smaller than the
measure in the conclusion of the rules.

\Cref{thm:wt} can be strengthened to a termination result for
\emph{deterministic} processes, those that do not contain non-deterministic
choices.

\begin{restatable}[termination]{theorem}{thmt}
    \label{thm:termination}
    If $P$ is deterministic and $\wtp[n]{P}\Context$, then $P$ is terminating.
\end{restatable}

The combination of \Cref{thm:wt} with \Cref{thm:df} also guarantees the absence
of \emph{orphan messages}. Indeed, suppose that $P$ is a well-typed closed
process and $P \wred Q$ where $Q$ contains non-empty buffers. By
\Cref{thm:wt,thm:df} we know that $Q \wred\pcong \Done$. Since in a closed
process messages cannot simply disappear and the only way of consuming them is
by means of the reductions \rulename{r-close}, \rulename{r-fork} and
\rulename{r-select}, we know that all the messages in $Q$ have been received.

Finally, using the proof principle stated in \Cref{thm:proof-principle} we
obtain, as a straightforward corollary of \Cref{thm:wt}, that well-typed
processes are \emph{fairly terminating}~(\Cref{def:ft}).

\begin{corollary}
    \label{cor:ft}
    If $\wtp[n]{P}\Context$, then $P$ is fairly terminating.
\end{corollary}

\section{Subtyping Usage Patterns}
\label{sec:subtyping-usage-patterns}


A typical usage of subtyping is to make sure that processes occurring in
different branches of a non-deterministic choice $P_1 \choice P_2$ or of a tag
input $\Case\x{\Tag_i : P_i}_{i\in I}$ can be typed in the same typing context.
Indeed, the typing rules \ChoiceRule and \CaseRule require each $P_i$ to be
typed in the same typing context (with the exception that in \CaseRule, the type
of $x$ can be different in each branch), but there are cases in which this
requirement is overly restrictive.

\newcommand{\BatchWorker}{\mathit{BatchWorker}}
\newcommand{\SendResults}{\mathit{SendResults}}

As an example, consider a variation of the scenario depicted in
\Cref{sec:introduction} in which the worker behaves either as $\Call\Worker\y$
(as defined as in \Cref{eq:worker}) or as $\Call\BatchWorker\y$ that first
gathers all tasks and then sends back all the results. The $\Call\BatchWorker\y$
version can be modeled by the equations
\begin{align*}
    \Let\BatchWorker\y & = \Case\y{\TagTask : \Call\BatchWorker\y, \TagStop : \Call\SendResults\y} \\
    \Let\SendResults\y & = \Select\y\TagRes.\Call\SendResults\y \choice \Select\y\TagStop.\Close\y
\end{align*}
for which it is easy to obtain a derivation $\wtp[]{\Call\BatchWorker\y}{y :
\dual{S}}$ where $S$ is defined as in \Cref{eq:split-type}.
Ideally, we could express the server using a non-deterministic choice, thus:
\[
    \Case\x{\TagReq : \Cut\y{\Call\Split{x,y}}{(\Call\Worker\y \choice \Call\BatchWorker\y)}}
\]
however, this process turns out to be ill typed because $\Call\Worker\y$ and
$\Call\BatchWorker\y$ use the channel $y$ according to different types.
Moreover, it should be clear that there is no one-size-fits-all session type
that can be associated with $y$.

We can exploit the relation $U \subt \dual{S}$ to make $\Call\Worker\y$ (which
complies with $U$ defined in \cref{eq:worker-type}) look like a process that
complies with $\dual{S}$. In many type systems with subtyping, this is simply
achieved with an application of the \emph{subsumption rule}. The type system we
have presented for \Calculus does not have a subsumption rule, but this rule is
derivable thanks to the formulation of \LinkRule. In our specific example, we
can derive
\[
    \begin{prooftree}
        \[
            \justifies
            \wtp[]{\Link\z\y}{z : \dual{U}, y : \dual{S}}
            \using\LinkRule
        \]
        \[
            \mathstrut\smash\vdots
            \justifies
            \wtp[]{\Call\Worker\z}{z : U}
        \]
        \justifies
        \wtp[]{\Cut\z{\Link\z\y}{\Call\Worker\z}}{y : \dual{S}}
        \using\CutRule
    \end{prooftree}
\]
to obtain an implementation of the streaming worker that complies with
$\dual{S}$. Now we derive
\[
    \wtp[]{\Cut\z{\Link\z\y}{\Call\Worker\z} \choice \Call\BatchWorker\y}{y : \dual{S}}
\]
to obtain a well-typed worker.
Incidentally, this example also illustrates the reason why \Cref{thm:sr} is
formulated so that, when $P \red Q$, the typing context of $Q$ may be smaller
(according to $\subt$) than the typing context of $P$. Indeed we have
$\Cut\z{\Link\z\y}{\Call\Worker\z} \choice \Call\BatchWorker\y \red
\Cut\z{\Link\z\y}{\Call\Worker\z} \red \Call\Worker\y$ and we would not be able
to obtain a typing derivation for $\wtp[]{\Call\Worker\y}{y :
\dual{S}}$.

Another typical usage pattern for links is the implementation of
\emph{delegation}, whereby the endpoint of an existing session is exchanged
through another session. Because of our formulation of \LinkRule, in our type
system this mechanism also enables a form of substitution principle for channels
so that a channel of type $T$ can be safely used where a channel of type $S$ is
expected if $S \subt T$ (the reason why this formulation of the substitution
principle may seem to go in the wrong direction is again due to the fact that we
adopt the viewpoint whereby session types describe processes rather than
channels).

To illustrate this usage pattern, consider the typing derivation
\[
    \begin{prooftree}
        \[
            \[
                \justifies
                \wtp[1]{\Link\y\z}{y : U, z : T}
                \using\LinkRule
            \]
            \[
                \mathstrut\smash\vdots
                \justifies
                \wtp[n]{P}{\Context, x : V}
            \]
            \justifies
            \wtp[n+1]{\Fork\x\y{\Link\y\z}{P}}{\Context, x : \Times{U}{V}, z :
            T} \using\ForkRule \]
        \[
            \[
                \mathstrut\smash\vdots
                \justifies
                \wtp[m]{Q}{\ContextD, x : W, y : S}
            \]
            \justifies
            \wtp[m]{\Join\x\y.Q}{\ContextD, x : \Par{S}{W}}
            \using\JoinRule
        \]
        \justifies
        \wtp[n+m+1]{\Cut\x{\Fork\x\y{\Link\y\z}{P}}{\Join\x\y.Q}}{\ContextC, \ContextD, z : T}
        \using\CutRule
    \end{prooftree}
\]
which concerns the composition of a process $\Fork\x\y{\Link\y\z}P$ that
delegates $z$ to another process $\Join\x\y.Q$. Note that the sender delegates a
channel of type $T$ while the receiver expects to receive a channel of type $S$.
In order for the cut to be well typed we must have
$\cc{\Times{U}{V}}{\Par{S}{W}}$, that is $\correct{U}{S}$ and $\correct{V}{W}$.
The application of \LinkRule requires $\dual{U} \subt T$, that is $\dual{T}
\subt U$ by \Cref{prop:dual-subtype}. From this relation and $\correct{U}{S}$ we
deduce $\correct{\dual{T}}{S}$ by definition of $\subt$. Using
\Cref{thm:correct-subt} we obtain $\dual{T} \subt \dual{S}$, that is $S \subt
T$.
\section{Comparison With Other Subtyping Relations}
\label{sec:subtyping-inclusions}

\newcommand\ett{\text{FFST}\xspace}

In this section we provide a more detailed comparison between $\subt$ and some
connected subtyping relations appeared in the literature.
We show that $\subt$ is coarser than these relations when we focus on the family
of \emph{First-order Fairly-terminating Session Types} (\ett). Fairly
terminating session types are those describing protocols that can always
eventually terminate.
The reason why this family is relevant for us is that \emph{fairly terminating
session types are the only inhabited types in our type system}: since well-typed
processes are fairly terminating (\cref{cor:ft}), they only use sessions that
can always eventually terminate.
For simplicity we also focus on \emph{first-order} session types (without the
multiplicative connectives $\Times{}{}$ and $\Par{}{}$) because the handling of
higher-order communications varies substantially depending on whether the type
system is based on linear logic or not.

To carry out the comparison between $\subt$ and the other subtyping relations we
define two restricted versions of the transition relation $\lred\Action$. We
write $\lmust\Action$ for the relation inductively defined by the
\rulename{must-*} rules and we write $\lind\Action$ for the relation inductively
defined by the \rulename{must-*} and \rulename{may-*} rules. So we have
${\lmust\Action} \subset {\lind\Action} \subset {\lred\Action}$ but the
relation $\lmust\Action$ only allows \emph{immediate transitions} and
$\lind\Action$ differs from $\lred\Action$ because early outputs and late inputs
can only be preceeded by a bounded number of inputs and outputs, respectively.
For example, if $S = \With{\TagA : S, \TagB : \Plus{\TagC : \One}}$ and $T =
\With{\TagA : T, \TagB : \One}$, we have $S \lred{\out\TagC} T$ and $S
\nlind{\out\TagC}$.

We define fairly terminating session types thus:

\begin{definition}
    \label{def:ft-type}
    We say that $S$ is \emph{fairly terminating} if, for every maximal strongly
    fair sequence $S = S_0 \lmust{\Action_1} S_1 \lmust{\Action_2} S_2
    \lmust{\Action_3} \cdots$ starting from $S$, there exists $k$ such that $S_k
    \in \set{\One,\Bot}$.
\end{definition}

Gay and Hole~\cite{GayHole05} introduced the first \emph{synchronous subtyping
relation} for session types. We qualify this relation as ``synchronous'' because
it does not allow any form of output anticipation. However, it supports
unconstrained output covariance and input contravariance as shown in its
characterization below, which uses immediate transitions only.

\begin{definition}
    \label{def:ghsubt}
    We write $\ghsubt$ for the largest relation such that $S \ghsubt T$ implies:
    \begin{enumerate}
    \item\label{gh:pol} either $\positive(S)$ or $\negative(T)$;
    \item\label{gh:inp} if $T \lmust{\inp\MessageType} T'$ then $S
        \lmust{\inp\MessageType} S'$ and $S' \ghsubt T'$;
    \item\label{gh:out} if $S \lmust{\out\MessageType} S'$ then $T
        \lmust{\out\MessageType} T'$ and $S' \ghsubt T'$.
    \end{enumerate}
\end{definition}

Mostrous~\etal~\cite{MostrousYoshidaHonda09,MostrousYoshida15} studied the first
\emph{asynchronous subtypings} for session types.
Chen~\etal~\cite{ChenDezaniScalasYoshida17} subsequently refined the relation of
Mostrous and Yoshida~\cite{MostrousYoshida15} so as to avoid orphan messages.
The resulting relation has been shown to be the largest one included in the
original asynchronous subtyping that is closed under
duality~\cite{BravettiCZ18}.
Using our LTS we can characterize the asynchronous subtyping relations in this
line of research by the following definition.

\begin{definition}
    \label{def:mysubt}
    We write $\mysubt$ for the largest relation such that $S \mysubt T$ implies:
    \begin{enumerate}
    \item\label{emy:pol} either $\positive(S)$ or $\negative(T)$;
    \item\label{emy:inp} if $T \lmust{\inp\MessageType} T'$ then $S
        \lind{\inp\MessageType} S'$ and $S' \mysubt T'$;
    \item\label{emy:out} if $S \lmust{\out\MessageType} S'$ then $T
        \lind{\out\MessageType} T'$ and $S' \mysubt T'$.
    \end{enumerate}
\end{definition}

Like $\ghsubt$, also $\mysubt$ supports unconstrained contravariance of inputs
and covariance of outputs but the clauses \eqref{emy:inp} and \eqref{emy:out}
are more genereous because early/late transitions enable the anticipation of
outputs.
For example, $\Plus{\TagA : \With{\TagB : S, \TagC : \With{\TagD : T}}} \mysubt
\With{\TagB : \Plus{\TagA : S}, \TagC : \With{\TagD : \Plus{\TagA : T}}}$.
However, the use of $\lind\Action$ (instead of $\lred\Action$) in
\cref{def:mysubt} means that anticipated outputs can only be preceeded by a
bounded number of inputs.
If we consider $S = \With{\TagB : S, \TagC : \One}$ and $T = \With{\TagB : T,
\TagC : \Plus{\TagA : \One}}$ we have $\Plus{\TagA : S} \not\mysubt T$ because
$S \lmust{\out\TagA}$ and $T \nlind{\out\TagA}$.
Note that clause \eqref{emy:inp} incorporates the condition of
Chen~\etal~\cite{ChenDezaniScalasYoshida17} that guarantees orphan-message
freedom: if the subtype $S$ anticipates an output and $T$ starts with an input,
\ie $S \lmust{\out\MessageType}$ and $T \lmust{\inp\MessageTypeT}$ for some
$\MessageType$ and $\MessageTypeT$, then clause \eqref{emy:inp} guarantees that
$S$ performs a corresponding late input transition.

The use of $\lind\Action$ is both a strength and a limitation of $\mysubt$. It
is a strength because, just like $\ghsubt$, the \emph{property that $\mysubt$
prevents orphan messages holds regardless of any fairness assumption.} It is a
limitation in the sense that, as we have seen since \cref{sec:introduction},
there are contexts in which it is desirable to work with a coarser asynchronous
subtyping relation.
To overcome this limitation, Bravetti~\etal~\cite{BravettiLangeZavattaro24} have
characterized a \emph{fair asynchronous subtyping relation} which, in our
setting, translates to the use of $\lred\Action$ instead of $\lind\Action$, as
shown below.

\begin{definition}
    \label{def:blzsubt}
    We write $\blzsubt$ for the largest relation such that $S \blzsubt T$ implies:
    \begin{enumerate}
    \item\label{blz:pol} either $\positive(S)$ or $\negative(T)$;
    \item\label{blz:inp} if $T \lmust{\inp\MessageType} T'$ then $S
        \lred{\inp\MessageType} S'$ and $S' \blzsubt T'$;
    \item\label{blz:out} if $S \lmust{\out\MessageType} S'$ then $T
        \lred{\out\MessageType} T'$ and $S' \blzsubt T'$ and $T
        \xlmust{\Actions\out\MessageTypeT}$ implies $S \lmust{\out\MessageTypeT}$
        for every sequence $\Actions$ of input labels.
    \end{enumerate}
\end{definition}

Note that the clauses (\ref{blz:inp}--\ref{blz:out}) of \cref{def:blzsubt} are
not symmetric, implying that $\blzsubt$ is not closed under duality. While
clause \eqref{blz:inp} supports input contravariance, differently from
all the other considered subtypings \emph{clause
\eqref{blz:out} disallows output covariance altogether}.
If we take $S = \Plus{ \TagTask : S, \TagStop : T}$, $T = \With{\TagRes : T,
\TagStop : \Bot}$ and $U = \With{\TagTask : \Plus{\TagRes : U}, \TagStop :
\Plus{\TagStop : \One}}$ from \cref{ex:server-worker-correct}, we have $U
\not\blzsubt \dual{S}$ because $\dual{S}
\xlmust{\inp\TagStop}\xlmust{\out\TagRes}$ but $U
\xlmust{\inp\TagStop}\nlmust{\out\TagRes}$.
This formulation of clause \eqref{blz:out} is due to the fact that $\blzsubt$,
unlike $\ghsubt$, $\mysubt$ and $\subt$ as well, is meant to preserve fair
session termination (at the type level) which requires a controlled form of
output covariance that Bravetti~\etal~\cite{BravettiLangeZavattaro24} have
conservatively approximated in this way.
The next example shows that removing the additional conditions in
clause~\eqref{blz:out} may compromise fair termination.

\begin{example}
    \label{ex:slot-machine}
    Consider the session types $S = \With{\TagPlay : \Plus{\TagWin : S, \TagLose
    : S}, \TagQuit : \One}$ and $U = \Plus{\TagPlay : \With{\TagWin :
    \Plus{\TagQuit : \Bot}, \TagLose : U}}$ where $S$ specifies the behavior of
    a slot machine that allows players to $\TagPlay$ an unbounded number of
    games and $U$ specifies the behavior of a player who plays relentlessly
    until he $\TagWin$s a game. Not only do we have $\correct{S}{U}$, but also
    the composition of $S$ and $U$ is fairly terminating in the sense that, at
    each stage of the interaction between player and slot machine, it is always
    possible to extend the interaction so that the player wins.
    If we now consider an \emph{unfair} implementation of the slot machine such
    that its actual behavior is described by the session type $T =
    \With{\TagPlay : \Plus{\TagLose : T}, \TagQuit : \Bot}$, we have that $T
    \subt S$ and therefore $\correct{T}{U}$ still holds, but the composition of
    $T$ and $U$ is no longer fairly terminating.
    Therefore $T \not\blzsubt S$ and rightly so.
    In our setting, preserving fair session termination at the type level is not
    so important because the type system is able to enforce fair termination at
    the process level anyway thanks to its logical foundation.
    \eoe
\end{example}


\begin{restatable}{theorem}{thminclusions}
    \label{thm:inclusions}
    For \ett session types we have ${\ghsubt} \subset {\mysubt}$ and
    ${\mysubt} \subset {\subt}$ and ${\blzsubt} \subset {\subt}$.
\end{restatable}

\section{Related Work}
\label{sec:related}

\subparagraph*{Asynchronous subtyping.}
The asynchronous subtyping relations for binary session types defined in the
literature~\cite{MostrousYoshida15,ChenDezaniScalasYoshida17,BravettiLangeZavattaro24}
are formulated syntactically using \emph{input contexts} to identify the output
actions in the larger session type that have been anticipated in the smaller
session type. The work of Ghilezan \etal~\cite{GhilezanEtAl23} considers
multiparty sessions: in their setting a smaller session type can anticipate
outputs also w.r.t. outputs sent to a different partner. This is achieved by
considering also appropriate \emph{output contexts}. We follow a different
approach: once we have defined a LTS for session types that captures their
asynchronous semantics (\cref{def:subt,def:cas}), we are able to provide a
characterization of $\subt$ that is essentially the same as the one for
synchronous subtyping (\Cref{def:ghsubt}).

It is worth pointing out that there is no contradiction in the fact that our
subtyping relation turns out to be coarser than others that have been proved
\emph{complete}~\cite{BravettiZ21} or
\emph{precise}~\cite{ChenDezaniScalasYoshida17,GhilezanEtAl23}. The point is
that the completeness of a subtyping relation for session types is always
relative to a notion of correct session composition. Depending on this notion,
the induced subtyping relation may vary. Our subtyping relation is coarser than
$\mysubt$~\cite{MostrousYoshida15,ChenDezaniScalasYoshida17} because it relies
on a fairness assumption that enables a larger degree of output anticipation and
it is coarser than $\blzsubt$~\cite{BravettiLangeZavattaro24} because it is not
meant to preserve fair session termination and therefore it allows (almost)
unconstrained output covariance.

Asynchronous subtyping relations for session types are known to be
undecidable~\cite{LangeY17,BravettiCZ17,BravettiCZ18} and $\subt$ is no
exception (\cref{thm:undecidable,thm:correct-subt}). Despite the proof technique
adopted to prove our undecidability results is similar to those used in other
papers~\cite{BravettiCZ17,BravettiLangeZavattaro24}, we could not directly
derive our results from the undecidability of other relations because $\subt$ is
coarser.

A synchronous version of fair subtyping for session types has been studied by
Padovani~\cite{Padovani16}.
This relation is stricter than $\subt$ because it does not allow early outputs
and it is meant to preserve fair session termination.
Unlike the other fair/asynchronous subtyping relations for session types
(with the exception of the orphan message free subtyping of 
Chen \etal~\cite{ChenDezaniScalasYoshida17}), $\subt$ is
closed under duality (\Cref{prop:dual-subtype}).

\subparagraph*{Fair session termination.}
Type systems ensuring the fair termination of binary sessions have been studied
by Ciccone, Dagnino and
Padovani~\cite{CicconePadovani22A,CicconePadovani22C,Padovani22,CicconeDagninoPadovani24}.
These works are all based on a synchronous communication model.
The technique adopted in this paper for measuring processes is essentially the
same presented by Dagnino and Padovani~\citep{DagninoPadovani24} which in turn
is an elaboration of resource-aware session type
systems~\citep{DasHoffmannPfenning18,DasEtAl21,DasPfenning22}.
Fair termination is a weaker form of
termination~\cite{GrumbergFrancezKatz84,AptFrancezKatz87,Francez86}. Session
type systems enforcing the termination of well-typed processes are typically
based on linear logic~\cite{Wadler14,LindleyMorris16} where termination is a
direct consequence of the cut elimination property of the logic.

\subparagraph*{Logical approaches to asynchrony and subtyping.}
The literature on asynchronous sessions is mostly based on calculi that are not
connected to linear logic and where buffers are modeled explicitly. A notable
exception is the work of DeYoung~\etal~\cite{DeYoungCairesPfenningToninho12}
which presents a type system based on linear logic for an asynchronous calculus
of sessions. In that work, asynchrony is intended as the fact that outputs are
non-blocking, but no anticipation is allowed. The basic idea is the same used in
the encoding of binary session into the linear
$\pi$-calculus~\citep{DardhaGiachinoSangiorgi17}: session communications make
use of explicit continuation channels and outputs can be spawned as soon as they
are produced leaving the sender process free to reduce further.

The deep cut reductions of \Calculus bear some similarities with \emph{deep
inference}~\cite{GuglielmiStrassburger01}, whereby logical rules (including the
cut) may operate on ``deep'' fragments of a proof derivation. At this stage we
are unable to say whether there is a more meaningful connection between our
modeling of asynchrony and deep inference.

The use of links to incorporate subtyping in the type system echoes the approach
studied by Horne and Padovani~\cite{HornePadovani24} who give a \emph{coercive
semantics} to subtyping in a logical setting: a subtyping relation $S \subt T$
corresponds to a forwarder process (\ie a link) that consumes a channel of type
$\dual{S}$ and operates on a channel of type $T$. In \Calculus, the availability
of the native link is fundamental since buffers may grow arbitrarily and the
forwarder process is neither finitely representable nor finitely measurable in
general.
\section{Concluding Remarks}
\label{sec:conclusion}

Sessions are meant to enable the compositional enforcement of safety and
liveness properties of communicating processes. However, most liveness
properties concerning one particular session -- like the fact that a given
message is eventually produced or consumed -- can only be ensured if one makes
the assumption that every other session eventually terminates. For this reason,
the eventual termination of sessions is of primary importance. From the property
that a session eventually terminates, other properties of interest usually
follow ``for free'' as a straightforward consequence of session typing.

In this paper we have introduced a new theory of asynchronous session types that
ensures the eventual termination of every session under a full fairness
assumption. As a consequence, every produced message is guaranteed to be
eventually consumed and every process waiting for a message is guaranteed to
eventually receive one.
This theory improves previous ones in a number of ways: we define a novel fair
asynchronous subtyping relation that supports the anticipation of outputs before
an unbounded number of inputs, output covariance and that is closed under
duality; we extend the soundness properties to multiple, possibly interleaved
and dynamically created sessions; we base our theory on a calculus of
asynchronous processes whose features, types and typing rules are rooted in
linear logic.

Clearly, the undecidability results about the correctness of session composition
and of fair asynchronous subtyping are an obstacle to the applicability of the
presented theory. While decidable approximations of these notions have been
presented in the
literature~\cite{BravettiCLYZ21,BravettiLangeZavattaro24,BocchiKM24}, they cover
relatively small families of session types. In future work we envision the
possibility to define sound (but necessarily incomplete) algorithms for checking
correct composition and subtyping inspired by the coinductive characterizations
of correct session composition and of fair asynchronous subtyping introduced in
this paper.

We also conjecture that our semantic approach to subtyping based on an LTS with
early/late transitions scales smoothly to the multiparty setting where it may
lead to simpler characterizations of (fair) asychronous subtyping relations. It
may be interesting to verify the validity of this conjecture in future work.

\bibliographystyle{plainurl}
\bibliography{main}

\begin{thebibliography}{10}

\bibitem{AnconaDagninoZucca17}
Davide Ancona, Francesco Dagnino, and Elena Zucca.
\newblock Generalizing inference systems by coaxioms.
\newblock In Hongseok Yang, editor, {\em Programming Languages and Systems - 26th European Symposium on Programming, {ESOP} 2017, Held as Part of the European Joint Conferences on Theory and Practice of Software, {ETAPS} 2017, Uppsala, Sweden, April 22-29, 2017, Proceedings}, volume 10201 of {\em Lecture Notes in Computer Science}, pages 29--55. Springer, 2017.
\newblock \href {https://doi.org/10.1007/978-3-662-54434-1_2} {\path{doi:10.1007/978-3-662-54434-1_2}}.

\bibitem{AptFrancezKatz87}
Krzysztof~R. Apt, Nissim Francez, and Shmuel Katz.
\newblock Appraising fairness in languages for distributed programming.
\newblock In {\em Conference Record of the Fourteenth Annual {ACM} Symposium on Principles of Programming Languages, Munich, Germany, January 21-23, 1987}, pages 189--198. {ACM} Press, 1987.
\newblock \href {https://doi.org/10.1145/41625.41642} {\path{doi:10.1145/41625.41642}}.

\bibitem{BaeldeDoumaneSaurin16}
David Baelde, Amina Doumane, and Alexis Saurin.
\newblock Infinitary proof theory: the multiplicative additive case.
\newblock In Jean{-}Marc Talbot and Laurent Regnier, editors, {\em 25th {EACSL} Annual Conference on Computer Science Logic, {CSL} 2016, August 29 - September 1, 2016, Marseille, France}, volume~62 of {\em LIPIcs}, pages 42:1--42:17. Schloss Dagstuhl - Leibniz-Zentrum f{\"{u}}r Informatik, 2016.
\newblock \href {https://doi.org/10.4230/LIPICS.CSL.2016.42} {\path{doi:10.4230/LIPICS.CSL.2016.42}}.

\bibitem{BocchiKM24}
Laura Bocchi, Andy King, and Maurizio Murgia.
\newblock Asynchronous subtyping by trace relaxation.
\newblock In Bernd Finkbeiner and Laura Kov{\'{a}}cs, editors, {\em Tools and Algorithms for the Construction and Analysis of Systems - 30th International Conference, {TACAS} 2024, Held as Part of the European Joint Conferences on Theory and Practice of Software, {ETAPS} 2024, Luxembourg City, Luxembourg, April 6-11, 2024, Proceedings, Part {I}}, volume 14570 of {\em Lecture Notes in Computer Science}, pages 207--226. Springer, 2024.
\newblock \href {https://doi.org/10.1007/978-3-031-57246-3_12} {\path{doi:10.1007/978-3-031-57246-3_12}}.

\bibitem{BravettiCLYZ21}
Mario Bravetti, Marco Carbone, Julien Lange, Nobuko Yoshida, and Gianluigi Zavattaro.
\newblock A sound algorithm for asynchronous session subtyping and its implementation.
\newblock {\em Log. Methods Comput. Sci.}, 17(1), 2021.
\newblock URL: \url{https://lmcs.episciences.org/7238}.

\bibitem{BravettiCZ17}
Mario Bravetti, Marco Carbone, and Gianluigi Zavattaro.
\newblock Undecidability of asynchronous session subtyping.
\newblock {\em Inf. Comput.}, 256:300--320, 2017.
\newblock URL: \url{https://doi.org/10.1016/j.ic.2017.07.010}, \href {https://doi.org/10.1016/J.IC.2017.07.010} {\path{doi:10.1016/J.IC.2017.07.010}}.

\bibitem{BravettiCZ18}
Mario Bravetti, Marco Carbone, and Gianluigi Zavattaro.
\newblock On the boundary between decidability and undecidability of asynchronous session subtyping.
\newblock {\em Theor. Comput. Sci.}, 722:19--51, 2018.
\newblock URL: \url{https://doi.org/10.1016/j.tcs.2018.02.010}, \href {https://doi.org/10.1016/J.TCS.2018.02.010} {\path{doi:10.1016/J.TCS.2018.02.010}}.

\bibitem{BravettiLangeZavattaro24}
Mario Bravetti, Julien Lange, and Gianluigi Zavattaro.
\newblock Fair asynchronous session subtyping.
\newblock {\em Log. Methods Comput. Sci.}, 20(4), 2024.
\newblock \href {https://doi.org/10.46298/LMCS-20(4:5)2024} {\path{doi:10.46298/LMCS-20(4:5)2024}}.

\bibitem{BravettiZ21}
Mario Bravetti and Gianluigi Zavattaro.
\newblock Asynchronous session subtyping as communicating automata refinement.
\newblock {\em Softw. Syst. Model.}, 20(2):311--333, 2021.
\newblock URL: \url{https://doi.org/10.1007/s10270-020-00838-x}, \href {https://doi.org/10.1007/S10270-020-00838-X} {\path{doi:10.1007/S10270-020-00838-X}}.

\bibitem{CairesPfenningToninho16}
Lu{\'{\i}}s Caires, Frank Pfenning, and Bernardo Toninho.
\newblock Linear logic propositions as session types.
\newblock {\em Math. Struct. Comput. Sci.}, 26(3):367--423, 2016.
\newblock \href {https://doi.org/10.1017/S0960129514000218} {\path{doi:10.1017/S0960129514000218}}.

\bibitem{ChenDezaniScalasYoshida17}
Tzu{-}Chun Chen, Mariangiola Dezani{-}Ciancaglini, Alceste Scalas, and Nobuko Yoshida.
\newblock On the preciseness of subtyping in session types.
\newblock {\em Log. Methods Comput. Sci.}, 13(2), 2017.
\newblock \href {https://doi.org/10.23638/LMCS-13(2:12)2017} {\path{doi:10.23638/LMCS-13(2:12)2017}}.

\bibitem{CicconeDagninoPadovani22}
Luca Ciccone, Francesco Dagnino, and Luca Padovani.
\newblock Fair termination of multiparty sessions.
\newblock In Karim Ali and Jan Vitek, editors, {\em 36th European Conference on Object-Oriented Programming, {ECOOP} 2022, June 6-10, 2022, Berlin, Germany}, volume 222 of {\em LIPIcs}, pages 26:1--26:26. Schloss Dagstuhl - Leibniz-Zentrum f{\"{u}}r Informatik, 2022.
\newblock \href {https://doi.org/10.4230/LIPICS.ECOOP.2022.26} {\path{doi:10.4230/LIPICS.ECOOP.2022.26}}.

\bibitem{CicconeDagninoPadovani24}
Luca Ciccone, Francesco Dagnino, and Luca Padovani.
\newblock Fair termination of multiparty sessions.
\newblock {\em J. Log. Algebraic Methods Program.}, 139:100964, 2024.
\newblock \href {https://doi.org/10.1016/J.JLAMP.2024.100964} {\path{doi:10.1016/J.JLAMP.2024.100964}}.

\bibitem{CicconePadovani22A}
Luca Ciccone and Luca Padovani.
\newblock Fair termination of binary sessions.
\newblock {\em Proc. {ACM} Program. Lang.}, 6({POPL}):1--30, 2022.
\newblock \href {https://doi.org/10.1145/3498666} {\path{doi:10.1145/3498666}}.

\bibitem{CicconePadovani22B}
Luca Ciccone and Luca Padovani.
\newblock Inference systems with corules for combined safety and liveness properties of binary session types.
\newblock {\em Log. Methods Comput. Sci.}, 18(3), 2022.
\newblock \href {https://doi.org/10.46298/LMCS-18(3:27)2022} {\path{doi:10.46298/LMCS-18(3:27)2022}}.

\bibitem{CicconePadovani22C}
Luca Ciccone and Luca Padovani.
\newblock An infinitary proof theory of linear logic ensuring fair termination in the linear {\(\pi\)}-calculus.
\newblock In Bartek Klin, Slawomir Lasota, and Anca Muscholl, editors, {\em 33rd International Conference on Concurrency Theory, {CONCUR} 2022, September 12-16, 2022, Warsaw, Poland}, volume 243 of {\em LIPIcs}, pages 36:1--36:18. Schloss Dagstuhl - Leibniz-Zentrum f{\"{u}}r Informatik, 2022.
\newblock \href {https://doi.org/10.4230/LIPICS.CONCUR.2022.36} {\path{doi:10.4230/LIPICS.CONCUR.2022.36}}.

\bibitem{Courcelle83}
Bruno Courcelle.
\newblock Fundamental properties of infinite trees.
\newblock {\em Theor. Comput. Sci.}, 25:95--169, 1983.
\newblock \href {https://doi.org/10.1016/0304-3975(83)90059-2} {\path{doi:10.1016/0304-3975(83)90059-2}}.

\bibitem{Dagnino19}
Francesco Dagnino.
\newblock Coaxioms: flexible coinductive definitions by inference systems.
\newblock {\em Log. Methods Comput. Sci.}, 15(1), 2019.
\newblock \href {https://doi.org/10.23638/LMCS-15(1:26)2019} {\path{doi:10.23638/LMCS-15(1:26)2019}}.

\bibitem{DagninoPadovani24}
Francesco Dagnino and Luca Padovani.
\newblock small caps: An infinitary linear logic for a calculus of pure sessions.
\newblock In Alessandro Bruni, Alberto Momigliano, Matteo Pradella, Matteo Rossi, and James Cheney, editors, {\em Proceedings of the 26th International Symposium on Principles and Practice of Declarative Programming, {PPDP} 2024, Milano, Italy, September 9-11, 2024}, pages 4:1--4:13. {ACM}, 2024.
\newblock \href {https://doi.org/10.1145/3678232.3678234} {\path{doi:10.1145/3678232.3678234}}.

\bibitem{DardhaGiachinoSangiorgi17}
Ornela Dardha, Elena Giachino, and Davide Sangiorgi.
\newblock Session types revisited.
\newblock {\em Inf. Comput.}, 256:253--286, 2017.
\newblock \href {https://doi.org/10.1016/j.ic.2017.06.002} {\path{doi:10.1016/j.ic.2017.06.002}}.

\bibitem{DasEtAl21}
Ankush Das, Stephanie Balzer, Jan Hoffmann, Frank Pfenning, and Ishani Santurkar.
\newblock Resource-aware session types for digital contracts.
\newblock In {\em 34th {IEEE} Computer Security Foundations Symposium, {CSF} 2021, Dubrovnik, Croatia, June 21-25, 2021}, pages 1--16. {IEEE}, 2021.
\newblock \href {https://doi.org/10.1109/CSF51468.2021.00004} {\path{doi:10.1109/CSF51468.2021.00004}}.

\bibitem{DasHoffmannPfenning18}
Ankush Das, Jan Hoffmann, and Frank Pfenning.
\newblock Work analysis with resource-aware session types.
\newblock In Anuj Dawar and Erich Gr{\"{a}}del, editors, {\em Proceedings of the 33rd Annual {ACM/IEEE} Symposium on Logic in Computer Science, {LICS} 2018, Oxford, UK, July 09-12, 2018}, pages 305--314. {ACM}, 2018.
\newblock \href {https://doi.org/10.1145/3209108.3209146} {\path{doi:10.1145/3209108.3209146}}.

\bibitem{DasPfenning22}
Ankush Das and Frank Pfenning.
\newblock Rast: {A} language for resource-aware session types.
\newblock {\em Log. Methods Comput. Sci.}, 18(1), 2022.
\newblock \href {https://doi.org/10.46298/LMCS-18(1:9)2022} {\path{doi:10.46298/LMCS-18(1:9)2022}}.

\bibitem{DeYoungCairesPfenningToninho12}
Henry DeYoung, Lu{\'{\i}}s Caires, Frank Pfenning, and Bernardo Toninho.
\newblock Cut reduction in linear logic as asynchronous session-typed communication.
\newblock In Patrick C{\'{e}}gielski and Arnaud Durand, editors, {\em Computer Science Logic (CSL'12) - 26th International Workshop/21st Annual Conference of the EACSL, {CSL} 2012, September 3-6, 2012, Fontainebleau, France}, volume~16 of {\em LIPIcs}, pages 228--242. Schloss Dagstuhl - Leibniz-Zentrum f{\"{u}}r Informatik, 2012.
\newblock \href {https://doi.org/10.4230/LIPICS.CSL.2012.228} {\path{doi:10.4230/LIPICS.CSL.2012.228}}.

\bibitem{Doumane17}
Amina Doumane.
\newblock {\em On the infinitary proof theory of logics with fixed points. (Th{\'{e}}orie de la d{\'{e}}monstration infinitaire pour les logiques {\`{a}} points fixes)}.
\newblock PhD thesis, Paris Diderot University, France, 2017.
\newblock URL: \url{https://tel.archives-ouvertes.fr/tel-01676953}.

\bibitem{Francez86}
Nissim Francez.
\newblock {\em Fairness}.
\newblock Texts and Monographs in Computer Science. Springer, 1986.
\newblock \href {https://doi.org/10.1007/978-1-4612-4886-6} {\path{doi:10.1007/978-1-4612-4886-6}}.

\bibitem{Gay16}
Simon~J. Gay.
\newblock Subtyping supports safe session substitution.
\newblock In Sam Lindley, Conor McBride, Philip~W. Trinder, and Donald Sannella, editors, {\em A List of Successes That Can Change the World - Essays Dedicated to Philip Wadler on the Occasion of His 60th Birthday}, volume 9600 of {\em Lecture Notes in Computer Science}, pages 95--108. Springer, 2016.
\newblock \href {https://doi.org/10.1007/978-3-319-30936-1_5} {\path{doi:10.1007/978-3-319-30936-1_5}}.

\bibitem{GayHole05}
Simon~J. Gay and Malcolm Hole.
\newblock Subtyping for session types in the pi calculus.
\newblock {\em Acta Informatica}, 42(2-3):191--225, 2005.
\newblock \href {https://doi.org/10.1007/S00236-005-0177-Z} {\path{doi:10.1007/S00236-005-0177-Z}}.

\bibitem{GhilezanEtAl23}
Silvia Ghilezan, Jovanka Pantovic, Ivan Prokic, Alceste Scalas, and Nobuko Yoshida.
\newblock Precise subtyping for asynchronous multiparty sessions.
\newblock {\em {ACM} Trans. Comput. Log.}, 24(2):14:1--14:73, 2023.
\newblock \href {https://doi.org/10.1145/3568422} {\path{doi:10.1145/3568422}}.

\bibitem{GrumbergFrancezKatz84}
Orna Grumberg, Nissim Francez, and Shmuel Katz.
\newblock Fair termination of communicating processes.
\newblock In {\em Proceedings of the Third Annual ACM Symposium on Principles of Distributed Computing}, PODC '84, page 254–265, New York, NY, USA, 1984. Association for Computing Machinery.
\newblock \href {https://doi.org/10.1145/800222.806752} {\path{doi:10.1145/800222.806752}}.

\bibitem{GuglielmiStrassburger01}
Alessio Guglielmi and Lutz Stra{\ss}burger.
\newblock Non-commutativity and {MELL} in the calculus of structures.
\newblock In Laurent Fribourg, editor, {\em Computer Science Logic, 15th International Workshop, {CSL} 2001. 10th Annual Conference of the EACSL, Paris, France, September 10-13, 2001, Proceedings}, volume 2142 of {\em Lecture Notes in Computer Science}, pages 54--68. Springer, 2001.
\newblock \href {https://doi.org/10.1007/3-540-44802-0_5} {\path{doi:10.1007/3-540-44802-0_5}}.

\bibitem{Honda93}
Kohei Honda.
\newblock Types for dyadic interaction.
\newblock In Eike Best, editor, {\em {CONCUR} '93, 4th International Conference on Concurrency Theory, Hildesheim, Germany, August 23-26, 1993, Proceedings}, volume 715 of {\em Lecture Notes in Computer Science}, pages 509--523. Springer, 1993.
\newblock \href {https://doi.org/10.1007/3-540-57208-2_35} {\path{doi:10.1007/3-540-57208-2_35}}.

\bibitem{HondaVasconcelosKubo98}
Kohei Honda, Vasco~Thudichum Vasconcelos, and Makoto Kubo.
\newblock Language primitives and type discipline for structured communication-based programming.
\newblock In Chris Hankin, editor, {\em Programming Languages and Systems - ESOP'98, 7th European Symposium on Programming, Held as Part of the European Joint Conferences on the Theory and Practice of Software, ETAPS'98, Lisbon, Portugal, March 28 - April 4, 1998, Proceedings}, volume 1381 of {\em Lecture Notes in Computer Science}, pages 122--138. Springer, 1998.
\newblock \href {https://doi.org/10.1007/BFB0053567} {\path{doi:10.1007/BFB0053567}}.

\bibitem{HornePadovani24}
Ross Horne and Luca Padovani.
\newblock A logical account of subtyping for session types.
\newblock {\em J. Log. Algebraic Methods Program.}, 141:100986, 2024.
\newblock \href {https://doi.org/10.1016/J.JLAMP.2024.100986} {\path{doi:10.1016/J.JLAMP.2024.100986}}.

\bibitem{HuttelEtAl16}
Hans H{\"{u}}ttel, Ivan Lanese, Vasco~T. Vasconcelos, Lu{\'{\i}}s Caires, Marco Carbone, Pierre{-}Malo Deni{\'{e}}lou, Dimitris Mostrous, Luca Padovani, Ant{\'{o}}nio Ravara, Emilio Tuosto, Hugo~Torres Vieira, and Gianluigi Zavattaro.
\newblock Foundations of session types and behavioural contracts.
\newblock {\em {ACM} Comput. Surv.}, 49(1):3:1--3:36, 2016.
\newblock \href {https://doi.org/10.1145/2873052} {\path{doi:10.1145/2873052}}.

\bibitem{KozenBook}
D.~Kozen.
\newblock {\em Automata and Computability}.
\newblock Springer, New York, 1997.

\bibitem{LangeY17}
Julien Lange and Nobuko Yoshida.
\newblock On the undecidability of asynchronous session subtyping.
\newblock In Javier Esparza and Andrzej~S. Murawski, editors, {\em Foundations of Software Science and Computation Structures - 20th International Conference, {FOSSACS} 2017, Held as Part of the European Joint Conferences on Theory and Practice of Software, {ETAPS} 2017, Uppsala, Sweden, April 22-29, 2017, Proceedings}, volume 10203 of {\em Lecture Notes in Computer Science}, pages 441--457, 2017.
\newblock \href {https://doi.org/10.1007/978-3-662-54458-7_26} {\path{doi:10.1007/978-3-662-54458-7_26}}.

\bibitem{LindleyMorris16}
Sam Lindley and J.~Garrett Morris.
\newblock Talking bananas: structural recursion for session types.
\newblock In Jacques Garrigue, Gabriele Keller, and Eijiro Sumii, editors, {\em Proceedings of the 21st {ACM} {SIGPLAN} International Conference on Functional Programming, {ICFP} 2016, Nara, Japan, September 18-22, 2016}, pages 434--447. {ACM}, 2016.
\newblock \href {https://doi.org/10.1145/2951913.2951921} {\path{doi:10.1145/2951913.2951921}}.

\bibitem{LiskovWing94}
Barbara Liskov and Jeannette~M. Wing.
\newblock A behavioral notion of subtyping.
\newblock {\em {ACM} Trans. Program. Lang. Syst.}, 16(6):1811--1841, 1994.
\newblock \href {https://doi.org/10.1145/197320.197383} {\path{doi:10.1145/197320.197383}}.

\bibitem{MostrousYoshida15}
Dimitris Mostrous and Nobuko Yoshida.
\newblock Session typing and asynchronous subtyping for the higher-order {\(\pi\)}-calculus.
\newblock {\em Inf. Comput.}, 241:227--263, 2015.
\newblock \href {https://doi.org/10.1016/J.IC.2015.02.002} {\path{doi:10.1016/J.IC.2015.02.002}}.

\bibitem{MostrousYoshidaHonda09}
Dimitris Mostrous, Nobuko Yoshida, and Kohei Honda.
\newblock Global principal typing in partially commutative asynchronous sessions.
\newblock In Giuseppe Castagna, editor, {\em Programming Languages and Systems, 18th European Symposium on Programming, {ESOP} 2009, Held as Part of the Joint European Conferences on Theory and Practice of Software, {ETAPS} 2009, York, UK, March 22-29, 2009. Proceedings}, volume 5502 of {\em Lecture Notes in Computer Science}, pages 316--332. Springer, 2009.
\newblock \href {https://doi.org/10.1007/978-3-642-00590-9_23} {\path{doi:10.1007/978-3-642-00590-9_23}}.

\bibitem{Padovani16}
Luca Padovani.
\newblock Fair subtyping for multi-party session types.
\newblock {\em Math. Struct. Comput. Sci.}, 26(3):424--464, 2016.
\newblock \href {https://doi.org/10.1017/S096012951400022X} {\path{doi:10.1017/S096012951400022X}}.

\bibitem{Padovani22}
Luca Padovani.
\newblock On the fair termination of client-server sessions.
\newblock In Delia Kesner and Pierre{-}Marie P{\'{e}}drot, editors, {\em 28th International Conference on Types for Proofs and Programs, {TYPES} 2022, June 20-25, 2022, LS2N, University of Nantes, France}, volume 269 of {\em LIPIcs}, pages 5:1--5:21. Schloss Dagstuhl - Leibniz-Zentrum f{\"{u}}r Informatik, 2022.
\newblock \href {https://doi.org/10.4230/LIPICS.TYPES.2022.5} {\path{doi:10.4230/LIPICS.TYPES.2022.5}}.

\bibitem{GlabbeekHofner19}
Rob van Glabbeek and Peter H{\"{o}}fner.
\newblock Progress, justness, and fairness.
\newblock {\em {ACM} Comput. Surv.}, 52(4):69:1--69:38, 2019.
\newblock \href {https://doi.org/10.1145/3329125} {\path{doi:10.1145/3329125}}.

\bibitem{Wadler14}
Philip Wadler.
\newblock Propositions as sessions.
\newblock {\em J. Funct. Program.}, 24(2-3):384--418, 2014.
\newblock \href {https://doi.org/10.1017/S095679681400001X} {\path{doi:10.1017/S095679681400001X}}.

\end{thebibliography}

\appendix
\section{Supplement to Section \ref{sec:language}}
\label{sec:proofs-language}

\begin{lemma}
    \label{lem:feasibility}
    Every finite reduction sequence of $P$ can be extended to a fair run of $P$.
\end{lemma}
\begin{proof}
    Consider a finite reduction sequence $(P_0,\dots,P_n)$ of $P$. We
    distinguish two possibilities.
    If $P$ is weakly terminating, then there exists a run $(P_n, \dots, Q)$ of
    $P_n$ such that $Q \nred$ and it suffices to consider the run
    $(P_0,\dots,P_n,\dots,Q)$ of $P$ which is finite, hence fair.
    If $P_n$ is diverging, then there exists an infinite run
    $(P_n,P_{n+1},\dots)$ in which no process is weakly terminating. Hence, it
    suffices to consider the run $(P_0,\dots,P_n,P_{n+1},\dots)$ of $P$ which
    contains finitely many weakly terminating states.
\end{proof}

\thmproofprinciple*
\begin{proof}
    ($\Rightarrow$)
    Consider the reduction sequence $(P,\dots,Q)$ corresponding to $P \wred Q$.
    From \Cref{lem:feasibility} we can extend this reduction sequence to a fair
    run $(P,\dots,Q,Q_1,\dots)$ of $P$. From the hypothesis that $P$ is fairly
    terminating we deduce that such run is finite. Let $Q_n$ be the last process
    in this run. Now $Q \wred Q_n \nred$, namely $Q$ is weakly terminating.

    ($\Leftarrow$)
    Consider a fair run $(P_0,\dots)$ of $P$ and suppose by contradiction that
    it is infinite. Every process $P_i$ in this run has the property $P \wred
    P_i$. From the hypothesis we deduce that every $P_i$ is weakly terminating.
    But then the run contains infinitely many weakly terminating processes,
    which contradicts the hypothesis that the run was fair.
\end{proof}

\section{Supplement to Section~\ref{sec:session-types}}
\label{sec:proofs-session-types}

\thmfas*
\begin{proof}    
    \newcommand{\rank}[1]{\mathsf{rank}(#1)}
    We prove the result when $\Action$ is an input transition
    $\inp\MessageType$, since the case when $\Action$ is an output transition is
    analogous.

    ($\Rightarrow$)
    By definition of GIS~\citep{AnconaDagninoZucca17,Dagnino19}, $S
    \lred{\inp\MessageType}$ implies $S \ired{\inp\MessageType}$. In the
    derivation of $S \lred{\inp\MessageType}$ we can rank each session type $S'$
    occurring in a judgment $S' \lred{\inp\MessageType}$ by the size of the
    (finite) proof of $S' \ired{\inp\MessageType}$, which we denote by
    $\rank{S'}$.
    Now suppose that there is an infinite, strongly fair sequence $S = S_1
    \lred{\out\MessageType_1} S_2 \lred{\out\MessageType_2} \cdots$ of immediate
    output transitions starting from $S$ and observe that each $S_i$ is a
    sub-tree of $S$.
    Since $S$ is regular it has finitely many distinct sub-trees. Therefore, at
    least one among the $S_i$, say $S_{k_1}$, occurs infinitely often in the
    sequence $S_1S_2\cdots$.
    Observe that $S_{k_1}$ occurs in a judgment $S_{k_1}
    \lred{\inp\MessageType}$ of the aforementioned derivation of $S
    \lred{\inp\MessageType}$.
    The judgment $S_{k_1} \lred{\inp\MessageType}$ cannot have been derived with
    an application of \BotT, \WithT or \ParT, because $S_{k_1}$ performs an
    immediate output transition. Therefore, $S_{k_1} \lred{\inp\MessageType}$
    must have been derived by either \APlusT or \ATimesT. Moreover, at least one
    of the premises in the inductive derivation of $S_{k_1}
    \ired{\inp\MessageType}$ is a judgment $S_{k_2} \ired{\inp\MessageType}$
    where $\rank{S_{k_2}} < \rank{S_{k_1}}$.
    Note that $S_{k_1} \lred{\out\MessageTypeT} S_{k_2}$ for some
    $\MessageTypeT$. Hence, from the hypothesis that the sequence of transition
    is strongly fair, we deduce that $S_{k_2}$ occurs infinitely often as well.
    By repeating this argument we are able to build an infinite sequence
    $S_{k_1}S_{k_2}\cdots$ of session types with strictly decreasing ranks,
    which is absurd. We conclude that every strongly fair sequence of immediate
    output transitions starting from $S$ is finite.
    Let $S = S_1 \lred{\out\MessageType_1} S_2 \lred{\out\MessageType_2} \cdots
    S_n$ be one of such sequences and let $T \eqdef S_n$. Since this sequence is
    maximal, it must be the case that $T$ does not perform output transitions.
    Also, from the hypothesis $S \lred{\inp\MessageType}$ it is easy to prove,
    by induction on $n$, that $T \lred{\inp\MessageType}$.
    Therefore, $T$ is a session type that does not perform immediate output
    transitions and that performs a $\inp\MessageType$ transition. By inspection
    of \Cref{tab:lts} this means that the transition $T \lred{\inp\Message}$ has
    been derived using an axiom.

    ($\Leftarrow$)
    First we show how to build a (possibly infinite) derivation for $S
    \lred{\inp\MessageType}$ using the singly-lined rules in \Cref{tab:lts}
    under the hypothesis that every strongly fair sequence of immediate output
    transitions starting from $S$ is finite.
    We reason by cases on the shape of $S$ to build the derivation one layer at
    a time.
    \begin{itemize}
    \item ($S = \One$)
        This case is impossible because $\One \lred{\out\unit} \One
        \lred{\out\unit} \cdots$ is an infinite strongly fair sequence of
        immediate output transitions starting from $\One$.
    \item ($S = \Bot$)
        Then $\MessageType = \unit$ and we conclude with an application of
        \BotT.
    \item ($S = \Plus{\Tag_i : S_i}_{i\in I}$)
        Then $S \lred{\out\Tag_i} S_i$ for every $i\in I$. Since every strongly
        fair sequence of immediate output transitions starting from $S$ is
        finite, then the same holds for $S_i$ for every $i\in I$.
        We conclude with an application of \APlusT.
    \item ($S = \With{\Tag_i : S_i}_{i\in I}$)
        Then $\MessageType = \inp\Tag_i$ and we conclude with an application of
        \WithT.
    \item ($S = \Times{T'}{S'}$)
        Then $S \lred{\out\dual{T'}} S'$.
        Since every strongly fair sequence of immediate output transitions
        starting from $S$ is finite, then the same holds for $S'$.
        We conclude with an application of \ATimesT.
    \item ($S = \Par{T'}{S'}$)
        Then $\MessageType = \inp T'$ and we conclude with an application of
        \ParT.
    \end{itemize}

    Next we have to show that, for each judgment of the form $S_i
    \lred{\inp\MessageType}$ that occurs in the aforementioned derivation, we
    are able to build a \emph{finite} derivation for $S_i
    \ired{\inp\MessageType}$.
    Recall that, for each $S_i$, every strongly fair sequence of immediate
    output transitions starting from $S_i$ is finite.
    The result follows by induction on the length of one of such strongly fair
    sequences, knowing that at least one does exist.\LP{Qui servirebbe il
    risultato di feasibility}
    The only interesting case is when $S_i = \Plus{\Tag_j : S_j}_{j\in J}$. Then
    the strongly fair sequence of immediate output transitions begins with $S_i
    \lred{\out\Tag_j} S_j$ for some $j\in J$.
    Using the induction hypothesis we build a derivation of $S_j
    \lred{\inp\MessageType}$ and then we conclude with an application of
    \CoPlusT.
\end{proof}

\propdiamond*
\begin{proof}
    Since $S$ is either positive or negative, then one among the transitions $S
    \lred{\inp\MessageTypeS} S'$ and $S \lred{\out\MessageTypeT} S''$ must be
    immediate. We discuss the case in which $S \lred{\out\MessageTypeT} S''$ is
    immediate, the other being symmetric. We proceed by cases on the rule used
    to derive $S \lred{\out\MessageTypeT} S''$.

    \proofcase\OneT
    Then $S = \One$, but this case is impossible because $\One$ does not perform
    input transitions.

    \proofcase\PlusT
    Then $S = \Plus{\Tag_i : S_i}_{i\in I}$ and $\MessageTypeT = \Tag_k$ and
    $S'' = S_k$ for some $k\in I$.
    From the hypothesis $S \lred{\inp\MessageTypeS} S'$ and \APlusT we deduce
    $S_i \lred{\inp\MessageTypeS} T_i$ for every $i\in I$ and $S' = \Plus{\Tag_i
    : T_i}_{i\in I}$.
    We conclude by taking $T \eqdef T_k$ and observing that $S'
    \lred{\out\Tag_k} T_k$.

    \proofcase\TimesT
    Then $S = \Times{T'}{S''}$ and $\MessageTypeT = T'$.
    From the hypothesis $S \lred{\inp\MessageTypeS} S'$ and \ATimesT we deduce
    $S'' \lred{\inp\MessageTypeS} T$ and $S' = \Times{T'}{T}$.
    We conclude by observing that $S' \lred{\out T'} T$.
\end{proof}

\section{Supplement to Section~\ref{sec:subtyping}}
\label{sec:proofs-subtyping}

\subsection{Properties of Correctness and Subtyping}

\propdualcorrect*
\begin{proof}
It suffices to show that $\set{ (S, \dual{S}) }$ is a coinductive asynchronous
subtyping.
\end{proof}

\thmsubt*
\begin{proof}
    First we prove that if $\srel$ is a coinductive asynchronous subtyping such
    that $(S,T) \in \srel$ and $\correct{R}{T}$, then $\correct{R}{S}$. It
    suffices to show that $\rrel \eqdef \set{ (R, S) \mid \exists T. (S,T) \in
    \srel \wedge \correct{R}{T}}$ is a correct asynchronous composition. Let
    $(R,S) \in \rrel$. Then there exists $T$ such that $(S,T) \in \srel$ and
    $\correct{R}{T}$. We prove the conditions of \Cref{def:cac} in order:
    \begin{enumerate}
    \item From $\correct{R}{T}$ and \rcac{pp} we deduce either $\positive(R)$ or
        $\positive(T)$. From $(S,T) \in \srel$ and \rcas{pol} we deduce either
        $\positive(S)$ or $\negative(T)$. Then either $\positive(R)$ or
        $\positive(S)$.
    \item Suppose $S \lred{\out\MessageType} S'$ and
        $\MessageType\in\set\unit\cup\TagSet$.
        From \rcas{out} we deduce $T \lred{\out\MessageType} T'$ and $(S',T')
        \in \srel$.
        From \rcac{io} we deduce $R \lred{\inp\MessageType} R'$ and
        $\correct{R'}{T'}$. Then $(R',S') \in \rrel$ by definition of $\rrel$.
    \item Suppose $R \lred{\out\MessageType} R'$ and $\MessageType \in \set\unit
        \cup \TagSet$.
        From \rcac{oi} we deduce $T \lred{\inp\MessageType} T'$ and
        $\correct{R'}{T'}$.
        From \rcas{inp} we deduce $S \lred{\inp\MessageType} S'$ and $(S',T')
        \in \srel$. Then $(R',S') \in \rrel$ by definition of $\rrel$.
    \item Suppose $S \xlred{\out S_1} S_2$.
        From \rcas{outs} we deduce $T \xlred{\out T_1} T_2$ and $(S_1, T_1) \in
        \srel$ and $(S_2, T_2) \in \srel$.
        From \rcac{ios} we deduce $R \xlred{\inp R_1} R_2$ and
        $\correct{R_1}{T_1}$ and $\correct{R_2}{T_2}$.
        Then $(R_1, S_1) \in \rrel$ and $(R_2, S_2) \in \rrel$ by definition of
        $\rrel$.
    \item Suppose $R \xlred{\out R_1} R_2$.
        From \rcac{ois} we deduce $T \xlred{\inp T_1} T_2$ and
        $\correct{R_1}{T_1}$ and $\correct{R_2}{T_2}$.
        From \rcas{inps} we deduce $S \xlred{\inp S_1} T_2$ and $(S_1, T_1) \in
        \srel$ and $(S_2, T_2) \in \srel$.
        Then $(R_1, S_1) \in \rrel$ and $(R_2, S_2) \in \rrel$ by definition of
        $\rrel$.
    \end{enumerate}

    Next we prove that $\subt$ is a coinductive asynchronous subtyping. Suppose
    $S \subt T$, then:
    \begin{enumerate}
    \item Suppose by contradiction that $\negative(S)$ and $\positive(T)$. Then
        from the hypothesis $S \subt T$ and the fact that
        $\correct{\dual{T}}{T}$ we deduce $\correct{\dual{T}}{S}$. But this is
        absurd because $\negative(\dual{T})$.
    \item Suppose $T \lred{\inp\unit} T$.
        Then $\dual{T} \lred{\out\unit}$. From $\correct{\dual{T}}{T}$ and $S
        \subt T$ we deduce $\correct{\dual{T}}{S}$ hence $S \lred{\inp\unit} S$
        and there is nothing left to prove.
    \item Suppose $T \lred{\inp\Tag} T'$.
        Let $R'$ be any session type such that $\correct{R'}{T'}$ and let $R =
        \Plus{\Tag : R'}$. Then $\correct{R}{T}$.
        From the hypothesis $S \subt T$ we deduce $\correct{R}{S}$, hence $S
        \lred{\inp\Tag} S'$.
        Also we have $\correct{R'}{S'}$ and since $R'$ is arbitrary we conclude
        $S' \subt T'$.
    \item Suppose $T \xlred{\inp T_1} T_2$.
        Let $R_1$ and $R_2$ by any session types such that $\correct{R_1}{T_1}$
        and $\correct{R_2}{T_2}$ and let $R \eqdef \Times{T_1}{T_2}$. Then
        $\correct{R}{T}$.
        From the hypothesis $S \subt T$ we deduce $\correct{R}{S}$, hence $S
        \xlred{\inp S_1} S_2$.
        Also we have $\correct{R_1}{S_1}$ and $\correct{R_2}{S_2}$ and since
        $R_1$ and $R_2$ are arbitrary we conclude $S_1 \subt T_1$ and $S_2 \subt
        T_2$.
    \item Suppose $S \lred{\out\unit} S$.\LP{Questa proprietà delle transizioni
      con $\unit$ andrebbe dimostrata?}
        By contradiction, suppose $T \nlred{\out\unit}$.
        Now $\correct{\dual{T}}{T}$ and $\dual{T} \nlred{\inp\unit}$, but we
        also have $\correct{\dual{T}}{S}$ from the hypothesis $S \subt T$, which
        is absurd. Hence $T \lred{\out\unit} T$ and there is nothing left to
        prove.
    \item Suppose $S \lred{\out\Tag} S'$.
        By contradiction, suppose $T \nlred{\out\Tag}$.
        Now $\correct{\dual{T}}{T}$ and $\dual{T} \nlred{\inp\Tag}$, but we also
        have $\correct{\dual{T}}{S}$ from the hypothesis $S \subt T$, which is
        absurd. Hence $T \lred{\out\Tag} T'$ for some $T'$.
        Let $R'$ be any session type such that $\correct{R'}{T'}$. Now consider
        $R = \With{\Tag : R'} \cup \set{\TagB : \dual{T'}}_{T \lred{\out\TagB}
        T', \TagA \ne \TagB}$. By construction we have $\correct{R}{T}$, hence
        from the hypothesis $S \subt T$ we deduce $\correct{R}{S}$.
        From \rcac{io} we deduce $\correct{R'}{S'}$, hence $S' \subt T'$ because
        $R'$ is arbitrary.
    \item Suppose $S \lred{\out S_1} S_2$.
        By contradiction, suppose $T \nlred{\out T'}$ for every $T'$. Now
        $\correct{\dual{T}}{T}$ and $\dual{T} \nlred{\inp T'}$ for every $T'$,
        but we also have $\correct{\dual{T}}{S}$ from the hypothesis $S \subt
        T$, which is absurd.
        Hence $T \xlred{\out T_1} T_2$ for some $T_1$ and $T_2$.
        Let $R_1$ and $R_2$ be any session types such that $\correct{R_1}{T_1}$
        and $\correct{R_2}{T_2}$ and consider $R \eqdef \Par{R_1}{R_2}$.
        By construction we have $\correct{R}{T}$, hence from the hypothesis $S
        \subt T$ we deduce $\correct{R}{S}$.
        From \rcac{ios} we deduce $\correct{R_1}{S_1}$ and $\correct{R_2}{S_2}$,
        hence $S_1 \subt T_1$ and $S_2 \subt T_2$ because $R_1$ and $R_2$ are
        arbitrary.
        \qedhere
    \end{enumerate}
\end{proof}

\thmcorrectsubt*
\begin{proof}
  ($\Leftarrow$)
  Suppose $S \subt \dual{T}$. From \Cref{prop:dual-correct} we deduce
  $\correct{T}{\dual{T}}$. From \Cref{def:subt} we conclude
  $\correct{S}{\dual{T}}$.

  ($\Rightarrow$)
  It is enough to prove that $\srel \eqdef \set{(S, \dual{T}) \mid
  \correct{S}{T}}$ is a coinductive asynchronous subtyping.
  Suppose $(S, \dual{T}) \in \srel$. Then $\correct{S}{T}$. We prove that $S$
  and $\dual{T}$ satisfy the conditions of \Cref{def:cas} in order:
  \begin{enumerate}
  \item Suppose by contradiction that $\negative(S)$ and
      $\positive(\dual{T})$. Then $\negative(T)$, which contradicts \rcac{pp}.
      Hence either $\positive(S)$ or $\negative(\dual{T})$, as required by
      \rcas{pol}.
  \item Suppose $\dual{T} \lred{\inp\MessageType} \dual{T'}$ and $\MessageType
      \in \set\unit \cup \TagSet$.
      Then $T \lred{\out\MessageType} T'$.
      From \rcac{io} we deduce $S \lred{\inp\MessageType} S'$ and
      $\correct{S'}{T'}$.
      We conclude $(S',\dual{T'}) \in \srel$ by definition of $\srel$, as
      required by \rcas{inp}.
  \item Suppose $S \lred{\out\MessageType} S'$ and $\MessageType \in \set\unit
      \cup \TagSet$.
      From \rcac{oi} we deduce $T \lred{\inp\MessageType} T'$ and
      $\correct{S'}{T'}$. Then $\dual{T} \lred{\out\MessageType} \dual{T'}$. 
      We conclude $(S',\dual{T'}) \in \srel$ by definition of $\srel$, as
      required by \rcas{out}.
  \item Suppose $\dual{T} \xlred{\inp \dual{T_1}} \dual{T_2}$.
      Then $T \lred{\out T_1} T_2$.
      From \rcac{ios} we deduce $S \lred{\inp S_1} S_2$ and $\correct{S_1}{T_1}$
      and $\correct{S_2}{T_2}$.
      We conclude $(S_1,\dual{T_1}) \in \srel$ and $(S_2, \dual{T_2}) \in \srel$
      by definition of $\srel$, as required by \rcas{inps}.
  \item Suppose $S \xlred{\out S_1} S_2$.
      From \rcac{ois} we deduce $T \xlred{\inp T_1} T_2$ and
      $\correct{S_1}{T_1}$ and $\correct{S_2}{T_2}$. Then $\dual{T}
      \xlred{\out\dual{T_1}} \dual{T_2}$.
      We conclude $(S_1,\dual{T_1}) \in \srel$ and $(S_2,\dual{T_2}) \in \srel$
      by definition of $\srel$, as required by \rcas{outs}.
      \qedhere
  \end{enumerate}
\end{proof}
\subsection{Undecidability Results}
\label{sec:proofs-undecidability}

\begin{definition}[Queue Machine]\label{def:queuemachines}
  A queue machine $M$ is defined by a six-tuple
  $(Q , \Sigma , \Gamma , \$ , s , \delta )$ where:
  \begin{itemize}
  \item $Q$ is a finite set of states;
  \item $\Sigma \subset \Gamma$ is a finite set denoting the input
    alphabet;
  \item $\Gamma$ is a finite set denoting the queue alphabet (ranged
    over by $A,B,C$);
  \item $\$ \in \Gamma -\Sigma$ is the initial queue symbol;
  \item $s \in Q$ is the start state;
  \item $\delta : Q \times \Gamma \rightarrow Q\times \Gamma ^{*}$ is
    the transition function ($\Gamma ^{*}$ is the set of sequences of
    symbols in $\Gamma$).
  \end{itemize}
Considering a queue machine
$M=(Q , \Sigma , \Gamma , \$ , s , \delta )$,
a {\em configuration} of $M$ is an ordered pair
$(q,\gamma)$ where $q\in Q$ is its {\em current state} and
$\gamma\in\Gamma ^{*}$ is the 
{\em queue}.  The
starting configuration on an input string $x \in \Sigma^*$ is 
$(s , x\$)$, composed of
the start state $s$ and the input $x$ followed by the initial queue symbol $\$$.
The transition relation ($\rightarrow_{M}$) over configurations 
$Q \times \Gamma ^{*}$, leading from a configuration to the next one,
is defined as follows.
For $p,q \in Q$, $A \in \Gamma$, and $\alpha,\gamma \in \Gamma ^{*}$, we have
$(p,A\alpha )\rightarrow _{M}(q,\alpha \gamma)$ whenever
$\delta (p,A)=(q,\gamma)$. Let $\rightarrow _{M}^{*}$ be the reflexive and
transitive closure of $\rightarrow _{M}$.
A machine $M$ accepts an input $x$ if it terminates on input $x$, i.e. it reaches
a blocking configuration with the empty queue (notice that, as the transition 
relation is total, the unique way to terminate is by emptying the queue). 
Formally $x$ is accepted by $M$ if and only if there exists $q \in Q$ such that 
$(s , x\$) \rightarrow _{M}^{*} (q,\varepsilon)$, where $\varepsilon$ is the empty string.
\end{definition}

Since queue machines can deterministically encode Turing machines
(see, e.g.,~\cite{KozenBook}, page~354, solution to exercise~99),
checking the acceptance of $x$ by a queue machine $M$ is an
undecidable problem.

\begin{definition}\label{encoding_queue_machines}


Consider a queue machine $M=(Q , \Sigma , \Gamma , \$ , s , \delta )$. 
Let $x \in \Gamma^*$ be a sequence of symbols $X_1 X_2 \cdots X_n$; 
we represent the queue containing such sequence $x$ with the session type 
$Q^M_x$ defined as follows:
$$
\begin{array}{llll}
Q^M_x & = & \Plus{X_1 : \Plus{X_2 : \cdots \Plus{X_n : Q^M } } }
	& 
	\\
Q^M & = & \With{A_i :\Plus{A_i : Q^M}}_{i\in I}
	& \mbox{with $I$ s.t. $\{A_i \mid i \in I\} = \Gamma$}
\end{array}
$$
The type $Q^M_x$ represents the content of the queue as a corresponding
sequence of output actions. After such sequence the type becomes the 
recursively defined type $Q^M$: this type starts with an input branching 
with one tag for each symbol $A_i$ in the queue alphabet, after each of
these inputs there is the output of the same tag, and then the type 
becomes $Q_M$ again.

We represent the transition function with the session type $S^M_s$.
In general, for each $q \in Q$, we consider the type $S^M_q$ defined as follows:
$$
\begin{array}{llll}
S^M_q & = & \With{A_i :\Plus{B_1 : \Plus{B_2 : \cdots \Plus{B_{n_i} : S^M_{q'} } } }}_{i\in I}
	&
	\begin{array}{l}
	\mbox{with $I$ s.t. $\{A_i \mid i \in I\} = \Gamma$}\\ 
	\mbox{and $\delta(q,A_i) = (q',B_1 B_2 \cdots B_{n_i})$} 
	\end{array}
\end{array}
$$
The type $S^M_q$, representing the transition function for the state $q$,
starts with an input action with an input branching 
with one tag for each symbol $A_i$ in the queue alphabet, after each of
these tags there is a sequence of outputs corresponding to the sequence
of symbols enqueued by the transition function after the consumption of 
the corresponding symbol $A_i$, and then the type becomes $S^M_{q'}$ where 
$q'$ is the new state of the queue machine.
\end{definition}

\begin{proposition}\label{from_machines_to_types_prop}
Consider a queue machine $M=(Q , \Sigma , \Gamma , \$ , s , \delta )$
and the transition $(p,\alpha )\rightarrow _{M}(q,\gamma)$,
with $\alpha = A_1 \cdots A_n$ and $\gamma = A_2 \cdots A_n B_1 \cdots B_m$,
because $\delta(p,A_1) = B_1 \cdots B_m$.
We have that $Q^M_\alpha \xlred{\out{A_1} \inp{B_1} \cdots \inp{B_m}} Q^M_\gamma$
and $S^M_p \xlred{\inp{A_1} \out{B_1} \cdots \out{B_m}} S^M_q$.
\end{proposition}
\begin{proof}
Trivial, by definition of $Q^M_\alpha$ and $S^M_p$.
\end{proof}

\begin{lemma}\label{from_machines_to_types}
Consider a queue machine $M=(Q , \Sigma , \Gamma , \$ , s , \delta )$
and an input $x$ accepted by $M$. We have that 
$\correct{Q^M_{x\$}}{S^M_s}$ does not hold.
\end{lemma}
\begin{proof}
By definition, if $x$ is accepted we have that 
$(s , x\$) \rightarrow _{M}^{*} (q,\varepsilon)$.
Suppose now, by contradiction, that $\correct{Q^M_{x\$}}{S^M_s}$.
This means that there exists a relation $\rrel$ which contains
$(Q^M_{x\$},S^M_s)$ and satisfies the properties in \cref{def:cac}.
By $(s , x\$) \rightarrow _{M}^{*} (q,\varepsilon)$ and repeated
application of \cref{from_machines_to_types_prop} (one application
for each transition in the sequence $(s , x\$) \rightarrow _{M}^{*} (q,\varepsilon)$) we have that
also $(Q^M_\varepsilon,S^M_q) \in \rrel$.
We have that $\negative(Q^M_\varepsilon)$ and $\negative(S^M_q)$.
But this contradicts the initial assumption about $\rrel$;
in fact $(Q^M_\varepsilon,S^M_q)$ does not satisfy the condition 
\ref{cac:pp} of \cref{def:cac} because neither $\positive(Q^M_\varepsilon)$ 
nor $\positive(S^M_q)$.
\end{proof}

\begin{lemma}\label{from_types_to_machines}
Consider a queue machine $M=(Q , \Sigma , \Gamma , \$ , s , \delta )$
and an input $x$ not accepted by $M$. We have that 
$\correct{Q^M_{x\$}}{S^M_s}$ holds.
\end{lemma}
\begin{proof}
In the proof we use decorations to keep track of outputs that are anticipated 
by the types $S^M_q$ representing the transition function.
Namely, we consider the definitions $S^M_q$ extended as follows:
$$
\begin{array}{ll}
S^M_q  =  \DecWith{q,\varepsilon}{A_i :\Plus{B_1 : \Plus{B_2 : \cdots \Plus{B_{n_i} : S^M_{q'} } } }}_{i\in I}
	&
	\begin{array}{l}
	\mbox{with $I$ s.t. $\{A_i \mid i \in I\} = \Gamma$}\\ 
	\mbox{and $\delta(q,A_i) = (q',B_1 B_2 \cdots B_{n_i})$} 
	\end{array}
\end{array}
$$
where $\color{red}q$ is the state represented by the type, and
$\color{red}\varepsilon$ is a the empty sequence. This sequence is used to keep track 
of the output anticipated by decorated terms. Formally, we consider an extended 
version of the semantics for types defined in \cref{tab:lts} which updates the decoration:
the rules $\AWithT$ and $\CoWithT$ updates the decorations by considering the 
following consequence:
$
\DecWith{q,\gamma}{\Tag_i : S_i}_{i\in I} \lred{\out\MessageType} \DecWith{q,\gamma\MessageType}{\Tag_i : T_i}_{i\in I}.
$

Consider a queue machine $M=(Q , \Sigma , \Gamma , \$ , s , \delta )$
and an input $x$ not accepted by $M$.
We now define the relation $\rrel$ as the minimal relation
satisfying the following conditions:
\begin{itemize}
\item $(Q^M_{x\$},S^M_s) \in \rrel$;
\item if $(Q,S) \in \rrel$,
      $Q \lred{\out A} Q'$, and $S \lred{\inp A} S'$,
      then $(S',T') \in \rrel$;
\item if $(Q,S) \in \rrel$,
      $S \lred{\out A} S'$, and $Q \lred{\inp A} Q'$,
      then $(S',T') \in \rrel$. 
\end{itemize}     
The remainder is dedicated to the proof that $\rrel$ is a
correct asynchronous composition relation, hence $\correct{Q^M_{x\$}}{S^M_s}$ holds
because $(Q^M_{x\$},S^M_s) \in \rrel$ (see \cref{def:cac}).

We first observe that $\rrel$ satisfies the conditions \ref{cac:oi} 
and \ref{cac:io} of the \cref{def:cac} of correct asynchronous composition relation. 
Let $(Q,S) \in \rrel$.
We have that, for each $A \in \Gamma$,
there exists $Q'$ and $S'$ such that $Q \lred{\inp A} Q'$
and $S \lred{\inp A} S'$.
This holds because both $Q$ and $S$ are types which begins with a 
(possibly empty) sequence of output actions followed by an
input branching with one label for each symbol $A$ in the 
queue alphabet (see \cref{encoding_queue_machines}).
Hence, for each $(Q,S) \in \rrel$ we have that if $Q \lred{\out A} Q'$
(resp. $S \lred{\out A} S'$) then there exists $S'$ (resp. $Q'$)
such that $S \lred{\inp A} S'$ (resp. $Q \lred{\inp A} Q'$).
Moreover, by definition of $\rrel$, we also have $(Q',S') \in \rrel$.

It remains to prove that also the condition \ref{cac:pp} of \cref{def:cac}
is satisfied by $\rrel$ (the remaining conditions \ref{cac:ois} and \ref{cac:ios} 
simply do not apply to $\rrel$ because the types $Q^M_{x\$}$ and $S^M_s$ do not 
include channel exchanges).
This is proved by showing a property which holds for each pair $(Q,S) \in \rrel$.
Let $(Q,S) \in \rrel$, we have that:
\begin{itemize} 
\item
$Q = \Plus{X_1 : \cdots \Plus{X_n : Q^M }}$,
for some $n \geq 0$,
\item
\begin{itemize}
\item
either $S = \Plus{X_{n+1} : \cdots \Plus{X_{m} : S^M_q }}$,
for some $m > n$ and some queue machine state $q$, with
$(s , x\$) \rightarrow _{M}^{*} (q,X_1 \cdots X_m)$,
\item
or
$S = \DecWith{q,\gamma}{\Tag_i : S_i}_{i\in I}$,
for some queue machine state $q$ and tag sequence $\gamma$ such that
there exists $1 \leq j \leq n$ for which $\gamma = X_{j+1} \cdots X_n$
and $(s , x\$) \rightarrow _{M}^{*} (q,X_1 \cdots X_j)$.
\end{itemize}
\end{itemize}
The initial pair  $(Q^M_{x\$},S^M_s) \in \rrel$
trivially satisfies the above conditions by definition of $Q^M_{x\$}$ and $S^M_s$
(see \cref{encoding_queue_machines}).\\
Consider now a generic pair $(Q,S) \in \rrel$ which satisfies
the above conditions.
The pair $(Q,S)$ implies the presence in $\rrel$ of at most two pairs:
$(Q',S')$, with $Q \lred{\out A} Q'$ and $S \lred{\inp A} S'$,
and $(Q'',S'')$, with $Q \lred{\inp A} Q''$ and $S \lred{\out A} S''$.
We proceed by case analysis to show that both pairs satisfy the above conditions. 

We have that $Q = \Plus{X_1 : \cdots \Plus{X_n : Q^M }}$,
for some $n \geq 0$. There are two possibile cases about the second type $S$.
\begin{itemize}
\item
$S = \Plus{X_{n+1} : \cdots \Plus{X_{m} : S^M_q }}$,
for some $m > n$ and some queue machine state $q$, with
$(s , x\$) \rightarrow _{M}^{*} (q,X_1 \cdots X_m)$.\\
We now consider the two possible pairs generated by $(Q,S)$:
\begin{itemize}
\item
$(Q',S')$, with $Q \lred{\out A} Q'$ and $S \lred{\inp A} S'$.\\
We have that $A=X_1$, 
$Q' = \Plus{X_2 : \cdots \Plus{X_n : Q^M } }$ and 
$S' = \Plus{X_{n+1} : \cdots \Plus{X_{m} : 
\Plus{B_{1} : \cdots \Plus{B_{k} : S^M_{q'}}}}}$,
with $\delta(q,X_1) = (q',B_1 \cdots B_{k})$.
The above properties holds for $(Q',S')$ because
$(q, X_{1} \cdots X_m)$ $\rightarrow _{M}$ 
$(q', X_{2} \cdots X_m B_1 \cdots B_{k})$.
\item
$(Q'',S'')$, with $Q \lred{\inp A} Q''$ and $S \lred{\out A} S''$.\\
In this case $Q'' = \Plus{X_1 : \cdots \Plus{X_n : \Plus{X_{n+1} : Q^M } }}$
and $S'' = \Plus{X_{n+2} : \cdots \Plus{X_{m} : S^M_q }}$.
The above properties holds for $(Q'',S'')$ because we already know that 
$(s , x\$) \rightarrow _{M}^{*} (q,X_1 \cdots X_m)$ and if $m = n+1$
we have that $S = \DecWith{q,\varepsilon}{\Tag_i : S_i}_{i\in I}$.
\end{itemize}

\item
$S = \DecWith{q,\gamma}{\Tag_i : S_i}_{i\in I}$,
for some queue machine state $q$ and tag sequence $\gamma$ such that
there exists $1 \leq j \leq n$ for which $\gamma = X_{j+1} \cdots X_n$
and $(s , x\$) \rightarrow _{M}^{*} (q,X_1 \cdots X_j)$.\\
We now consider the two possible pairs generated by $(Q,S)$:
\begin{itemize}
\item
$(Q',S')$, with $Q \lred{\out A} Q'$ and $S \lred{\inp A} S'$.\\
We have that $A=X_1$ and 
$Q' = \Plus{X_2 : \cdots \Plus{X_n : Q^M } }$.
Let $\delta(q,X_1) = (q',B_1 \cdots B_{k})$.
We consider two subcases:
\begin{itemize}
\item
$\gamma$ is a prefix of $B_1 \cdots B_{k}$.\\
Let $k' \leq k$ such that $\gamma = B_1 \cdots B_{k'}$.
In this case we have that $S' = \Plus{B_{k'+1} : \cdots \Plus{B_k : 
S^M_{q'} } }$. The above properties holds for $(Q',S')$ because
$(q, X_{1} \cdots X_j)$ $\rightarrow _{M}$ 
$(q', X_{2} \cdots X_j B_1 \cdots B_{k})$.
\item
$\gamma$ is a not prefix of $B_1 \cdots B_{k}$.\\
In this case, we have that $B_1 \cdots B_{k}$ is a prefix of $\gamma$
and $S' = \DecWith{q',\gamma'}{\Tag_i : S'_i}_{i\in I}$
with $\gamma'$ such that $\gamma = B_1 \cdots B_{k} \gamma'$.
We have that there exists $j \leq j' \leq n$ for which $\gamma' = X_{j'+1} \cdots X_n$.
The above properties holds for $(Q',S')$ because
$(q, X_{1} \cdots X_j)$ $\rightarrow _{M}$ 
$(q', X_{2} \cdots X_j B_1 \cdots B_{k})$.
\end{itemize}
\item
$(Q'',S'')$, with $Q \lred{\inp A} Q''$ and $S \lred{\out A} S''$.\\
In this case $Q'' = \Plus{X_1 : \cdots \Plus{X_n : \Plus{X_{n+1} : Q^M } }}$
and $S'' = \DecWith{q,\gamma X_{n+1}}{\Tag_i : S''_i}_{i\in I}$.
The above properties holds for $(Q'',S'')$ because we already know that 
there exists $1 \leq j \leq n$ for which $\gamma = X_{j+1} \cdots X_n$
(hence also $\gamma X_{n+1}= X_{j+1} \cdots X_{n+1}$)
and $(s , x\$) \rightarrow _{M}^{*} (q,X_1 \cdots X_j)$.
\end{itemize}
\end{itemize}
We can finally conclude by showing that the condition \ref{cac:pp} of \cref{def:cac}
is satisfied by $\rrel$. Suppose, by contradiction, that there exists a pair
$(Q,S) \in \rrel$ such that it is not true that either $\positive(Q)$ or $\positive(S)$.
This means that both $Q$ and $S$ begins with an input action, hence
$Q = Q^M$ and $S = \DecWith{q,\gamma}{\Tag_i : S_i}_{i\in I}$.
By the above properties which are satisfied by all pairs in $\rrel$ we have
that $(s , x\$) \rightarrow _{M}^{*} (q,\varepsilon)$, but this contradicts
the hypotesis about $x$ not accepted by $M$.
\end{proof}

\undecidability*
\begin{proof}
By \cref{from_machines_to_types} and \cref{from_types_to_machines}, given a
queue machine $M=(Q , \Sigma , \Gamma , \$ , s , \delta )$, we have that an
input $x$ is accepted by $M$ if and only if $\correct{Q^M_{x\$}}{S^M_s}$ does
not hold. Hence the undecidability of the correct asynchronous composition
relation $\correct{}{}$ is a direct consequence of the undecidability of
acceptance in queue machines.
\end{proof}

\section{Supplement to Section~\ref{sec:type-system}}
\label{sec:proofs-type-system}

\subsection{Proof of Theorem~\ref{thm:sr}}
\label{sec:proof-subject-reduction}

We introduce some additional notation for consuming and creating session types:
\begin{itemize}
    \item If $S \lred\Action$ we write $S(\Action)$ for the (unique) $T$ such
    that $S \lred\Action T$ when $S \lred\Action$. We extend this notation to
    sequences of actions in the obvious way, that is $S(\varepsilon) \eqdef S$
    and $S(\Action\Actions) \eqdef S(\Action)(\Actions)$.
    \item We write $\Action.S$ for the (unique) session type $T$ such that $T
    \lred\ActionB$ implies $\ActionB = \Action$ and $T \lred\Action S$. For
    example, if $\Action = \out\Tag$, then $\Action.S = \Plus{\Tag : S}$. We
    extend this notation to sequences of actions in the obvious way, that is
    $\varepsilon.S \eqdef S$ and $\Action\Actions.S \eqdef \Action.\Actions.S$.
\end{itemize}

\begin{proposition}
    \label{prop:refine-subt}
    If\/ $S \lred\Actions$, then the following properties hold:
    \begin{enumerate}
    \item if $\Actions$ is made of output actions only, then
        $\Actions.S(\Actions) \subt S$;
    \item if $\Actions$ is made of input actions only, then $S \subt
    \Actions.S(\Actions)$.
    \end{enumerate}
\end{proposition}
\begin{proof}
    By induction on $\Actions$.
    If $\Actions = \varepsilon$, then we conclude immediately by reflexivity of
    $\subt$.
    If $\Actions = \Action\ActionsB$, then $S(\Action) \lred\ActionsB$.
    By induction hypothesis we deduce that if $\ActionsB$ is made of output
    actions only, then $\ActionsB.S(\Action\ActionsB) \subt S(\Action)$ and that
    if $\ActionsB$ is made of input actions only, then $S(\Action) \subt
    \ActionsB.S(\Action\ActionsB)$.
    We conclude by recalling \Cref{def:cas} and observing that if $\Action$ is
    an output action, then $\Action.S(\Action) \subt S$ and that if $\Action$ is
    an input action, then $S \subt \Action.S(\Action)$.
\end{proof}

\newcommand{\TContext}{\TContextS}
\newcommand{\TContextS}{\mathcal{S}}
\newcommand{\TContextT}{\mathcal{T}}

\begin{table}
    \caption{\label{tab:buffer-typing-rules}Typing rules for buffers.}
    \begin{mathpar}
        \inferrule[\EmptyBufferRule]{~}{
            \wtb[0]\Hole\EmptyContext\x\varepsilon
        }
        \and
        \inferrule[\SelectBufferRule]{
            \wtb[n]{\Buffer_x}\Context\x\Actions
        }{
            \wtb[1+m+n]{\Select\x\Tag.\Buffer_x}\Context\x{\out\ann[m]\Tag\Actions}
        }
        \and
        \inferrule[\ForkBufferRule]{
            \wtp[m]{P}{\ContextC, y : T}
            \\
            \wtb[n]{\Buffer_x}\ContextD\x\Actions
        }{
            \wtb[1+m+n]{\Fork\x\y{P}{\Buffer_x}}{\ContextC, \ContextD}\x{\out{T}\Actions}
        }
    \end{mathpar}
\end{table}

In order to prove \Cref{thm:sr} we need to introduce typing rules for buffers,
shown in \Cref{tab:buffer-typing-rules}. Note that, in the typing judgments
$\wtb[n]{\Buffer_x}\Context\x\Actions$, we single out the variable $x$ bound by
the buffer because its ``type'' is actually a sequence of output actions.

We now prove two auxiliary results expressing the effect of assembling and
disassembling buffers and processes.

\begin{lemma}
    \label{lem:buffer-type-inv}
    If $\wtb[m]{\Buffer_x}\ContextC\x\Actions$ and $\wtp[n]{P}{\ContextD, x :
    S}$ and $\dom\ContextC \cap \dom\ContextD = \emptyset$, then
    $\wtp[m+n]{\Buffer_x[P]}{\ContextC, \ContextD, x : \Actions.S}$.
\end{lemma}
\begin{proof}
    By induction on $\Buffer_x$ and by cases on its shape.

    \proofcase{$\Buffer_x = \Hole$}
    From \EmptyBufferRule we deduce
    \begin{itemize}
        \item $m = 0$
        \item $\ContextC = \EmptyContext$
        \item $\Actions = \varepsilon$
    \end{itemize}
    
    We conclude immediately since $\Buffer_x[P] = P$ and $\Actions.S = S$.

    \proofcase{$\Buffer_x = \Select\x\Tag.\Buffer_x'$}
    From \SelectBufferRule we deduce
    \begin{itemize}
        \item $m = 1 + m'' + m'$
        \item $\Actions = \out\ann[m'']\Tag\Actions'$
        \item $\wtb[m']{\Buffer_x'}\ContextC\x{\Actions'}$
    \end{itemize}

    Using the induction hypothesis we deduce
    $\wtp[m'+n]{\Buffer_x'[P]}{\ContextC, \ContextD, x : \Actions'.S}$.
    We conclude $\wtp[m + n]{\Buffer_x[P]}{\ContextC, \ContextD, x :
    \Actions.S}$ with one application of \SelectRule.

    \proofcase{$\Buffer_x = \Fork\x\y{Q}{\Buffer_x'}$}
    From \ForkBufferRule we deduce
    \begin{itemize}
        \item $m = 1 + m'' + m'$
        \item $\Actions = \out{T}\Actions'$
        \item $\Context = \ContextE, \Context'$
        \item $\wtp[m'']{Q}{\ContextE, y : T}$
        \item $\wtb[m']{\Buffer_x'}{\Context'}\x{\Actions'}$
    \end{itemize}

    From $\dom\ContextC \cap \dom\ContextD = \emptyset$ we deduce
    $\dom{\ContextC'} \cap \dom\ContextD = \emptyset$. Using the induction
    hypothesis we deduce $\wtp[m'+n]{\Buffer_x'[P]}{\Context', \ContextD, x :
    \Actions'.S}$.
    We conclude $\wtp[m+n]{\Buffer_x[P]}{\Context, \ContextD, x : \Actions.S}$
    with one application of \ForkRule.
\end{proof}

\begin{lemma}
    \label{lem:buffer-type}
    If $\wtp[n]{\Buffer_x[P]}{\ContextC, x : S}$, then there exist $m\leq n$,
    $\ContextE$, $\ContextD$ and $\Actions$ such that $\Context = \ContextE,
    \ContextD$ and $\wtb[m]{\Buffer_x}\ContextE\x\Actions$ and
    $\wtp[n-m]{P}{\ContextD, x : S(\Actions)}$.
\end{lemma}
\begin{proof}
    By induction on $\Buffer_x$ and by cases on its shape.

    \proofcase{$\Buffer_x = \Hole$}
    We conclude by taking $m \eqdef 0$, $\ContextE = \EmptyContext$, $\ContextD
    = \ContextC$ and $\Actions \eqdef \varepsilon$ with one application of
    \EmptyBufferRule.
        
    \proofcase{$\Buffer_x = \Select\x\Tag.\Buffer_x'$}
    From \SelectRule we deduce
    \begin{itemize}
        \item $S = \Plus{\ann[m_i]{\Tag_i} : S_i}_{i\in I}$
        \item $\Tag = \Tag_k$ for some $k \in I$
        \item $n = 1 + m_k + n'$
        \item $\wtp[n']{\Buffer_x'[P]}{\Context, x : S_k}$
    \end{itemize}

    Using the induction hypothesis we deduce that there exist $m'\leq n'$,
    $\ContextE$, $\ContextD$ and $\Actions'$ such that $\Context = \ContextE,
    \ContextD$ and $\wtb[m']{\Buffer_x}\ContextE\x{\Actions'}$ and
    $\wtp[n'-m']{P}{\ContextD, x : S_k(\Actions')}$.
    We conclude by taking $m \eqdef 1 + m_k + m'$ and $\Actions \eqdef
    \out\ann[m_k]{\Tag}\Actions'$ with one application of \SelectBufferRule
    and observing that $n - m = n - 1 - m_k - m' = n' - m'$ and $S(\Actions) =
    S(\out\ann[m_k]{\Tag}\Actions') = S_k(\Actions')$.

    \proofcase{$\Buffer_x = \Fork\x\y{Q}{\Buffer_x'}$}
    From \ForkRule we deduce
    \begin{itemize}
        \item $\Context = \ContextE'', \Context'$
        \item $S = \Times{T}{S'}$
        \item $n = 1 + m'' + n'$
        \item $\wtp[m'']{Q}{\ContextE'', y : T}$
        \item $\wtp[n']{\Buffer_x'[P]}{\Context', x : S'}$
    \end{itemize}

    Using the induction hypothesis we deduce that there exist $m'\leq n'$,
    $\ContextE'$, $\ContextD$ and $\Actions'$ such that $\Context' = \ContextE',
    \ContextD$ and $\wtb[m']{\Buffer_x}{\ContextE'}\x{\Actions'}$ and $\wtp[n' -
    m']{P}{\ContextD, x : S'(\Actions')}$.
    We conclude by taking $m \eqdef 1 + m'' + m'$ and $\ContextE = \ContextE',
    \ContextE''$ and $\Actions \eqdef \out{T}\Actions'$ with one application of
    \ForkBufferRule and observing that $n - m = n' - m'$ and $S(\Actions) =
    S(\out{T}\Actions') = S'(\Actions')$.
\end{proof}

The next lemma proves that structural precongruence preserves typing, in the
sense that if $P \pcong Q$ then $Q$ is well typed in a the typing context
$\ContextD$ that is $\subt$-smaller than the one used for typing $P$. Also, the
measure of $Q$ is not greater than the measure of $P$.

\begin{lemma}
    \label{lem:pcong-preservation}
    If $P \pcong Q$, then $\wtp\Context{P}$ implies $\wtp[m]\ContextD{Q}$ for
    some $m \leq n$ and $\ContextD \subt \ContextC$.
\end{lemma}
\begin{proof}
    By induction on the derivation of $P \pcong Q$ and by cases on the last rule
    applied. We only discuss a few representative cases.

    \proofcase{\rulename{s-comm}}
    Then $P = \Cut\x{P_1}{P_2} \pcong \Cut\x{P_2}{P_1} = Q$.
    From \CutRule we deduce $\Context = \Context_1, \Context_2$ and
    $\wtp[n_i]{P_i}{\Context_i, S_i}$ for $i=1,2$ and $n = n_1 + n_2$ and
    $\cc{S_1}{S_2}$.
    We conclude $\wtp{Q}\Context$ with one application of \CutRule since
    $\cc{S_1}{S_2}$ implies $\cc{S_2}{S_1}$.

    \proofcase{\rulename{s-pull-0}}
    Then $P = \Cut\x{\Buffer_x[R]}{\Buffer_x'[\Link\x\y]} \pcong
    \Buffer_y[\Cut\x{R}{\Buffer_x'[\Link\x\y]}] = Q$.
    From \CutRule we deduce
    \begin{itemize}
        \item $\Context = \Context_1, \Context_2$
        \item $n = n_1 + n_2$
        \item $\wtp[n_1]{\Buffer_x[R]}{\Context_1, x : S}$
        \item $\wtp[n_2]{\Buffer_x'[\Link\x\y]}{\Context_2, x : T}$
        \item $\cc{S}{T}$
    \end{itemize}

    From \Cref{lem:buffer-type} we deduce that there exist $m_1 \leq n_1$,
    $\ContextE_1$, $\ContextD_1$ and $\Actions_1$ such that $\Context_1 =
    \ContextE_1, \ContextD_1$ and
    $\wtb[m_1]{\Buffer_x}{\ContextE_1}\x{\Actions_1}$ and $\wtp[n_1 -
    m_1]{R}{\ContextD_1, x : S(\Actions_1)}$.
    From \Cref{lem:buffer-type} we deduce that there exist $m_2 \leq n_2$,
    $\ContextE_2$, $\ContextD_2$ and $\Actions_2$ such that $\Context_2 =
    \ContextE_2, \ContextD_2$ and
    $\wtb[m_2]{\Buffer_x'}{\ContextE_2}\x{\Actions_2}$ and $\wtp[n_2 -
    m_2]{\Link\x\y}{\ContextD_2, x : T(\Actions_2)}$.
    From \LinkRule we deduce
    \begin{itemize}
        \item $n_2 - m_2 = 1$
        \item $\ContextD_2 = y : S'$
        \item $\dual{T(\Actions_2)} \subt S'$
    \end{itemize}

    From \Cref{prop:refine-subt} and $\cc{S}{T}$ we deduce
    $\cc{\Actions_1.S(\Actions_1)}{T}$ hence
    $\cc{\Actions_1.S(\Actions_1\dual{\Actions_2})}{T(\Actions_2)}$.
    From $\dual{T(\Actions_2)} \subt S'$ we deduce $\dual{S'} \subt
    T(\Actions_2)$ hence $\cc{\Actions_1.S(\Actions_1\dual{\Actions_2})} \subt
    S'$.\LP{Direi che è corretto ma qui si usa di tutto e di più.}
    From $T \lred{\Actions_2}$ and $\cc{S}{T}$ and \Cref{def:cac} we deduce $S
    \lred{\dual\Actions_2}$ hence $S(\Actions_1) \lred{\dual\Actions_2}$ and
    $S(\Actions_1\dual\Actions_2)$ is defined.
    We derive $\wtp[1]{\Link\x\y}{x : \dual{S(\Actions_1\dual\Actions_2)}, y :
    S(\Actions_1\dual\Actions_2)}$ with one application of \LinkRule.
    From \Cref{lem:buffer-type-inv} we deduce
    $\wtp[n_2]{\Buffer_x'[\Link\x\y]}{\ContextE_2, x :
    \Actions_2.\dual{S(\Actions_1\dual\Actions_2)}, y :
    S(\Actions_1\dual\Actions_2)}$.
    From \Cref{prop:refine-subt} we deduce $S(\Actions_1) \subt
    \dual\Actions_2.S(\Actions_1\dual\Actions_2) =
    \dual{(\Actions_2.\dual{S(\Actions_1\dual\Actions_2)})}$ hence
    $\cc{S(\Actions_1)}{\Actions_2.\dual{S(\Actions_1\dual\Actions_2)}}$.
    We derive $\wtp[n - m_1]{\Cut\x{R}{\Buffer_x'[\Link\x\y]}}{\ContextD_1,
    \ContextE_2, y : S(\Actions_1\dual\Actions_2)}$ with one application of
    \CutRule.
    From \Cref{lem:buffer-type-inv} we deduce $\wtp[n]{Q}{\Context_1,
    \ContextE_2, y : \Actions_1.S(\Actions_1\dual\Actions_2)}$.
    We conclude by taking $\ContextD \eqdef \Context_1, \ContextE_2, y :
    \Actions_1.S(\Actions_1\dual\Actions_2)$ recalling that
    $\Actions_1.S(\Actions_1\dual\Actions_2) \subt S'$.

    \proofcase{\rulename{s-pull-1}}
    Then $P = \Cut\x{\Fork\y\z{P_1}{P_2}}{P_3} \pcong
    \Fork\y\z{\Cut\x{P_1}{P_3}}{P_2} = Q$ where $x \ne y$ and $x\in\fn{P_1}$
    From \CutRule we deduce
    \begin{itemize}
        \item $\Context = \Context_{12}, \Context_3$
        \item $n = n_{12} + n_3$
        \item $\wtp[n_1]{\Fork\y\z{P_1}{P_2}}{\Context_{12}, x : S}$
        \item $\wtp[n_3]{P_3}{\Context_3, x : T}$
        \item $\cc{S}{T}$
    \end{itemize}

    From \ForkRule we deduce
    \begin{itemize}
        \item $\Context_{12} = \Context_1, \Context_2, y : \Times{S_1}{S_2}$
        \item $n_{12} = 1 + n_1 + n_2$
        \item $\wtp[n_1]{P_1}{\Context_1, x : S, z : S_1}$
        \item $\wtp[n_2]{P_2}{\Context_2, y : S_2}$
    \end{itemize}

    We derive $\wtp[n_1 + n_3]{\Cut\x{P_1}{P_3}}{\Context_1, \Context_3, z :
    S_1}$ with one application of \CutRule.
    We derive $\wtp[n]{Q}\Context$ with one application of \ForkRule and we
    conclude by taking $\ContextD \eqdef \Context$.

    \proofcase{\rulename{s-pull-2}}
    Then $P = \Cut\x{\Buffer_y[P_1]}{P_2} \pcong \Buffer_y[\Cut\x{P_1}{P_2}] =
    Q$ where $x \ne y$ and $x \not\in \fn{\Buffer_y}$. Without loss of
    generality we may assume $\Buffer_y \ne \Hole$, or else the result is
    trivial.
    From \CutRule we deduce
    \begin{itemize}
        \item $\Context = \Context_1, \Context_2, y : S$
        \item $n = n_1 + n_2$
        \item $\wtp[n_1]{\Buffer_y[P_1]}{\Context_1, x : S_1, y : S}$
        \item $\wtp[n_2]{P_2}{\Context_2, x : S_2}$
        \item $\cc{S_1}{S_2}$
    \end{itemize}

    From \Cref{lem:buffer-type} we deduce that there exist $m \leq n_1$,
    $\ContextE$, $\Context_1'$ and $\Actions$ such that $\Context_1 = \ContextE,
    \Context_1'$ and $\wtb[m]{\Buffer_y}\ContextE\y\Actions$ and $\wtp[n_1 -
    m]{P_1}{\Context_1', x : S_1, y : S(\Actions)}$.
    We derive $\wtp[n - m]{\Cut\x{P_1}{P_2}}{\Context_1', \Context_2, y :
    S(\Actions)}$ with one application of \CutRule.
    From \Cref{lem:buffer-type-inv} we deduce $\wtp[n]{Q}{\Context_1,
    \Context_2, y : \Actions.S(\Actions)}$.
    We conclude by taking $\ContextD \eqdef \Context_1, \Context_2, y :
    \Actions.S(\Actions)$ and observing that $\Actions.S(\Actions) \subt S$ by
    \Cref{prop:refine-subt}.

    \proofcase{\rulename{s-pull-3}}
    Then $P = \Buffer_x[\Fork\y\z{P_1}{P_2}] \pcong
    \Fork\y\z{\Buffer_x[P_1]}{P_2} = Q$ where $x \ne y$ and $x\in\fn{P_1}$ and
    $z\not\in\fn{\Buffer_x}$.
    Without loss of generality we may assume that $\Buffer_x \ne \Hole$ or else
    the result is trivial. Then $\Context = \Context', x : S$.
    From \Cref{lem:buffer-type} we deduce that there exist $m \leq n$,
    $\ContextE$, $\Context''$ and $\Actions$ such that $\Context' = \ContextE,
    \Context''$ and $\wtb[m]{\Buffer_x}\ContextE\x\Actions$ and $\wtp[n -
    m]{\Fork\y\z{P_1}{P_2}}{\Context'', x : S(\Actions)}$.
    From \ForkRule and $x\in\fn{P_1}$ we deduce that
    \begin{itemize}
        \item $\Context'' = \Context_1, \Context_2, y : \Times{S_1}{S_2}$ for
        some $\Context_1$, $\Context_2$, $S_1$ and $S_2$
        \item $n - m = 1 + n_1 + n_2$ for some $n_1$ and $n_2$
        \item $\wtp[n_1]{P_1}{\Context_1, x : S(\Actions), z : S_1}$
        \item $\wtp[n_2]{P_2}{\Context_2, y : S_2}$
    \end{itemize}

    From \Cref{lem:buffer-type-inv} we deduce $\wtp[n_1 +
    m]{\Buffer_x[P_1]}{\ContextE, \Context_1, x : \Actions.S(\Actions), z :
    S_1}$.
    We derive $\wtp[n]{Q}{\Context', x : \Actions.S(\Actions)}$ with an
    application of \ForkRule and we conclude by taking $\ContextD \eqdef
    \Context', x : \Actions.S(\Actions)$ and observing that
    $\Actions.S(\Actions) \subt S$ by \Cref{prop:refine-subt}.

    \proofcase{\rulename{s-pull-4}}
    Then $P = \Buffer_x[\Buffer_y'[R]] \pcong \Buffer_y'[\Buffer_x[R]] = Q$
    where $x \ne y$ and $x \in \fn{R}$.
    Without loss of generality we may assume that $\Buffer_x \ne \Hole$ and
    $\Buffer_y' \ne \Hole$ or else the result is trivial. Then $\Context =
    \Context', x : S, y : T$.
    From \Cref{lem:buffer-type} we deduce that there exist $m_1 \leq n$,
    $\ContextE_1$, $\Context''$ and $\Actions_1$ such that $\Context' =
    \ContextE_1, \Context''$ and
    $\wtb[m_1]{\Buffer_x}{\ContextE_1}\x{\Actions_1}$ and $\wtp[n -
    m_1]{\Buffer_y'[R]}{\Context'', x : S(\Actions_1), y : T}$.
    From \Cref{lem:buffer-type} and $x\in\fn{R}$ we deduce that there exist $m_2
    \leq n - m_1$, $\ContextE_2$, $\Context'''$ and $\Actions_2$ such that
    $\Context'' = \ContextE_1, \Context'''$ and
    $\wtb[m_2]{\Buffer_y'}{\ContextE_2}\y{\Actions_2}$ and $\wtp[n - m_1 -
    m_2]{R}{\Context''', x : S(\Actions_1), y : T(\Actions_2)}$.
    Using \Cref{lem:buffer-type-inv} we deduce that $\wtp[n -
    m_2]{\Buffer_x[R]}{\ContextE_1, \Context''', x : \Actions_1.S(\Actions_1), y
    : T(\Actions_2)}$.
    Using \Cref{lem:buffer-type-inv} we deduce that $\wtp[n]{Q}{\Context', x :
    \Actions_1.S(\Actions_1), y : \Actions_2.T(\Actions_2)}$.
    We conclude by taking $\ContextD \eqdef \Context', x :
    \Actions_1.S(\Actions_1), y : \Actions_2.T(\Actions_2)$ and by
    \Cref{prop:refine-subt}.

    \proofcase{\rulename{s-call}}
    Then $P = \Call{A}{\seqof\x} \pcong Q$ where $\Context = \seqof{x : S}$ and
    $\Let{A}{\seqof\x} = Q$.
    From \CallRule we deduce $\wtp[m]{Q}\Context$ for some $m \leq n$.
    We conclude by taking $\ContextD \eqdef \ContextC$.
\end{proof}

\begin{lemma}[substitution]
    \label{lem:substitution}
    If $\wtp[n]{P}{\Context, x : S}$ and $y\not\in\dom\Context$, then
    $\wtp[n]{P\subst\y\x}{\Context, y : S}$.
\end{lemma}
\begin{proof}
    The derivation for $\wtp[n]{P\subst\y\x}{\Context, y : S}$ is obtained
    corecursively from the one for $\wtp[n]{P}{\Context, x : S}$ by renaming
    each free occurrence of $x$ into $y$.
\end{proof}

\thmsr*
\begin{proof}
    By induction on the derivation of $P \red Q$ and by cases on the last rule
    applied.

    \proofcase{\rulename{r-choice}}
    Then $P = P_1 \choice P_2 \red P_k = Q$ for some $k\in\set{1,2}$.
    From \ChoiceRule we deduce $\wtp[\Measure_i]{P_i}\Context$ for $i=1,2$ and
    $n = 1 + n_l$ for some $l\in\set{1,2}$.
    We conclude by taking $m \eqdef n_l$ and $\ContextD \eqdef
    \ContextC$.
    Note that, if $l = k$, then $m = n_k < 1 + n_k = n$.

    \proofcase{\rulename{r-link}}
    Then $P = \Cut\x{\Link\x\y}{P} \red P\subst\y\x = Q$.
    From \CutRule we deduce
    \begin{itemize}
        \item $n = n_1 + n_2$    
        \item $\Context = \Context_1, \Context_2$
        \item $\wtp[n_1]{\Link\x\y}{\Context_1,x : S}$
        \item $\wtp[n_2]{P}{\Context_2,x : T}$
        \item $\cc{S}{T}$
    \end{itemize}

    From \LinkRule we deduce
    \begin{itemize}
        \item $n_1 \geq 1$
        \item $\Context_1 = y : S'$ for some $S'$
        \item $\dual{S} \subt S'$
    \end{itemize}

    From \Cref{lem:substitution} we deduce $\wtp[n_2]{P\subst\y\x}{\Context_2, y
    : T}$ and we conclude by taking $m \eqdef n_2$ and $\ContextD \eqdef
    \Context_2, y : T$. Note that $\ContextD = y : T, \Context_2 \subt y : S',
    \Context_2 = \Context_1, \Context_2 = \ContextC$ and that $m = n_2 < 1 + n_2
    \leq n_1 + n_2 = n$.\LP{Spiegare come si deduce la relazione $T \subt S'$.}

    \proofcase{\rulename{r-close}}
    Then $P = \Cut\x{\Close\x}{\Wait\x.P} \red P$.
    From \CutRule we deduce
    \begin{itemize}
        \item $\Context = \Context_1, \Context_2$
        \item $\wtp[n_1]{\Close\x}{\Context_1,x : S}$
        \item $\wtp[n_2]{\Wait\x.P}{\Context_2,x : T}$
        \item $\cc{S}{T}$
        \item $n = n_1 + n_2$
    \end{itemize}

    From \CloseRule we deduce
    \begin{itemize}
        \item $\Context_1 = \EmptyContext$, therefore $\Context_2 = \Context$
        \item $n_1 \geq 1$
        \item $S = \One$
    \end{itemize}

    From \WaitRule we deduce
    \begin{itemize}
        \item $T = \Bot$
        \item $\wtp[\Measure_2]{P}\Context$
    \end{itemize}
    
    We conclude by taking $m \eqdef n_2$ and $\ContextD \eqdef \ContextC$. Note
    that $m = n_2 < 1 + n_2 \leq n_1 + n_2 = n$.

    \proofcase{\rulename{r-select}}
    Then $P = \Cut\x{\Select\x{\Tag_k}.R}{\Buffer_x[\Case\x{\Tag_i:P_i}_{i\in
    I}]} \red \Cut\x{R}{P_k} = Q$ with $k\in I$.
    From \CutRule we deduce
    \begin{itemize}
        \item $n = n_1 + n_2$
        \item $\Context = \Context_1, \Context_2$
        \item $\wtp[n_1]{\Select\x{\Tag_k}.R}{\Context_1, x : S}$
        \item $\wtp[n_2]{\Buffer_x[\Case\x{\Tag_i:P_i}_{i\in I}]}{\Context_2, x : T}$
        \item $\cc{S}{T}$
    \end{itemize}

    From \SelectRule we deduce
    \begin{itemize}
        \item $S = \Plus{\ann[m_i]{\Tag_i} : S_i}_{i\in K}$
        \item $n_1 = 1 + m_k + n_1'$
        \item $k\in K$
        \item $\wtp[n_1']{R}{\Context_1, x : S_k}$
    \end{itemize}

    From \Cref{lem:buffer-type} we deduce that there exist $m' \leq n_2$,
    $\ContextE$, $\ContextD'$ and $\Actions$ such that $\Context_2 = \ContextE,
    \ContextD'$ and $\wtb[m']{\Buffer_x}\ContextE\x\Actions$ and
    $\wtp[n_2-m']{\Case\x{\Tag_i : P_i}_{i\in I}}{\ContextD', x : T(\Actions)}$.
    From \CaseRule we deduce
    \begin{itemize}
        \item $T(\Actions) = \With{\ann[m_i]{\Tag_i} : T_i}_{i\in J}$ with $J
        \subseteq I$
        \item $\wtp[n_2 - m' + m_i]{P_i}{\ContextD', x : T_i}$ for every $i\in J$
    \end{itemize}

    From \Cref{lem:buffer-type-inv} we derive
    $\wtp[n_2+m_k]{\Buffer_x[P_k]}{\Context_2, x : \Actions.T_k}$.
    From $\cc{S}{T}$ we deduce $K \subseteq J$ and
    $\cc{S_k}{T(\inp\ann[m_k]{\Tag_k})}$.
    Also using \Cref{prop:refine-subt} we have $\Actions.T_k =
    \Actions.T(\Actions)(\inp\ann[m_k]{\Tag_k}) =
    (\Actions.T(\Actions))(\inp\ann[m_k]{\Tag_k}) \subt
    T(\inp\ann[m_k]{\Tag_k})$, therefore $\cc{S_k}{(\Actions.T_k)}$.
    From \CutRule we derive $\wtp[n-1]{Q}\Context$ and we conclude by taking $m
    \eqdef n - 1$ and $\ContextD \eqdef \Context$.

    \proofcase{\rulename{r-fork}}
    Then $P = \Cut\x{\Fork\x\y{P_1}{P_2}}{\Buffer_x[\Join\x\y.P_3]} \red
    \Cut\y{P_1}{\Cut\x{P_2}{\Buffer_x[P_3]}} = Q$.
    From \CutRule we deduce
    \begin{itemize}
        \item $n = n_{12} + n_3$
        \item $\Context = \Context_{12}, \Context_3$
        \item $\wtp[n_{12}]{\Fork\x\y{P_1}{P_2}}{\Context_{12}, x : S}$
        \item $\wtp[n_3]{\Buffer_x[P_3]}{\Context_3, x : T}$
        \item $\cc{S}{T}$
    \end{itemize}

    From \ForkRule we deduce
    \begin{itemize}
        \item $\Context_{12} = \Context_1, \Context_2$
        \item $n_{12} = 1 + n_1 + n_2$
        \item $S = \Times{S_1}{S_2}$
        \item $\wtp[n_1]{P_1}{\Context_1, y : S_1}$
        \item $\wtp[n_2]{P_2}{\Context_2, x : S_2}$
    \end{itemize}

    From \Cref{lem:buffer-type} we deduce that there exist $m' \leq n_3$,
    $\ContextE$, $\ContextD'$ and $\Actions$ such that $\Context_3 = \ContextE,
    \ContextD'$ and $\wtb[m']{\Buffer_x}\ContextE\x\Actions$ and $\wtp[n_3 -
    m']{\Join\x\y.P_3}{\ContextD', x : T(\Actions)}$.
    From \JoinRule we deduce
    \begin{itemize}
        \item $T(\Actions) = \Par{T_1}{T_2}$
        \item $\wtp[n_3 - m']{P_3}{\ContextD', y : T_1, x : T_2}$
    \end{itemize}

    From $\cc{S}{T}$ we deduce $\cc{S_1}{T_1}$ and $\cc{S_2}{T(\inp T_1)}$.
    Also, using \Cref{prop:refine-subt} we have $\Actions.T_2 =
    \Actions.T(\Actions)(\inp T_1) = (\Actions.T(\Actions))(\inp T_1) \subt
    T(\inp T_1)$, therefore $\cc{S_2}{\Actions.T_2}$.
    Using \Cref{lem:buffer-type-inv} we deduce
    $\wtp[n_3]{\Buffer_x[P_3]}{\Context_3, y : T_1, x : \Actions.T_2}$.
    We derive $\wtp[n_2 + n_3]{\Cut\x{P_2}{\Buffer_x[P_3]}}{\Context_2,
    \Context_3, y : T_1}$ with an application of \CutRule.
    We derive $\wtp[n - 1]{Q}\Context$ with another application of \CutRule and
    we conclude by taking $m \eqdef n - 1$ and $\ContextD \eqdef \ContextC$.

    \proofcase{\rulename{r-cut}}
    Then $P = \Cut\x{P_1}{P_2} \red \Cut\x{P_1'}{P_2} = Q$ where $P_1 \red P_1'$.
    From \CutRule we deduce
    \begin{itemize}
        \item $n = n_1 + n_2$
        \item $\Context = \Context_1, \Context_2$
        \item $\wtp[n_1]{P_1}{\Context_1, x : S}$
        \item $\wtp[n_2]{P_2}{\Context_2, x : T}$
        \item $\cc{S}{T}$
    \end{itemize}

    Using the induction hypothesis we deduce $\wtp[n_1']{P_1'}{\Context_1', x :
    S'}$ for some $\Context_1' \subt \Context_1$ and $S' \subt S$, hence
    $\cc{S'}{T}$. We conclude by taking $m \eqdef n_1' + n_2$ and $\ContextD
    \eqdef \Context_1', \Context_2$ with one application of \CutRule.

    \proofcase{\rulename{r-buffer}}
    Then $P = \Buffer_x[R] \red \Buffer_x[R'] = Q$ where $R \red R'$.
    From \Cref{lem:buffer-type} we deduce that there exist $m' \leq n$,
    $\ContextE$, $\ContextD'$ and $\Actions$ such that $\Context = \ContextE,
    \ContextD'$ and $\wtb[m']{\Buffer_x}\ContextE\x\Actions$ and
    $\wtp[n-m']{R}{\ContextD', x : T(\Actions)}$.
    Using the induction hypothesis we deduce that $\wtp[m'']{R'}{\ContextD'', x
    : T''}$ for some $\ContextD'' \subt \ContextD'$ and $T'' \subt T(\Actions)$.
    From \Cref{lem:buffer-type-inv} we derive $\wtp[m' +
    m'']{\Buffer_x[R']}{\ContextE, \ContextD'', x : \Actions.T''}$.
    Note that using \Cref{prop:refine-subt} we have $\Actions.T'' \subt
    \Actions.T(\Actions) \subt T$.
    We conclude by taking $m \eqdef m' + m''$ and $\ContextD \eqdef \ContextE,
    \ContextD'', x : \Actions.T''$.

    \proofcase{\rulename{r-str}}
    Then $P \pcong R \red Q$.
    From \Cref{lem:pcong-preservation} we deduce $\wtp[n']{R}\ContextE$ for some
    $n' \leq n$ and $\ContextE \subt \Context$.
    Using the induction hypothesis we conclude $\wtp[m]{Q}\ContextD$ for some
    $\ContextD \subt \ContextE \subt \Context$.
\end{proof}
\subsection{Proof of Theorem~\ref{thm:df}}
\label{sec:proofs-deadlock-freedom}

\begin{table}
    \caption{\label{tab:obs}Observability predicate for processes.}
    \begin{mathpar}
        \inferrule{~}{
            \thread\Done
        }
        \and
        \inferrule{~}{
            \thread{\Link\x\y}
        }
        \and
        \inferrule{~}{
            \thread{\Close\x}
        }
        \and
        \inferrule{~}{
            \thread{\Wait\x.P}
        }
        \and
        \inferrule{~}{
            \thread{\Join\x\y.P}
        }
        \and
        \inferrule{~}{
            \thread{\Case\x{\Tag_i : P_i}_{i\in I}}
        }
        \and
        \inferrule{
            \thread{P}
        }{
            \thread{\Buffer_x[P]}
        }
    \end{mathpar}
\end{table}

As we have anticipated in \Cref{sec:soundness} the proof of the deadlock freedom
result goes through a more general productivity result stating that every
well-typed process is either reducible or is structurally precongruent to a
thread. In this context, the notion of ``thread'' is defined by the predicate
$\thread\cdot$ inductively defined in \Cref{tab:obs}.

Next we define the \emph{depth} of a process as the size of its unguarded
portion, which is made of cuts and of process invocations. Because of the
guardedness assumption made in \Cref{sec:language}, this quantity is always well
defined.

\begin{definition}[depth of a process]
    \label{def:depth}
    The \emph{depth} of a process $P$ is defined as the quantity
    \[
        \depth{P} =
        \begin{cases}
            1 + \depth{P_1} + \depth{P_2} & \text{if $P = \Cut\x{P_1}{P_2}$} \\
            1 + \depth{Q} & \text{if $P = \Call{A}{\seqof\x}$ and $\Let{A}{\seqof\x} = Q$} \\
            0 & \text{otherwise}
        \end{cases}
    \]
\end{definition}

\begin{lemma}
    \label{lem:df}
    If $\wtp[n]{P}\Context$, then there exists $Q$ such that either $P \pcong Q$
    and $\thread{Q}$ or $P \red Q$.
\end{lemma}
\begin{proof}
    By induction on the lexicographically ordered pair $(n,\depth{P})$ and by
    cases on the shape of $P$.

    \proofcase{$P = \Done$ or $P = \Close\x$ or $P = \Link\x\y$ or $P =
    \Wait\x.Q$ or $P = \Case\x{\Tag_i : P_i}_{i\in I}$ or $P = \Join\x\y.Q$}
    Then $\thread{P}$ holds and we conclude by reflexivity of $\pcong$.

    \proofcase{$P = \Call{A}{\seqof\x}$ where $\Let{A}{\seqof\x} = P'$}
    Then $\depth{P'} < \depth{P}$. By induction hypothesis there exists $Q$ such
    that either $P' \pcong Q$ and $\thread{Q}$ or $P' \red Q$. In the first
    sub-case we conclude since $P \pcong P' \pcong Q$. In the second sub-case we
    conclude $P \pcong P' \red Q$ by \rulename{r-str}.

    \proofcase{$P = P_1 \choice P_2$}
    Then we conclude by taking $Q \eqdef P_1$ and observing that $P \red Q$ by
    \rulename{r-choice}.

    \proofcase{$P = \Select\x\Tag.P'$}
    From \SelectRule we deduce that
    \begin{itemize}
        \item $\Context = \ContextD, x : S$
        \item $S = \Plus{\ann[m_i]{\Tag_i} : S_i}_{i\in I}$
        \item $\Tag = \Tag_k$ for some $k\in I$
        \item $n = 1 + m_k + n'$
        \item $\wtp[n']{P'}{\ContextD, x : S_k}$
    \end{itemize}

    Note that $n' < n$. Using the induction hypothesis we deduce that there
    exists $Q'$ such that either $P' \pcong Q'$ and $\thread{Q'}$ or $P' \red
    Q'$. Let $Q \eqdef \Select\x\Tag.Q'$. In the first sub-case we observe that
    $\thread{Q}$ and $P \pcong Q$. In the second sub-case we observe that $P
    \red Q$ by \rulename{r-buffer}.

    \proofcase{$P = \Fork\x\y{P''}{P'}$} Analogous to the previous case.

    \proofcase{$P = \Cut\x{P_1}{P_2}$}
    From \CutRule we deduce
    \begin{itemize}
        \item $n = n_1 + n_2$ for some $n_1$ and $n_2$
        \item $\Context = \Context_1, \Context_2$ for some $\Context_1$ and $\Context_2$
        \item $\wtp[n_i]{P_i}{\Context_i, x : S_i}$ for $i=1,2$
        \item $\cc{S}{T}$
    \end{itemize}

    We have $n_i \leq n$ and $\depth{P_i} < \depth{P}$. Using the induction
    hypothesis twice we deduce that, for every $i=1,2$, there exists $Q_i$ such
    that either $\thread{Q_i}$ and $P_i \pcong Q_i$ or $P_i \red Q_i$.
    If $P_1 \red Q_1$, then we conclude by taking $Q \eqdef \Cut\x{Q_1}{P_2}$
    and observing that $P \red Q$ by \rulename{r-cut}.
    If $P_2 \red Q_2$, then we conclude by taking $Q \eqdef \Cut\x{Q_2}{P_1}$
    and observing that $P \pcong \Cut\x{P_2}{P_1} \red Q$ by \rulename{s-comm}
    followed by \rulename{r-cut}.
    
    The most interesting case is when $\thread{Q_i}$ and $P_i \pcong Q_i$ for
    every $i=1,2$. In this case we have $P = \Cut\x{P_1}{P_2} \pcong
    \Cut\x{Q_1}{Q_2}$ where $Q_i =
    \Buffer_{x_{i1}}[\cdots\Buffer_{x_{in}}[R_i]]$, none of the
    $\Buffer_{x_{ij}}$ is $\Hole$ and $R_i$ is one of the processes in the base
    cases of the $\thread\cdot$ predicate. For brevity, we classify those
    processes as \emph{links} ($\Link\x\y$), \emph{closures} ($\Close\x$) and
    \emph{inputs} (all the remaining forms).

    As long as there exist $i$ and $j$ such that $x_{ij} \ne x$, then we can use
    \rulename{s-pull-3} and \rulename{s-pull-4} to pull the buffer
    $\Buffer_{x_{ij}}$ to the top level of $Q_i$ and then \rulename{s-pull-1} or
    \rulename{s-pull-2} to pull it outside of the cut.
    If in this procedure we have used either \rulename{s-pull-1} or
    \rulename{s-pull-3}, then we have obtained the desired thread $Q$ such that
    $P \pcong Q$.
    If in this procedure we have only used \rulename{s-pull-2} and
    \rulename{s-pull-4}, then we are left with a process of the form
    $\Buffer_{x_1}[\cdots\Buffer_{x_m}[\Cut\x{\Buffer_x^1[R_1]}{\Buffer_x^2[R_2]}]]$
    where none of the $x_i$ is $x$.
    Also, it must be the case that $x$ occurs free in $R_i$ for
    $i=1,2$.\LP{Questo sarebbe da provare}

    If either $R_1$ or $R_2$ is an input on a channel $y \ne x$, then we can use
    one of the rules \rulename{s-wait}, \rulename{s-case} or \rulename{s-join}
    to obtain the desired thread $Q$.
    If none of $R_1$ and $R_2$ is an input on a channel $y \ne x$, then we argue
    that the cut $\Cut\x{\Buffer_x^1[R_1]}{\Buffer_x^2[R_2]}$ in this process
    does reduce, which is enough to obtain the desired $Q$ such that $P \red Q$.
    Recall that this cut is well typed, hence it introduces a name $x$ which is
    typed by some $S_i$ in $\Buffer_x^i[R_i]$ such that $\cc{S_1}{S_2}$
    holds.
    We discuss all the possibilities below, omitting (some) symmetric cases.
    \begin{itemize}
    \item (either $R_1$ or $R_2$ is $\Done$) This case is impossible since $x$
        does not occur free in $\Done$.
    \item (both $R_1$ and $R_2$ are links) If $\Buffer_x^i = \Hole$ for some
        $i=1,2$, then we can use \rulename{r-link} possibly preceded by
        \rulename{s-comm} and/or \rulename{s-link} to reduce the cut. If
        $\Buffer_x^i \ne \Hole$ for every $i=1,2$ we can use \rulename{s-pull-0}
        to pull one of the (non-empty) buffers out of the cut and then proceed
        as just discussed when (at least) one buffer is empty.
    \item (both $R_1$ and $R_2$ are inputs) From $\cc{S_1}{S_2}$ we deduce that
        at least one among $S_1$ and $S_2$, say $S_1$, must start with an
        output. But then $\Buffer_x^1$ cannot be empty and must actually have
        the right shape to synchronize with $R_2$.
    \item (both $R_1$ and $R_2$ are $\Close\x$) This case is ruled out by
        $\cc{S_1}{S_2}$.
    \item ($R_1$ is a link and $R_2$ is an input) If $\Buffer_x^1 = \Hole$ then
        the cut reduces by \rulename{r-link}. If $\Buffer_x^1 \ne \Hole$, then
        the buffer must start with an output that synchronizes with $R_2$.
    \item ($R_1$ is a link and $R_2$ is $\Close\x$) Then $\Buffer_x^1 = \Hole$
        and the cut reduces by \rulename{r-link}.
    \item ($R_1$ is an input and $R_2$ is $\Close\x$) If $\Buffer_x^2 = \Hole$
        then from $\cc{S_1}{S_2}$ we deduce that $\Buffer_x^1 = \Hole$ as well,
        hence the cut reduces by \rulename{r-close}. If $\Buffer_x^2 \ne \Hole$,
        then $\Buffer_x^2$ must start with an output that synchronizes with
        $R_1$.
        \qedhere
    \end{itemize}
\end{proof}

\thmdf*
\begin{proof}
    Consider a reduction $P \wred Q \nred$.
    From \Cref{thm:sr} we deduce $\wtp[m]{Q}\EmptyContext$.
    From \Cref{lem:df} we deduce that $Q \pcong R$ and $\thread{R}$. The only
    thread that is well typed in the empty context is $\Done$, hence $Q \pcong
    \Done$.
\end{proof}
\subsection{Proof of Theorem~\ref{thm:wt}}
\label{sec:proofs-fair-termination}

\thmwt*
\begin{proof}
    We have to show that $P$ has a finite run.
    By inspecting the proof of \Cref{thm:sr} we see that, in all the base cases
    except when $P = P_1 \choice P_2$, the measure of the process by at least
    one unit. This property holds also when $P = P_1 \choice P_2$, provided that
    we choose the reduction $P \to P_k$ where $k$ is the index of the $P_k$ that
    determines the measure of the choice in the conclusion of \ChoiceRule (\cf
    \Cref{tab:typing-rules}).
    Therefore we can build a finite run for $P$ by induction on $n$ as follows.
    If $P \nred$ we conclude by taking the sequence $(P)$.
    If $P \red$, then by the above reasoning we can find $Q$ such that $P \red
    Q$ and $\wtp[m]{Q}\ContextD$ where $m < n$ and $\ContextD \subt \Context$.
    Using the induction hypothesis we deduce that $Q$ has a finite run
    $(Q_0,Q_1,\dots)$, but then $(P,Q_0,Q_1,\dots)$ is a finite run of $P$.
\end{proof}

\thmt*
\begin{proof}
    This is a trivial special case of \Cref{thm:wt} where the case in which a
    non-deterministic choice is reduce cannot happen.
\end{proof}

\section{Supplement to Section~\ref{sec:subtyping-inclusions}}

We define the following auxiliary relation $\auxsubt$.
\begin{definition}
    \label{def:auxsubt}
    We write $\auxsubt$ for the largest relation such that $S \auxsubt T$ implies:
    \begin{enumerate}
    \item\label{aux:pol} either $\positive(S)$ or $\negative(T)$;
    \item\label{aux:inp} if $T \lmust{\inp\MessageType} T'$ then $S
        \lred{\inp\MessageType} S'$ and $S' \auxsubt T'$;
    \item\label{aux:out} if $S \lmust{\out\MessageType} S'$ then $T
        \lred{\out\MessageType} T'$ and $S' \auxsubt T'$.
    \end{enumerate}
\end{definition}

This relation is useful in the proof of \cref{thm:inclusions} because it
obviously includes both the subtyping relations $\mysubt$ and $\blzsubt$; hence
we can prove ${\mysubt} \subseteq {\subt}$ and ${\blzsubt} \subseteq {\subt}$
simply by showing that $\auxsubt$ is included in $\subt$. We need some
preliminary lemmata before proving this inclusion. The first of these lemmata
uses the transition relation $S \cored\Action T$ which corresponds to the
coinductive part of the transition relation $S \lred\Action T$, \ie it considers
only the singly-lined rules in \cref{tab:lts} interpreted coinductively.

\begin{lemma}
  \label{lem:coind}
  If $S$ is fairly terminating and $S \cored\Action T$, then $S \lred\Action T$.
\end{lemma}
\begin{proof}
  It suffices to show that $S \cored\Action T$ implies $S \ired\Action T$. This
  follows by an easy induction on the minimum depth of a leaf in $S$, knowing
  that such depth is finite from the assumption that $S$ is fairly terminating,
  and by cases on the shape of $S$.
\end{proof}

\begin{lemma}
  \label{lem:aux-inp}
  If $S \auxsubt T$ and $T \lred{\inp\MessageType}$, then $S
  \lred{\inp\MessageType}$.
\end{lemma}
\begin{proof}
  In this proof we treat $\lred{\inp\MessageType}$ as a unary
  predicate.
  By \cref{lem:coind} it is enough to show that $S \lred{\inp\MessageType}$ is
  \emph{coinductively} derivable using the singly-lines rules in \cref{tab:lts}.
  To this aim it suffices to show that
  \[
    \rrel \eqdef \set{ S \mid S \auxsubt T \wedge T \lred{\inp\MessageType}}
  \]
  is backward closed with respect to the singly-lined rules in \cref{tab:lts}.
  Let $S \in \rrel$. Then $S \auxsubt T$ and $T \lred{\inp\MessageType}$ for
  some $T$ and $\MessageType$. We reason by cases on the shape of $T$ to show
  that $S \lred{\inp\MessageType}$ is the conclusion of one of the rules in
  \cref{tab:lts} whose premises are in $\rrel$.

  \proofcase{$T = \Bot$}
  Then $\MessageType = \unit$ and $T \lmust{\inp\unit}$. From
  \cref{def:auxsubt}\eqref{aux:inp} we deduce $S \lred{\inp\unit}$. We have that
  $\positive(S)$ does not hold, otherwise $S \lmust{\out\MessageTypeT} S'$ which
  implies, by \cref{def:auxsubt}\eqref{aux:out},  $T \lred{\out\MessageTypeT}
  T'$ that does not hold because $T = \Bot$. 
  %
  Hence also $S=\Bot$.
  In this case we have that $S \lred{\inp\MessageType}$ is the conclusion of the
  axiom \BotT.


  \proofcase{$T = \With{\Tag_i : T_i}_{i\in I}$}
  Then $\MessageType = \Tag_k$ for some $k\in I$ and $T \xlmust{\inp\Tag_k}$.
  From \cref{def:auxsubt}\eqref{aux:inp} we deduce $S \xlred{\inp\Tag_k}$. Note
  that this rules out the possibility that $S = \Bot$ or $S = \One$.
  We distinguish two sub-cases:
  \begin{itemize}
  \item If $\negative(S)$, then $S = \With{\Tag_j : S_j}_{j\in J}$ and $J
    \supseteq I$.
    As $k\in I$, we have $k\in J$ and $\MessageType = \Tag_k$ guarantees that in
    this case $S \lred{\inp\MessageType}$ is the conclusion of the axiom \WithT.
  \item If $\positive(S)$, then $S = \Plus{\Tag_j : S_j}_{j\in J}$.
    From \cref{def:auxsubt}\eqref{aux:out} we deduce that there exists a family
    $\set{T_j}_{j\in J}$ of session types such that $T \lred{\out\Tag_j} T_j$
    and $S_j \auxsubt T_j$ for every $j\in J$.
    From $T \xlmust{\inp\Tag_k}$ and $T \lred{\out\Tag_j} T_j$ we deduce $T_j
    \xlmust{\inp\Tag_k}$ for every $j\in J$.
    Therefore $S_j \in \rrel$ by definition of $\rrel$.
    In this case we have that $S \lred{\inp\MessageType}$ is the conclusion of
    \APlusT.  
  \end{itemize}

  \proofcase{$T = \Plus{\Tag_i : T_i}_{i\in I}$}
  Then $T_i \lred{\inp\MessageType}$ for every $i\in I$.
  From \cref{def:auxsubt}\eqref{aux:pol} we deduce $\positive(S)$, hence $S =
  \Plus{\Tag_i : S_i}_{i\in J}$ and, from \cref{def:auxsubt}\eqref{aux:out}, we
  also have  $J \subseteq I$ and $S_i \auxsubt T_i$ for every $i\in J$.
  Now we have $S_i \in \rrel$ for every $i\in J$ by definition of $\rrel$. In
  this case we have that $S \lred{\inp\MessageType}$ is the conclusion of
  \APlusT.
\end{proof}

\begin{lemma}
  \label{lem:aux-inp-subt}
  If $S \auxsubt T$ and $T \lred{\inp\MessageType} T'$, then $S
  \lred{\inp\MessageType} S'$ and $S' \auxsubt T'$.
\end{lemma}
\begin{proof}
  It suffices to show that
  \[
    \srel \eqdef {\auxsubt} \cup \set{ (S', T') \mid S \auxsubt T \wedge S \lred{\inp\MessageType} S' \wedge T \lred{\inp\MessageType} T'}
  \]
  satisfies the clauses of \cref{def:auxsubt}. Clearly any pair of types in $\auxsubt$
  satisfies those clauses, so consider $(S',T') \in \srel$ such that $S \auxsubt
  T$ and $S \lred{\inp\MessageType} S'$ and $T \lred{\inp\MessageType} T'$.
  We reason by cases on the shape of $T$.

  \proofcase{$T = \Bot$}
  Then $\MessageType = \unit$ and $T \lmust{\inp\unit} T'$. From
  \cref{def:auxsubt}\eqref{aux:inp} we deduce $S \lred{\inp\unit} S'$ with $S'
  \auxsubt T'$. In this case $(S',T')$  satisfies the clauses in
  \cref{def:auxsubt} because $S' \auxsubt T'$.

  \proofcase{$T = \With{\Tag_i : T_i}_{i\in I}$}
  Then $\MessageType = \Tag_k$ for some $k\in I$ and $T \xlmust{\inp\Tag_k} T'$.
  From \cref{def:auxsubt}\eqref{aux:inp} we deduce $S \xlred{\inp\Tag_k} S'$
  with $S' \auxsubt T'$. In this case $(S',T')$  satisfies the clauses in
  \cref{def:auxsubt} because $S' \auxsubt T'$.

  \proofcase{$T = \Plus{\Tag_i : T_i}_{i\in I}$}
  Then there exists a family $\set{T_i'}_{i\in I}$ of session types such that
  $T_i \lred{\inp\MessageType}  T_i'$ for every $i\in I$ and $T' = \Plus{\Tag_i
  : T_i'}_{i\in I}$.
  From \cref{def:auxsubt}\eqref{aux:pol} we deduce $\positive(S)$, hence $S =
  \Plus{\Tag_i : S_i}_{i\in J}$.
  From \cref{def:auxsubt}\eqref{aux:out} we deduce $J \subseteq I$ and $S_i
  \auxsubt T_i$ for every $i\in J$.
  From \cref{lem:aux-inp} we deduce that there exists a family $\set{S_i'}_{i\in
  J}$ of session types such that $S_i \lred{\inp\MessageType} S_i'$ for every
  $i\in J$ and $S' \eqdef \Plus{\Tag_i : S_i'}_{i\in J}$.

  Now consider the clauses of \cref{def:auxsubt} in order:
  \begin{enumerate}
  \item This clause holds because we have $\positive(S')$.
  \item This clause holds vacuously becase $T' \nlmust{\inp\MessageTypeT}$ for
    every $\MessageTypeT$.
  \item Suppose $S' \lmust{\out\MessageTypeT} S''$. Then $\MessageTypeT =
    \Tag_i$ and $S'' = S_i'$ for some $i \in J$.
    Moreover we have $T \lmust{\out\MessageTypeT} T_i'$ and $S_i' \auxsubt
    T_i'$.
    We conclude by observing that $(S_i', T_i') \in \srel$ by definition of
    $\srel$.
    \qedhere
  \end{enumerate}
\end{proof}

\begin{lemma}
  \label{lem:aux-out-subt}
  If $S \auxsubt T$ and $S \lred{\out\MessageType} S'$, then $T
  \lred{\out\MessageType} T'$ and $S' \auxsubt T'$.
\end{lemma}
\begin{proof}
  Since the LTS is symmetric and we are considering first-order session types
  only, from $R \lred{\out\MessageTypeT} R'$ we deduce $\dual{R}
  \lred{\inp\MessageTypeT} \dual{R'}$.
  From this symmetry, we deduce that $\auxsubt$ is closed under duality;
  hence, from $S \auxsubt T$ we deduce $\dual{T} \auxsubt \dual{S}$.
  Consider now $S \lred{\out\MessageType} S'$; by the above symmetry we have
  $\dual{S} \lred{\inp\MessageType} \dual{S'}$. 
  Using \cref{lem:aux-inp-subt} we deduce $\dual{T} \lred{\inp\MessageType}
  \dual{T'}$ and $\dual{T'} \auxsubt \dual{S'}$, that is $S' \auxsubt T'$.
\end{proof}

\begin{theorem}
  \label{thm:auxsubt}
  If $S \auxsubt T$, then $S \subt T$.
\end{theorem}
\begin{proof}
  It suffices to show that the relation $\auxsubt$ satisfies the clauses of \cref{def:cas}.
  Let $(S,T) \in\ \auxsubt$.
  From $S \auxsubt T$ and \cref{def:auxsubt}\eqref{aux:pol} we deduce either
  $\positive(S)$ or $\negative(T)$, hence~\rcas{pol} is satisfied.
  Suppose $T \lred{\inp\MessageType} T'$.
  From \cref{lem:aux-inp-subt} we deduce $S \lred{\inp\MessageType} S'$ and $S'
  \auxsubt T'$, hence~\rcas{inp} is satisfied because $(S',T') \in\ \auxsubt$.
  Suppose $S \lred{\out\MessageType} S'$.
  From \cref{lem:aux-out-subt} we deduce $T \lred{\out\MessageType} T'$ and $S'
  \auxsubt T'$, hence~\rcas{out} is satisfied because $(S',T') \in\ \auxsubt$.
\end{proof}

We are finally ready to prove \cref{thm:inclusions}
\thminclusions*
\begin{proof}
  The inclusion ${\ghsubt} \subseteq {\mysubt}$ follows from the definitions of
  $\ghsubt$ and $\mysubt$ because ${\lmust{\Action}} \subseteq
  {\lind{\Action}}$.
  Consider now the auxiliary subtyping relation $\auxsubt$.
  We have that ${\mysubt} \subseteq {\auxsubt}$ follows from the definitions of
  $\mysubt$ and $\auxsubt$ because ${\lind\Action} \subseteq {\lred\Action}$.
  Also ${\blzsubt} \subseteq {\auxsubt}$ because \cref{def:auxsubt} is the same
  as \cref{def:blzsubt} except for the additional constraint in the clause
  \cref{def:blzsubt}\eqref{blz:out}. 
  By \cref{thm:auxsubt}, these two inclusions imply ${\mysubt} \subseteq {\subt}$
  and ${\blzsubt} \subseteq {\subt}$.
  
  Strictness of the inclusions has already been shown by the examples in
  \cref{sec:subtyping-inclusions}.
\end{proof}

\end{document}